\documentclass[english]{article}
\usepackage[T1]{fontenc}
\usepackage[latin9]{inputenc}
\usepackage{geometry}
\usepackage{color}
\usepackage{babel}
\usepackage{array}
\usepackage{units}
\usepackage{multirow}
\usepackage{amsmath}
\usepackage{amsthm}
\usepackage{amssymb}
\usepackage{graphicx}
\usepackage{setspace}
\usepackage[authoryear]{natbib}
\usepackage[unicode=true,pdfusetitle,
 bookmarks=true,bookmarksnumbered=false,bookmarksopen=false,
 breaklinks=false,pdfborder={0 0 0},pdfborderstyle={},backref=false,colorlinks=true]
 {hyperref}
\hypersetup{
 citecolor=blue, linkcolor=blue, urlcolor=blue}

\makeatletter

\providecommand{\tabularnewline}{\\}

\theoremstyle{definition}
\newtheorem{defn}{\protect\definitionname}
\theoremstyle{remark}
\newtheorem{rem}{\protect\remarkname}
\theoremstyle{definition}
 \newtheorem{example}{\protect\examplename}
\theoremstyle{plain}
\newtheorem{prop}{\protect\propositionname}
\theoremstyle{plain}
\newtheorem{thm}{\protect\theoremname}
\theoremstyle{plain}
\newtheorem{lem}{\protect\lemmaname}
\theoremstyle{plain}
\newtheorem{fact}{\protect\factname}

 \usepackage{lmodern}
 \usepackage[T1]{fontenc}

\makeatother

\providecommand{\definitionname}{Definition}
\providecommand{\examplename}{Example}
\providecommand{\factname}{Fact}
\providecommand{\lemmaname}{Lemma}
\providecommand{\propositionname}{Proposition}
\providecommand{\remarkname}{Remark}
\providecommand{\theoremname}{Theorem}

\begin{document}
\title{Observations on Cooperation }
\author{Yuval Heller\thanks{Affiliation: Department of Economics, Bar Ilan University, Israel.
 E-mail: yuval.heller@biu.ac.il.} ~and Erik Mohlin\thanks{Affiliation: Department of Economics, Lund University, Sweden.  E-mail:
erik.mohlin@nek.lu.se. }\thanks{A previous version of this paper was circulated under the title ``Stable
observable behavior.'' We have benefited greatly from discussions
with Vince Crawford, Eddie Dekel, Christoph Kuzmics, Ariel Rubinstein,
Larry Samuelson, Bill Sandholm, Rann Smorodinsky, Rani Spiegler, Balázs
Szentes, Satoru Takahashi, Jörgen Weibull, and Peyton Young. We would
like to express our deep gratitude to seminar/workshop participants
at the University of Amsterdam (CREED), University of Bamberg, Bar
Ilan University, Bielefeld University, University of Cambridge, Hebrew
University of Jerusalem, Helsinki Center for Economic Research, Interdisciplinary
Center Herzliya, Israel Institute of Technology, Lund University,
University of Oxford, University of Pittsburgh, Stockholm School of
Economics, Tel Aviv University, NBER Theory Workshop at Wisconsin-Madison,
KAEA session at the ASSA 2015, the Biological Basis of Preference
conference at Simon Fraser University, and the 6th workshop on stochastic
methods in game theory at Erice, for many useful comments. Danial
Ali Akbari provided excellent research assistance. Yuval Heller is
grateful to the European Research Council for its financial support
(Starting Grant \#677057). Erik Mohlin is grateful to Handelsbankens
forskningsstiftelser (grant \#P2016-0079:1), the Swedish Research
Council (grant \#2015-01751), and the Knut and Alice Wallenberg Foundation
(Wallenberg Academy Fellowship \#2016-0156) for their financial support.
Last but not least, we thank Renana Heller for suggesting the title.}}

\maketitle
Final pre-print of a manuscript published in the \emph{Review of Economic
Studies}, 85(4), 2018, pp. 2253-2282.
\begin{abstract}
We study environments in which agents are randomly matched to play
a Prisoner's Dilemma, and each player observes a few of the partner's
past actions against previous opponents. We depart from the existing
related literature by allowing a small fraction of the population
to be commitment types. The presence of committed agents destabilizes
previously proposed mechanisms for sustaining cooperation. We present
a novel intuitive combination of strategies that sustains cooperation
in various environments. Moreover, we show that under an additional
assumption of stationarity, this combination of strategies is essentially
the \emph{unique }mechanism to support full cooperation, and it is
robust to various perturbations. Finally, we extend the results to
a setup in which agents also observe actions played by past opponents
against the current partner, and we characterize which observation
structure is optimal for sustaining cooperation. 

\textbf{JEL Classification:} C72, C73, D83.\textbf{ Keywords}: Community
enforcement; indirect reciprocity; random matching; Prisoner's Dilemma;
image scoring.
\end{abstract}

\section{Introduction}

Consider the following example of a simple yet fundamental economic
interaction. Alice has to trade with another agent, Bob, whom she
does not know. Both sides have opportunities to cheat, to their own
benefit, at the expense of the other. Alice is unlikely to interact
with Bob again, and thus her ability to retaliate, in case Bob acts
opportunistically, is restricted. The effectiveness of external enforcement
is also limited, e.g., due to incompleteness of contracts, non-verifiability
of information, and court costs. Thus cooperation may be impossible
to achieve. Alice searches for information about Bob's past behavior,
and she obtains anecdotal evidence about Bob's actions in a couple
of past interactions. Alice considers this information when she decides
how to act. Alice also takes into account that her behavior toward
Bob in the current interaction may be observed by her future partners.
Historically, the above-described situation was a challenge to the
establishment of long-distance trade (\citealp{milgrom1990role,greif1993contract}),
and it continues to play an important role in the modern economy,
in both offline (\citealp{bernstein1992opting,dixit2003modes}) and
online interactions (\citealp{resnick2002trust,josang2007survey}).

Several papers have studied the question of how cooperation can be
supported by means of community enforcement. Most of these papers
assume that all agents in the community are rational and, in equilibrium,
best reply to what everyone else is doing. As argued by \citet[p. 578]{ellison1994cooperation},
this assumption may be fairly implausible in large populations. It
seems quite likely that, in a large population, there will be at least
some agents who fail to best respond to what the others are doing,
either because they are boundedly rational, have idiosyncratic preferences,
or because their expectations about other agents' behavior are incorrect.
Motivated by this argument, we allow a few agents in the population
to be committed to behaviors that do not necessarily maximize their
payoffs. It turns out that this seemingly small modification completely
destabilizes existing mechanisms for sustaining cooperation when agents
are randomly matched with new partners in each period. Specifically,
both the contagious equilibria (\citealp{kandori1992social,ellison1994cooperation})
and the ``belief-free'' equilibria (\citealp{takahashi2010community,deb2012cooperation})
fail in the presence of a small fraction of committed agents.\footnote{In contagious equilibria players start by cooperating. If one player
defects at stage $t$, her partner defects at stage $t+1$, infecting
another player who defects at stage $t+2$, and so on. In belief-free
equilibria players are always indifferent between their actions, but
they choose different mixed actions depending on the signal they obtain
about the partner. We discuss the non-robustness of these classes
of equilibria at the end of Section \ref{subsec:Stability-of-Defection}. } 

\emph{Our key results are as follows}. First, we show that always
defecting is the unique perfect equilibrium, regardless of the number
of observed actions, provided that the bonus of defection in the underlying
Prisoner\textquoteright s Dilemma is larger when the partner cooperates
than when the partner defects. Second, in the opposite case, when
the bonus of defection is larger when the partner defects than when
the partner cooperates, we present a novel and essentially unique
combination of strategies that sustains cooperation: all agents cooperate
when they observe no defections and defect when they observe at least
two defections.\footnote{As discussed later, our uniqueness results also rely on an additional
assumption that agents are restricted to choose stationary strategies,
which depend only on the signal about the partner. As shown in Appendix
\ref{sec:Conventional-Repeated-Game}, all other results hold also
in a standard setup without the restriction to stationary strategies. } Some of the agents also defect when observing a single defection.
Importantly, this cooperative behavior is robust to various perturbations,
and it appears consistent with experimental data. Third, we extend
the model to environments in which an agent also obtains information
about the behavior of past opponents against the current partner.
We show that in this setup cooperation can be sustained if and only
if the bonus of defection of a player is less than half the loss she
induces a cooperative partner to suffer. Finally, we characterize
an observation structure that allows cooperation to be supported as
a perfect equilibrium outcome in \emph{all} Prisoner's Dilemma games.
In all observation structures we use the same essentially unique construction
to sustain cooperation.

\paragraph{Overview of the Model}

Agents in an infinite population are randomly matched into pairs to
play the Prisoner's Dilemma game, in which each player decides simultaneously
whether to cooperate or defect (see the payoff matrix in Table \ref{tab:Prisoner-Dilemma-1-1}).
If both players cooperate they obtain a payoff of one, if both defect
they obtain a payoff of zero, and if one of the players defects, the
defector gets $1+g$, while the cooperator gets $-l$, where $g,l>0$
and $g<l+1$. (The latter inequality implies that mutual cooperation
is the efficient outcome that maximizes the sum of payoffs.)
\begin{table}[h]
\caption{\label{tab:Prisoner-Dilemma-1-1}Matrix Payoffs of Prisoner's Dilemma
Games }

\centering{}%
\begin{tabular}{|c|c|c|}
\hline 
 & \emph{c} & \emph{d}\tabularnewline
\hline 
\emph{~~c~~} & \emph{\Large{}$_{\underset{\,}{1}}\,\,\,\,\,\,^{\overset{\,}{1}}$} & \emph{\Large{}$_{\underset{\,}{-l}}\,\,^{\overset{\,}{1+g}}$}\tabularnewline
\hline 
\emph{d} & \emph{\Large{}$_{\underset{\,}{1+g}}\,\,^{\overset{\,}{-l}}$} & \emph{\Large{}$_{0}\,\,\,\,\,\,^{0}$}\tabularnewline
\hline 
\multicolumn{3}{c}{$g,l>0$ ,~ $g<l+1$}\tabularnewline
\end{tabular}
\end{table}

Before playing the game, each agent privately draws a random sample
of $k$ actions that have been played by her partner against other
opponents in the past. The assumption that a small random sample is
taken from the entire history of the partner is intended to reflect
a setting in which the memory of past interactions is long and accurate
but dispersed. This means that the information that reaches an agent
about her partner (through gossip) arrives in a non-deterministic
fashion and may stem from any point in the past.

We require each agent to follow a\emph{ stationary} \emph{strategy},
i.e., a mapping that assigns a mixed action to each signal that the
agent may observe about the current partner. (That is, the action
is not allowed to depend on calendar time or on the agent's own history.)
A \emph{steady state} of the environment is a pair consisting of:
(1) a distribution of strategies with a finite support that describes
the fractions of the population following the different strategies,
and (2) a \emph{signal profile} that describes the distribution of
signals that is observed when an agent is matched with a partner playing
any of the strategies present in the population. The signal profile
is required to be \emph{consistent} with the distribution of strategies
in the sense that a population of agents who follow the distribution
of strategies and observe signals about the partners sampled from
the signal profile will behave in a way that induces the same signal
profile.\footnote{The reason why the consistent signal profile is required to be part
of the description of a steady state, rather than being uniquely determined
by the distribution of strategies, is that our environment, unlike
a standard repeated game, lacks a global starting time that determines
the initial conditions. An example of a strategy that has multiple
consistent signal profiles is as follows. The parameter $k$ is equal
to three, and everyone plays the most frequently observed action in
the sample of the three observed actions. There are three behaviors
that are consistent with this population: one in which everyone cooperates,
one in which everyone defects, and one in which everyone plays (on
average) uniformly. }

Our restriction to stationary strategies and our focus on consistent
steady states allow us to relax the standard assumption that there
is an initial time zero at which an entire community starts to interact.
In various real-life situations, the interactions within the community
have been going on from time immemorial. Consequently the participants
may have only a vague idea of the starting point. Arguably, agents
might therefore be unable to condition their behavior on everything
that has happened since the beginning of the interactions. 

We perturb the environment by introducing $\epsilon$ \emph{committed
agents} who each follow one strategy from an arbitrary finite set
of \emph{commitment strategies}. We assume that at least one of the
commitment strategies is totally mixed, which implies that all signals
(i.e., all sequences of \emph{k} actions) are observed with positive
probability. A \emph{steady state} in a perturbed environment describes
a population in which $1-\epsilon$ of the agents are \emph{normal};\emph{
}i.e.,\emph{ }they play strategies that maximize their long-run payoffs,
while $\epsilon$ of the agents follow commitment strategies. 

We adapt the notions of Nash equilibrium and perfect equilibrium (\citealp{selten1975reexamination})
to our setup. A steady state is a \emph{Nash equilibrium }if no normal
agent can gain in the long run by deviating to a different strategy
(the agents are assumed to be arbitrarily patient). The deviator's
payoff is calculated in the new steady state that emerges following
her deviation. A steady state is a \emph{perfect equilibrium }if it
is the limit of a sequence of Nash equilibria in a converging sequence
of perturbed environments.\footnote{In Appendix \ref{sec:Evolutionary-Stability} we show that all the
equilibria presented in this paper satisfy two additional refinements:
(1) evolutionary stability (\citealp{Maynard-Smith1974theory}) \textendash{}
any small group of agents who jointly deviate are outperformed, and
(2) \emph{robustness} \textendash{} no small perturbation in the distribution
of observed signals can move the population's behavior away from a
situation in which everyone plays the equilibrium outcome. In addition,
most of these equilibria also satisfy the refinement of strict perfection
(\citealp{okada1981stability}) \textendash{} the equilibrium remains
stable with respect to all commitment strategies.}

\paragraph{Summary of Results}

We start with a simple result (Prop. \ref{pro:defection-is-evol-stable})
that shows that defection is a perfect equilibrium outcome for any
number of observed actions. 

We say that a Prisoner's Dilemma game is \emph{offensive }if\emph{
}there is a stronger incentive to defect against a cooperator than
against a defector (i.e., $g>l$); in a \emph{defensive} Prisoner\textquoteright s
Dilemma\emph{ }the opposite holds (i.e., $g<l$). Our first main result
(Theorem \ref{thm:only-defection-is-stable}) shows that always defecting
is the unique perfect equilibrium in any offensive Prisoner's Dilemma
game (i.e., $g>l$) for any number of observed actions. The result
assumes a mild \emph{regularity} condition on the set of commitment
strategies (Def. \ref{def:regularity-1}), namely, that this set is
rich enough that, in any steady state of the perturbed environment,
at least one of the commitment strategies induces agents to defect
with a different probability than that of some of the normal agents.
The intuition is as follows. The mild assumption that not all agents
defect with exactly the same probability implies that the signal that
Alice observes about her partner Bob is not completely uninformative.
In particular, the more often Alice observes Bob to defect, the more
likely Bob will defect against Alice. In offensive games, it is better
to defect against partners who are likely to cooperate than to defect
against partners who are likely to defect. This implies that a deviator
who always defects is more likely to induce normal partners to cooperate.
Consequently, such a deviator will outperform any agent who cooperates
with positive probability. 

Theorem \ref{thm:only-defection-is-stable} may come as a surprise
in light of a number of existing papers that have presented various
equilibrium constructions that support cooperation in any Prisoner's
Dilemma game that is played in a population of randomly matched agents.
Our result demonstrates that, in the presence of a small fraction
of committed agents, the mechanisms that have been proposed to support
cooperation fail, regardless of how these committed agents play (except
in the ``knife-edge'' case of $g=l$; see \citealp{dilme2016helping}
and Remark \ref{enu:The-threshold-case} in Section \ref{subsec:Stability-of-Cooperation}).
Thus, our paper provides an explanation of why experimental evidence
suggests that subjects' behavior corresponds neither to contagious
equilibria (see, e.g., \citealp{duffy2009cooperative}) nor to belief-free
equilibria (see, e.g., \citealp{matsushima2013behavioral}). The empirical
predictions of our model are discussed in Appendix \ref{sec:Empirical-Predictions}.

Our second main result (Theorem \ref{thm:cooperation-defensive-PDs})
shows that cooperation is a perfect equilibrium outcome in any defensive
Prisoner's Dilemma game ($g<l$) when players observe at least two
actions. Moreover, there is an essentially unique distribution of
strategies that support cooperation, according to which: (a) all agents
cooperate when observing no defections, (b) all agents defect when
observing at least 2 defections, (c) the normal agents defect with
an average probability of $0<q<1$ when observing a single defection.
The intuition for the result is as follows. Defection yields a direct
gain that is increasing in the partner's probability of defection
(due to the game being defensive). In addition, defection results
in an indirect loss because it induces future partners to defect when
they observe the current defection. This indirect loss is independent
of the current partner's behavior. One can show that there always
exists a probability $q$ such that the above distribution of strategies
balances the direct gain and the indirect loss of defection, conditional
on the agent observing a single defection. Furthermore, cooperation
is the unique best reply conditional on the agent observing no defections,
and defection is the unique best reply conditional on the agent observing
at least two defections.

Next, we analyze the case of the observation of a single action (i.e.,
$k=1$). Proposition \ref{prop:observing-single-action-1} shows that
cooperation is a perfect equilibrium outcome in a defensive Prisoner's
Dilemma if and only if the bonus of defection is not too large (specifically,
$g\leq1)$. The intuition is that similar arguments used to obtain
the result above imply that there exists a unique mean probability
$q<1$ by which agents defect when observing a defection in any cooperative
perfect equilibrium. This implies that a deviator who always defects
succeeds in getting a payoff of $1+g$ in a fraction $1-q$ of the
interactions, and that such a deviator outperforms the incumbents
if $g$ is too large. 

\paragraph{Observations Based on Action Profiles}

So far we have assumed that each agent observes only the partner's
(Bob's) behavior against other opponents, but that she cannot observe
the behavior of the past opponents against Bob. In Section \ref{sec:General-Observation-Structures}
we relax this assumption. Specifically, we study three observation
structures: the first two seem to be empirically relevant, and the
third one is theoretically important since it allows us to construct
an equilibrium that sustains cooperation in all Prisoner's Dilemma
games.
\begin{enumerate}
\item \emph{Observing conflicts}: Each agent observes, in each of the $k$
sampled interactions of her partner, whether there was mutual cooperation
(i.e., no conflict: both partners are ``happy'') or not (i.e., partners
complain about each other, but it is too costly for an outside observer
to verify who actually defected). Such an observation structure (which
we have not seen described in the existing literature) seems like
a plausible way to capture non-verifiable feedback about the partner's
behavior. 
\item \emph{Observing action profiles: }Each agent observes the full action
profile in each of the sampled interactions.
\item \emph{Observing actions against cooperation: }Each agent observes,
in each of the sampled interactions, what action the partner took
provided that the partner's opponent cooperated. If the partner's
opponent defected then there is no information about what the partner
did.
\end{enumerate}
~~~~It turns out that the stability of cooperation in the first
two observation structures crucially depends on a novel classification
of Prisoner's Dilemma games. We say that a Prisoner's Dilemma game
is \emph{acute} if $g>\frac{l+1}{2}$, and \emph{mild} if $g<\frac{l+1}{2}$.
The threshold between the two categories, namely, $g=\frac{l+1}{2}$,
is characterized by the fact that the gain from a single unilateral
defection is exactly half the loss incurred by the partner who is
the sole cooperator. Consider a setup in which an agent is deterred
from unilaterally defecting because it induces future partners to
unilaterally defect against the agent with some probability. Deterrence
in acute Prisoner\textquoteright s Dilemmas requires this probability
to be more than 50\%, while a probability of below 50\% is enough
to deter deviations in mild PDs. Figure \ref{fig:Classification-of-Prisoner's}
(in Section \ref{subsec:Acute-and-Mild}) illustrates the classification
of Prisoner's Dilemma games.

Our next results (Theorems \ref{thm:stable-cooperation-observing-conflicts}\textendash \ref{thm:stable-cooperation-observing-action-profiles})
show that in both observation structures (conflicts or action profiles,
and any $k\geq2$) cooperation is a perfect equilibrium outcome if
and only if the underlying Prisoner's Dilemma game is mild. Moreover,
cooperation is supported by essentially the same unique behavior as
in Theorem \ref{thm:cooperation-defensive-PDs}. The intuition for
why cooperation cannot be sustained in acute games with observation
of conflicts is as follows. In order to support cooperation agents
should be deterred from defecting against cooperators. As discussed
above, in acute games, such deterrence requires that each such defection
induce future partners to defect with a probability of at least 50\%.
However, this requirement implies that defection is contagious: each
defection by an agent makes it possible that future partners observe
a conflict both when being matched with the defecting agent and when
being matched with the defecting agent's partner. Such future partners
defect with a probability of at least $50\%$ when making such observations.
Thus the fraction of defections grows steadily, until all normal agents
defect with high probability.

The intuition for why cooperation cannot be sustained in acute games
with observation of action profiles is as follows. The fact that deterring
defections in acute games requires future partners to defect with
a probability of at least 50\% when observing a defection implies
that when an agent (Alice) observes her partner (Bob) to defect against
a cooperative opponent, then Bob is more likely to do so because he
is a normal agent who observed his past opponent to defect than because
Bob is a committed agent. This implies that Alice puts a higher probability
on Bob defecting against her if she observes Bob to have defected
against a partner who also defected than she does if she observes
Bob to have defected against an opponent who cooperated. Thus, defecting
is the unique best reply when observing the partner defect against
a defector, but it removes the incentives required to support stable
cooperation.

Finally, we show that the third observation structure, \emph{observing
actions against cooperation}, is optimal in the sense that it sustains
cooperation as a perfect equilibrium outcome for any Prisoner\textquoteright s
Dilemma game (Theorem \ref{thmtertiary-observation}). The intuition
for this result is that not allowing Alice to observe Bob's behavior
against a defector helps to sustain cooperation because it implies
that defecting against a defector does not have any negative indirect
effect (in any steady state) because it is never observed by future
opponents. This encourages agents to defect against partners who are
more likely to defect (regardless of the values of $g$ and $l$).

\paragraph{Conventional Model and Unrestricted Strategies}

In Appendix \ref{sec:Conventional-Repeated-Game}, we relax the assumption
that agents are restricted to choosing only stationary strategies.
We present a conventional model of repeated games with random matching
that differs from the existing literature only by our introducing
a few committed agents. We show that this difference is sufficient
to yield most of our key results.

Specifically, the characterization of the conditions under which cooperation
can be sustained as a perfect equilibrium outcome (as summarized in
Table \ref{tab:Prisoner-Dilemma-1-1} in Section \ref{subsec:Analysis-of-the-mild-acute})
holds also when agents are not restricted to stationary strategies,
and even when agents observe the most recent past actions of the partner.
On the other hand, the relaxation of the stationarity assumption in
Appendix \ref{sec:Conventional-Repeated-Game} weakens the uniqueness
results of the main model in two respects: (1) rather than showing
that defection is the unique equilibrium outcome in offensive games,
we show only that it is impossible to sustain full cooperation in
such games; and (2) while a variant of the simple strategy of the
main model still supports cooperation when the set of strategies is
unrestricted, we are no longer able to show that this strategy is
the unique way to support full cooperation.

\paragraph{Structure}

Section \ref{sec:Model} presents the model. Our solution concept
is described in Section \ref{sec:Solution-Concept}. Section \ref{sec:Main-Results}
studies the observation of actions. Section \ref{sec:General-Observation-Structures}
extends the model to deal with general observation structures. We
discuss the related literature in Section \ref{sec:Related-Literature},
and conclude in Section \ref{subsec:Conclusion-and-Directions}. Appendix
\ref{sec:Conventional-Repeated-Game} adapts our key result to a conventional
model with an unrestricted set of strategies. Appendix \ref{sec:Empirical-Predictions}
discusses our empirical predictions. Appendix \ref{sec:Additional-Notation-and}
presents technical definitions. In Appendix \ref{sec:Evolutionary-Stability}
we present the refinements of strict perfection, evolutionary stability,
and robustness. The formal proofs appear in Appendix \ref{sec:Proofs}.
Appendix \ref{subsec:Cheap-Talk} studies the introduction of cheap
talk to our setup. The appendices are available online in the supplementary
material.

\section{Stationary Model\label{sec:Model}}

\subsection{Environment}

We model an environment in which patient agents in a large population
are randomly matched in each round to play a two-player symmetric
one-shot game. For tractability we assume throughout the paper that
the population is a continuum.\footnote{\label{fn:continuum-population}The results can be adapted to a setup
with a large finite population. We do not formalize a large finite
population, as this adds much complexity to the model without giving
substantial new insights. Most of the existing literature also models
large populations as continua (see, e.g., \citealp{Rubinstein-Wolinsky,weibull1995evolutionary,dixit2003modes,herold2009evolutionary,sakovics2012matters,alger2013homo}).
\citet{kandori1992social} and \citet{ellison1994cooperation} show
that large finite populations differ from infinite populations because
only the former can induce contagious equilibria. However, as noted
by \citet[p. 578]{ellison1994cooperation}, and as discussed in Section
\ref{subsec:Stability-of-Defection}, these contagious equilibria
fail in the presence of a single ``crazy'' agent who always defects.} We further assume that the agents are infinitely lived and do not
discount the future (i.e., they maximize the average per-round long-run
payoff). Alternatively, our model can be interpreted as representing
interactions between finitely lived agents who belong to infinitely
lived dynasties, such that an agent who dies is succeeded by a protégé
who plays the same strategy as the deceased mentor, and each agent
observes $k$ random actions played by the partner's dynasty.

Before playing the game, each agent (she) privately observes $k$
random actions that her partner (he) played against other opponents
in the past. As described in detail below, agents are restricted to
using only stationary strategies, such that each agent's behavior
depends only on the signal about the partner, and not on the agent's
own past play or on time. Thus, if all agents observe signals that
come from a stationary distribution then the agents' behavior will
result in a well-defined aggregate distribution of actions that is
also stationary. We focus on steady states of the population, in which
the distribution of actions, and hence the distribution of signals,
is indeed stationary. In such steady states, the $k$ actions that
an agent observes about her partner are drawn independently from the
partner's stationary distribution of actions. This sampling procedure
may be interpreted as the limit of a process in which each agent randomly
observes $k$ actions that are uniformly sampled from the last $n$
interactions of the partner, as $n\rightarrow\infty$. 

To simplify the notation, we assume that the underlying game has two
actions, though all our concepts are applicable to games with any
finite number of actions. An \emph{environment} is a pair $E=\left(G,k\right)$,
where $G=\left(A=\left\{ c,d\right\} ,\pi\right)$ is a two-player
symmetric normal-form game, and $k\in\mathbb{N}$ is the number of
observed actions. Let $\pi:A\times A\rightarrow\mathbb{R}$ be the
payoff function of the underlying game. We refer to action $c$ (resp.,
$d$) as \emph{cooperation} (resp., \emph{defection}), since we will
focus on the Prisoner's Dilemma in our results. Let $\Delta\left(A\right)$
denote the set of mixed actions (distributions over $A$), and let
$\pi$ be extended to mixed actions in the usual linear way. We use
the letter $a$ (resp., $\alpha$) to denote a typical pure (mixed)
action. With a slight abuse of notation let $a\in A$ also denote
the element in $\Delta\left(A\right)$ that assigns probability 1
to $a$. We adopt this convention for all probability distributions
throughout the paper. 

\subsection{Stationary Strategy\label{subsec:Stationary-Strategy} }

The signal observed about the partner is the number of times he played
each action $a\in A$ in the sample of $k$ observed actions. Let
$M=\left\{ 0,...,k\right\} $ denote the set of feasible signals,
where signal $m\in M$ is interpreted as the number of times that
the partner defected in the sampled $k$ observations.\footnote{We do not allow agents to manipulate the observed signals. In our
companion paper (\citealp{heller-mohlin}) we study a related setup
in which agents are allowed to exert effort in deception by influencing
the signal observed by the opponent.}

Given a distribution of actions $\alpha\in\Delta\left(A\right)$ and
an environment $E=\left(G,k\right)$, let $\nu_{\alpha}\left(m\right)$
be the probability of an agent observing signal $m$ conditional on
being matched with a partner who plays on average the distribution
of actions $\alpha$. That is, $\nu\left(\alpha\right):=\nu_{\alpha}\in\Delta\left(M\right)$
is a binomial signal distribution that describes a sample of $k$
i.i.d. actions, where each action is distributed according to $\alpha$:
\begin{equation}
\forall\left(m\right)\in M,\,\,\,\,\nu_{\alpha}\left(m\right)=\frac{k!\cdot\left(\alpha\left(d\right)\right)^{m}\cdot\left(\alpha\left(c\right)\right)^{\left(k-m\right)}}{m!\cdot\left(k-m\right)!}.\label{eq:multinomial}
\end{equation}

A \emph{stationary strategy} (henceforth, \emph{strategy)} is a mapping
\emph{$s:M\rightarrow\Delta\left(A\right)$} that assigns a mixed
action to each possible signal. Let $s_{m}\in\Delta\left(A\right)$
denote the mixed action assigned by strategy $s$ after observing
signal $m$. That is, for each action $a\in A$, $s_{m}\left(a\right)=s\left(m\right)\left(a\right)$
is the probability that a player who follows strategy $s$ plays action
$a$ after observing signal $m$. We also let $a$ denote the strategy
$s$ that plays action $a$ regardless of the signal, i.e., $s_{m}\left(a\right)=1$
for all $m\in M$. Strategy $s$ is \emph{totally mixed}, if for each
action $a\in A$, and signal $m\in M$ $s_{m}\left(a\right)>0$. Let
$\mathcal{S}$ denote the set of all strategies. Given strategy $s$
and distribution of signals $\nu\in\Delta\left(M\right)$, let $s\left(\nu\right)\in\Delta\left(A\right)$
be the distribution of actions played by an agent who follows strategy
$s$ and observes a signal sampled from $\nu$:
\[
\forall a\in A,\,\,\,\,s\left(\nu\right)\left(a\right)=\sum_{m\in M}\nu\left(m\right)\cdot s_{m}\left(a\right).
\]

\subsection{Signal Profile and Steady State}

Fix an environment and a finite set of strategies $S$. A \emph{signal
profile $\theta:S\rightarrow\Delta\left(M\right)$ is }a function
that assigns a distribution of signals for each strategy in $S$.
We interpret $\theta_{s}\left(m\right)$ as the probability that signal
$m$ is observed when a partner playing strategy $s$ is encountered.
Let $O_{S}$ be the set of all signal profiles defined over $S$.
Given a strategy $\sigma\in\Delta\left(S\right)$ and a signal profile
$\theta\in O_{S}$, let $\theta_{\sigma}\in\Delta\left(M\right)$
be the \emph{average distribution of signals in the population}, i.e.,
$\theta_{\sigma}\left(m\right):=\sum_{s\in S}\sigma\left(s\right)\cdot\theta_{s}\left(m\right)$. 

We say that a signal profile $\theta:S\rightarrow\Delta\left(M\right)$
is \emph{consistent} with distribution of strategies $\sigma\in\Delta\left(S\right)$
if 
\begin{equation}
\forall m\in M,\,\,s\in S,\,\,\,\,\theta_{s}\left(m\right)=\nu\left(s\left(\theta_{\sigma}\right)\right)\left(m\right).\label{eq:transformation-between-signal=00003Dprofiles-2}
\end{equation}

The interpretation of the consistency requirement is that a population
of agents who follow the distribution of strategies $\sigma$ and
observe signals about the partners sampled from the profile $\theta$
have to behave in a way that induces the same profile of signal distributions
$\theta$. Specifically, when Alice, who follows strategy $s$, is
being matched with a random partner whose strategy is sampled according
to $\sigma$, she observes a random signal according to the ``current''
average distribution of signals in the population $\theta_{\sigma}$.
As a result her distribution of actions is $s\left(\theta_{\sigma}\right)$,
and thus her behavior induces the signal distribution $\nu\left(s\left(\theta_{\sigma}\right)\right)$.
Consistency requires that this induced signal distribution coincide
with $\theta_{s}$.

A steady state is a triple consisting of (1) a finite set of strategies
$S$ interpreted as the strategies that are played by the agents in
the population, (2) a distribution $\sigma$ over $S$ interpreted
as a description of the fraction of agents following each strategy,
and (3) a consistent signal profile $\theta:S\rightarrow\Delta\left(M\right)$.
Formally:
\begin{defn}
\label{def:state-unperturbed}A \emph{steady state }(or \emph{state}
for short) of an environment \emph{$\left(G,k\right)$} is a triple
$\left(S,\sigma,\theta\right)$ where $S\subseteq\mathcal{S}$ is
a finite set of strategies, $\sigma\in\Delta\left(S\right)$ is a
distribution with full support over $S$, and $\theta:S\rightarrow\Delta\left(M\right)$
is a consistent signal profile.
\end{defn}
When the set of strategies is a singleton, i.e., $S=\left\{ s\right\} $,
we omit the degenerate distribution assigning a mass of one to $s$,
and we write the steady state as a pair $\left(\left\{ s\right\} ,\theta\right).$
We adopt this convention, of omitting reference to degenerate distributions,
throughout the paper. 

A standard argument shows that any distribution of strategies admits
a consistent signal profile (Lemma \ref{lem:existance_consitent_outcomes-1-1}
in Appendix \ref{sec:Additional-Notation-and}). Some distributions
induce multiple consistent profiles of signal distributions. For example,
suppose that $k=3$, and everyone follows the strategy of playing
the most frequently observed action (i.e., defecting iff $m\geq2$).
In this setting there are three consistent signal profiles: one in
which everyone cooperates, one in which everyone defects, and one
in which everyone plays (on average) uniformly\footnote{In a companion paper (\citealp{heller2017social}) we study in a broader
setup necessary and sufficient conditions for a strategy distribution
admitting a unique consistent signal profile.}.

\subsection{Perturbed Environment}

As discussed in the Introduction, and as argued by \citet[p. 578]{ellison1994cooperation},
it seems implausible that in large populations all agents are rational
and know exactly the strategies played by other agents in the community.
Motivated by this observation, we introduce the notion of a perturbed
environment in which a small fraction of agents in the population
are committed to playing specific strategies, even though these strategies
are not necessarily payoff-maximizing.

A perturbed environment is a tuple consisting of (1) an environment,
(2) a distribution $\lambda$ over a set of commitment strategies
$S^{C}$ that includes a totally mixed strategy, and (3) a number
$\epsilon$ representing the share of agents who are committed to
playing strategies in $S^{C}$ (henceforth, \emph{committed agents}).
The remaining $1-\epsilon$ share of the agents can play any strategy
in $\mathcal{S}$ (henceforth, \emph{normal agents}). Formally:
\begin{defn}
\label{def:A-perturbed-environment}A \emph{perturbed} \emph{environment}
is a tuple $E_{\epsilon}=\left(\left(G,k\right),\left(S^{C},\lambda\right),\epsilon\right)$,
where $G$ is the underlying game, $k\in\mathbb{N}$ is the number
of observed actions, $S^{C}$ is a non-empty finite set of strategies
(called, \emph{commitment strategies}) that includes a totally mixed
strategy, $\lambda\in\Delta\left(S^{C}\right)$ is a distribution
with full support over the commitment strategies, and $\epsilon\geq0$
is the mass of committed agents in the population.

We require $S^{C}$ to include at least one totally mixed strategy
because we want all signals to be observed with positive probability
in a perturbed environment when $\epsilon>0$. (This is analogous
to the requirement in \citealp{selten1975reexamination}, that all
actions be played with positive probability in the perturbations defining
a perfect equilibrium.)

Throughout the paper we look at the limit in which the share of committed
agents, $\epsilon$, converges to zero. This is the only limit taken
in the paper. We use the notation of $O\left(\epsilon\right)$ (resp.,
$O\left(\epsilon^{2}\right)$) to refer to functions that are in the
order of magnitude of $\epsilon$ (resp., $\epsilon^{2}$), i.e.,
$\frac{f\left(\epsilon\right)}{\epsilon}\rightarrow_{\epsilon\rightarrow0}0$
(resp., $\frac{f\left(\epsilon\right)}{\epsilon^{2}}\rightarrow_{\epsilon\rightarrow0}0$).

We refer to $\left(S^{C},\lambda\right)$ as a \emph{distribution
of commitments}. With a slight abuse of notation, we identify an \emph{unperturbed
environment} $\left(\left(G,k\right),\left(S^{C},\lambda\right),\epsilon=0\right)$
with the equivalent environment $\left(G,k\right)$.
\end{defn}
\begin{rem}
To simplify the presentation, the definition of perturbed environment
includes only commitment strategies, and it does not allow ``trembling
hand'' mistakes. As discussed in Remark \ref{enu:General-Noise-Structures:}
in Section \ref{subsec:Stability-of-Cooperation}, the results also
hold in a setup in which agents also tremble, as long as the probability
by which a normal agent trembles is of the same order of magnitude
as the frequency of committed agents.
\end{rem}
One of our main results (Theorem \ref{thm:only-defection-is-stable})
requires an additional mild assumption on the perturbed environment
that rules out the knife-edge case in which all agents (committed
and non-committed alike) behave exactly the same. Specifically, a
set of commitments is regular if for each distribution of actions
$\alpha$, there exists a committed strategy $s$ that does not play
distribution $\alpha$ when observing the signal distribution induced
by $\alpha$. Formally:
\begin{defn}
\label{def:regularity-1}A set of commitment strategies $S^{C}$ is
\emph{regular} if for each distribution of actions $\alpha\in\Delta\left(A\right)$,
there exists a strategy $s\in S^{C}$ such that $s_{\nu\left(\alpha\right)}\neq\alpha$.
\end{defn}
If the set of commitments is regular, then we say that the distribution
$\left(S^{C},\lambda\right)$ and the perturbed environment $\left(\left(G,k\right),\left(S^{C},\lambda\right),\epsilon\right)$
are regular. An example of a regular set of commitments is the set
that includes strategies $s\equiv\alpha_{1}$ and $s'\equiv\alpha_{2}$
that induce agents to play mixed actions $\alpha_{1}\neq\alpha_{2}$
regardless of the observed signal.

\subsection{Steady State in a Perturbed Environment}

Fix a perturbed environment $E_{\epsilon}=\left(\left(G,k\right),\left(S^{C},\lambda\right),\epsilon\right)$
and a finite set of strategies $S^{N}$, interpreted as the strategies
followed by the normal agents in the population. We redefine a \emph{signal
profile $\theta:S^{C}\cup S^{N}\rightarrow\Delta\left(M\right)$ }as
a function that assigns a binomial distribution of signals to each
strategy in $S^{C}\cup S^{N}$.

Given a distribution over strategies of the normal agents $\sigma\in\Delta\left(S^{N}\right)$
and a signal profile $\theta\in O_{S^{C}\cup S^{N}}$, let $\theta_{\left(\left(1-\epsilon\right)\cdot\sigma+\epsilon\cdot\lambda\right)}\in\Delta\left(M\right)$
be the \emph{average distribution of signals in the population}, i.e.,
$\theta_{\left(\left(1-\epsilon\right)\cdot\sigma+\epsilon\cdot\lambda\right)}\left(m\right):=\sum_{s\in S^{C}\cup S^{N}}\left(\left(1-\epsilon\right)\cdot\sigma+\epsilon\cdot\lambda\right)\left(s\right)\cdot\theta_{s}\left(m\right)$.
We adapt the definitions of a consistent signal profile and of a steady
state to perturbed environments. This straightforward adaptation is
presented in detail in Appendix \ref{sec:Additional-Notation-and}.

The following example demonstrates a specific steady state in a specific
perturbed environment. The example is intended to clarify the various
definitions of this section and, in particular, the consistency requirement.
Later, we revisit the same example to explain the essentially unique
perfect equilibrium that supports cooperation.
\begin{example}
\label{exa:revisiting-example}Consider the perturbed environment
\emph{$\left(\left(G,k=2\right),\left(\left\{ s^{u}\equiv0.5\right\} \right),\epsilon\right)$,
}in which each agent observes two of her partner's actions, there
is a single commitment strategy, denoted by $s^{u}$, which is followed
by a fraction $0<\epsilon<<1$ of committed agents, who choose each
action with probability $0.5$ regardless of the observed signal.
Let $\left(S=\left\{ s^{1},s^{2}\right\} ,\sigma=\left(\frac{1}{6},\frac{5}{6}\right),\theta\right)$
be the following steady state. The state includes two normal strategies:
$s^{1}$ and $s^{2}$. The strategy $s^{1}$ defects iff $m\geq1$,
and the strategy $s^{2}$ defects iff $m\geq2$. The distribution
$\sigma$ assigns a mass of $\frac{1}{6}$ to $s^{1}$ and a mass
of $\frac{5}{6}$ to $s^{2}$. The consistent signal profile $\theta$
is defined as follows (neglecting terms of $O\left(\epsilon^{2}\right)$
throughout the example): 
\begin{equation}
\theta_{s^{u}}\left(m\right)=\begin{cases}
25\% & if\,m=0\\
50\% & if\,m=1\\
25\% & if\,m=2,
\end{cases}\,\,\,\,\,\theta_{s^{1}}\left(m\right)=\begin{cases}
1-3.5\cdot\epsilon & if\,m=0\\
3.5\cdot\epsilon & if\,m=1\\
0 & if\,m=2
\end{cases}\,\,\,\,\,\theta_{s^{2}}\left(m\right)=\begin{cases}
1-0.5\cdot\epsilon & if\,m=0\\
0.5\cdot\epsilon & if\,m=1\\
0 & if\,m=2.
\end{cases}\label{eq:eta-star-example}
\end{equation}
To confirm the consistency of $\theta$, we have first to calculate
the average distribution of signals in the population:

\[
\theta_{\left(\left(1-\epsilon\right)\cdot\sigma+\epsilon\cdot\lambda\right)}\left(m\right)=\begin{cases}
1-1.75\cdot\epsilon & if\,m=0\\
1.5\cdot\epsilon & if\,m=1\\
0.25\cdot\epsilon & if\,m=2.
\end{cases}
\]
Using $\theta_{\left(\left(1-\epsilon\right)\cdot\sigma+\epsilon\cdot\lambda\right)}$,
we confirm the consistency of $\theta_{s^{1}}$ and $\theta_{s^{2}}$
(the consistency of $\theta_{s^{u}}$is immediate). We do so by calculating
distribution of actions played by a player following strategy $s_{i}$
who observes the distribution of actions of a random partner:

\[
\begin{array}{ccc}
s^{1}\left(\theta_{\left(\left(1-\epsilon\right)\cdot\sigma+\epsilon\cdot\lambda\right)}\right)\left(c\right) & = & 1-1.75\cdot\epsilon\\
s^{1}\left(\theta_{\left(\left(1-\epsilon\right)\cdot\sigma+\epsilon\cdot\lambda\right)}\right)\left(d\right) & = & 1.75\cdot\epsilon
\end{array}\,\,\,\,\,\,\,\,\,\begin{array}{ccc}
s^{2}\left(\theta_{\left(\left(1-\epsilon\right)\cdot\sigma+\epsilon\cdot\lambda\right)}\right)\left(c\right) & = & 1-0.25\cdot\epsilon,\\
s^{2}\left(\theta_{\left(\left(1-\epsilon\right)\cdot\sigma+\epsilon\cdot\lambda\right)}\right)\left(d\right) & = & 0.25\cdot\epsilon.
\end{array}
\]

Note that $s^{1}\left(\theta_{\left(\left(1-\epsilon\right)\cdot\sigma+\epsilon\cdot\lambda\right)}\right)\left(d\right)=1-\theta_{\left(\left(1-\epsilon\right)\cdot\sigma+\epsilon\cdot\lambda\right)}\left(2\cdot c\right)$
and $s^{2}\left(\theta_{\left(\left(1-\epsilon\right)\cdot\sigma+\epsilon\cdot\lambda\right)}\right)\left(d\right)=\theta_{\left(\left(1-\epsilon\right)\cdot\sigma+\epsilon\cdot\lambda\right)}\left(2\cdot d\right)$.
The final step in showing that $\theta$ is a consistent profile is
the observation that each $\theta_{s^{i}}$ coincides with the binomial
distribution that is induced by $s^{i}\left(\theta_{\left(\left(1-\epsilon\right)\cdot\sigma+\epsilon\cdot\lambda\right)}\right)$.
\end{example}

\subsection{Discussion of the Model\label{subsec:model-Discussion}}

Our model differs from most of the existing literature on community
enforcement in three key dimensions (see, e.g., \citealp{kandori1992social,ellison1994cooperation,dixit2003modes,deb2012cooperation,deb2012community}).
In what follows we discuss these three key differences, and their
implications on our results.
\begin{enumerate}
\item \emph{The presence of a few committed agents}. If one removes the
commitment types from our setup, then one can show (by using belief-free
equilibria, as in \citealp{takahashi2010community}) that: (1) it
is always possible to support full cooperation as an equilibrium outcome,
and (2) there are various strategies that sustain full cooperation.
The results of this paper show that the introduction of a few committed
agents, regardless of how they behave, implies very different results:
(1) defection is the \emph{unique }equilibrium payoff in offensive\emph{
}Prisoner's Dilemmas (Theorem \ref{thm:only-defection-is-stable}),
and (2) there is an essentially unique strategy combination that supports
a cooperative equilibrium in defensive Prisoner's Dilemmas. The intuition
is that the presence of committed agents implies that observation
of past actions must have some influence on the likely behavior of
the partner in the current match (more detailed discussions of this
issue follow Theorem \ref{thm:only-defection-is-stable} and Remark
10).
\item \emph{Restriction to Stationary Strategies}. In our model we restrict
agents to using stationary strategies that condition only on the number
of times they observed each of the partner's actions being played
in past interactions. We allow agents to condition their play neither
on the order in which the observed actions were played in the past,
nor on the agent's own history of play, nor on calendar time. The
assumption simplifies the presentation of the model and results. In
addition, the assumption allows us to achieve uniqueness results that
might not hold without stationarity (as discussed in Section \ref{subsec:Discussion-of-the-results-repeated-games}). 
\item \emph{Not having a ``global time zero.'' }Most of the existing literature
represents interactions within a community as a repeated game that
has a ``global time zero,'' in which the first ever interaction
takes place. In many real-life situations, the interactions within
a community began a long time ago and have continued, via overlapping
generations, to the present day. It seems implausible that today's
agents condition their behavior on what happened in the remote past
(or on calendar time). For example, trade interactions have been been
taking place from time immemorial. It seems unreasonable to assume
that Alice's behavior today is conditioned on what transpired in some
long-forgotten time $t=0$, when, say, two hunter-gatherers were involved
in the first ever trade. We suggest that, even though real-world interactions
obviously begin at some definite date, a good way of modeling what
the interacting agents think about the situation may be to get rid
of global time zero and focus on strategies that do not condition
on what happened in the remote past. The lack of a global time zero
is the reason why, unlike in repeated games, a distribution of strategies
does not uniquely determine the behavior and the payoffs of the agent,
so that one must explicitly add the consistent signal profile $\theta$
as part of the description of the state of the population. \\
It is possible to interpret a steady state $\left(S,\sigma,\theta\right)$
as a kind of initial condition for society, in which agents already
have a long-existing past. That is, we begin our analysis of community
interaction at a point in time when agents have for a long time followed
the strategy distribution $\left(S,\sigma\right)$ yielding the consistent
signal profile $\theta$. We then ask whether any patient agent has
a profitable deviation from her strategy. If not, then the steady
state $\left(S,\sigma,\theta\right)$ is likely to persist. This approach
stands in contrast to the standard approach that studies whether or
not agents have a profitable deviation at a time $t>>1$ following
a long history that started with the first ever interaction at $t=0$.
\end{enumerate}
\emph{In Appendix \ref{sec:Conventional-Repeated-Game} we present
a conventional repeated game  model that differs from the existing
literature in only one key aspect: the presence of a few committed
agents.} In particular, this alternative model features standard calendar
time, and agents discount the future, observe the most recent past
actions of the partner, and are not limited to choosing only stationary
strategies. We show that most of our results hold also in this setup.
We feel that this alternative model, while being closer to the existing
literature than the main model, suffers from added technical complexity
that may hinder the model from being insightful and accessible. 

\section{Solution Concept\label{sec:Solution-Concept}}

\subsection{Long-Run Payoff}

In this subsection we define the long-run average (per-round) payoff
of a patient agent who follows a stationary strategy $s$, given a
steady state $\left(S^{N},\sigma,\theta\right)$ of a perturbed environment
$\left(\left(G,k\right),\left(S^{C},\lambda\right),\epsilon\right)$.
The same definition, when taking $\epsilon=0$, holds for an unperturbed
environment. 

We begin by extending the definition of a consistent signal profile
$\theta$ to non-incumbent strategies. For each non-incumbent strategy
$\hat{s}\in\mathcal{S}\backslash\left(S^{N}\cup S^{C}\right)$, define
$\theta\left(\hat{s}\right)=\theta_{\hat{s}}$ as the distribution
of signals induced by a deviating agent who follows strategy $\hat{s}$
and observes the distribution of signals induced by a random partner
in the population (sampled according to $\left(1-\epsilon\right)\cdot\sigma\left(s'\right)+\epsilon\cdot\lambda\left(s'\right)$).
That is, for each strategy $\hat{s}\in\mathcal{S}\backslash\left(S\cup S^{C}\right)$,
and each signal $m\in M$, we define
\[
\theta_{\hat{s}}\left(m\right)=\left(\nu\left(\hat{s}\left(\theta_{\left(\left(1-\epsilon\right)\cdot\sigma+\epsilon\cdot\lambda\right)}\right)\right)\right)\left(m\right).
\]
 We define the long-run payoff of an agent who follows an arbitrary
strategy $s\in\mathcal{S}$ as:
\begin{equation}
\pi_{s}\left(S^{N},\sigma,\theta\right)=\sum_{s'\in S^{N}\cup S^{C}}\left(\left(1-\epsilon\right)\cdot\sigma\left(s'\right)+\epsilon\cdot\lambda\left(s'\right)\right)\cdot\left(\sum_{\left(a,a'\right)\in A\times A}s_{\theta\left(s'\right)}\left(a\right)\cdot s'_{\theta\left(s\right)}\left(a'\right)\cdot\pi\left(a,a'\right)\right).\label{eq:long-run payoff}
\end{equation}
Eq. (\ref{eq:long-run payoff}) is straightforward. The inner (right-hand)
sum (i.e., $\sum_{\left(a,a'\right)\in A\times A}s_{\theta\left(s'\right)}\left(a\right)\cdot s'_{\theta\left(s\right)}\left(a'\right)\cdot\pi\left(a,a'\right)$)
calculates the expected payoff of Alice who follows strategy $s$
conditional on being matched with a partner who follows strategy $s'$.
The outer sum weighs these conditional expected payoffs according
to the frequency of each incumbent strategy $s'$ (i.e., $\left(\left(1-\epsilon\right)\cdot\sigma\left(s'\right)+\epsilon\cdot\lambda\left(s'\right)\right)$),
which yields the expected payoff of Alice against a random partner
in the population.

Let $\pi\left(S,\sigma,\theta\right)$ be the average payoff of the
\emph{normal} agents in the population: 
\[
\pi\left(S^{N},\sigma,\theta\right)=\sum_{s\in S^{N}}\sigma\left(s\right)\cdot\pi_{s}\left(S^{N},\sigma,\theta\right).
\]

\subsection{Nash and Perfect Equilibrium}

A steady state is a Nash equilibrium if no agent can obtain a higher
payoff by a unilateral deviation. Formally:
\begin{defn}
\label{def:Nash-equilibrium}The steady state $\left(S^{N},\sigma,\theta\right)$
of perturbed environment $\left(\left(G,k\right),\left(S^{C},\lambda\right),\epsilon\right)$
is a \emph{Nash equilibrium} if for each strategy $s\in\mathcal{S}$,
it is the case that $\pi_{s}\left(S^{N},\sigma,\theta\right)\leq\pi\left(S^{N},\sigma,\theta\right)$.
\end{defn}
Note that the $1-\epsilon$ normal agents in such a Nash equilibrium
must obtain the same maximal payoff. That is, each normal strategy
$s\in S^{N}$ satisfies $\pi_{s}\left(S^{N},\sigma,\theta\right)=\pi\left(S^{N},\sigma,\theta\right)\geq\pi_{s'}\left(S^{N},\sigma,\theta\right)$
for each strategy $s'\in\mathcal{S}$. However, the $\epsilon$ committed
agents may obtain lower payoffs.

A steady state is a (regular) perfect equilibrium if it is the limit
of Nash equilibria of (regular) perturbed environments when the frequency
of the committed agents converges to zero. Formally (where the standard
definitions of convergence of strategies, distributions and states
is presented in Appendix \ref{sec:Additional-Notation-and}):
\begin{defn}
A steady state $\left(S^{*},\sigma^{*},\theta^{*}\right)$ of the
environment $\left(G,k\right)$ is a \emph{(regular) perfect equilibrium
}if there exist a (regular) distribution of commitments $\left(S^{C},\lambda\right)$
and converging sequences $\left(S_{n}^{N},\sigma_{n},\theta_{n}\right)_{n}\rightarrow_{n\rightarrow\infty}\left(S^{*},\sigma^{*},\theta^{*}\right)$
and $\left(\epsilon_{n}>0\right)_{n}\rightarrow_{n\rightarrow\infty}0$,
such that for each $n$, the state $\left(S_{n}^{N},\sigma_{n},\theta_{n}\right)$
is a Nash equilibrium of the perturbed environment $\left(\left(G,k\right),\left(S^{C},\lambda\right),\epsilon_{n}\right)$.
In this case, we say that $\left(S^{*},\sigma^{*},\theta^{*}\right)$
is a (regular)\emph{ }perfect equilibrium with respect to distribution
of commitments\emph{ }$\left(S^{C},\lambda\right)$. If $\theta^{*}\equiv a$
, we say that action $a\in A$ is a (\emph{regular}) \emph{perfect
equilibrium action}.
\end{defn}
By standard arguments, any perfect equilibrium is a Nash equilibrium
of the unperturbed environment. In Appendix \ref{subsec:Implementing-a-Trembling-Hand}
we show that any symmetric (perfect) Nash equilibrium of the underlying
game corresponds to a (perfect) Nash equilibrium of the environment
in which all normal agents ignore the observed signal.

\subsection{Stronger Refinements of Perfect Equilibrium\label{subsec:Stronger-Refinements-of}}

In Appendix \ref{sec:Evolutionary-Stability} we present three refinements
of perfect equilibrium: strict perfection, evolutionary stability,
and robustness. The first refinement (strict perfection) is satisfied
by the equilibria constructed in Proposition \ref{pro:defection-is-evol-stable},
Theorem \ref{thm:cooperation-defensive-PDs}, and Theorem \ref{thm:stable-cooperation-observing-conflicts}.
The remaining refinements (evolutionary stability and robustness)
are satisfied by all the equilibria constructed in the paper.

The notion of perfect equilibrium might be considered too weak because
it may crucially depend on a specific set of commitment strategies.
The refinement of \emph{strict perfection} ($\grave{\textrm{a}}$
la \citealp{okada1981stability}) requires the equilibrium outcome
to be sustained regardless of which commitment strategies are present
in the population.

The notion of perfect equilibrium considers only deviations by a single
agent (who has mass zero in the infinite population). The refinement
of an \emph{evolutionarily stable strategy} ($\grave{\textrm{a}}$
la \citealp{smith1973lhe}) requires stability against a group of
agents with a small positive mass who jointly deviate.

The outcome of a perfect equilibrium may be non-robust in the sense
that small perturbations of the distribution of observed signals may
induce a change of behavior that moves the population away from the
consistent signal profile. We address this issue by introducing a
refinement that we call\emph{ robustness}, which requires that if
we slightly perturb the distribution of observed signals, then the
agents still play the same equilibrium outcome with a probability
very close to one (in the spirit of the notion of Lyapunov stability). 

\section{Prisoner's Dilemma and Observation of Actions\label{sec:Main-Results}}

\subsection{The Prisoner\textquoteright s Dilemma\label{subsec:The-Prisoner=002019s-Dilemma}}

Our results focus on environments in which the underlying game is
the Prisoner's Dilemma (denoted by $G_{PD}$), which is described
in Table \ref{tab:Prisoner-Dilemma-detailed}. The class of Prisoner's
Dilemma games is fully described by two positive parameters $g$ and
$l$. The two actions are denoted $c$ and $d$ , representing cooperation
and defection, respectively. When both players cooperate they both
get a high payoff (normalized to one), and when they both defect they
both get a low payoff (normalized to zero). When a single player defects
he obtains a payoff of $1+g$ (i.e., an additional payoff of $g$)
while his opponent gets $-l$. 
\begin{center}
\begin{table}[h]
\caption{\label{tab:Prisoner-Dilemma-detailed}Matrix Payoffs of Prisoner's
Dilemma Games }

\centering{}%
\begin{tabular}{|c|c|c|}
\hline 
 & \emph{c} & \emph{d}\tabularnewline
\hline 
\emph{~~c~~} & \emph{\Large{}$_{\underset{\,}{1}}\,\,\,\,\,\,^{\overset{\,}{1}}$} & \emph{\Large{}$_{\underset{\,}{-l}}\,\,^{\overset{\,}{1+g}}$}\tabularnewline
\hline 
\emph{d} & \emph{\Large{}$_{\underset{\,}{1+g}}\,\,^{\overset{\,}{-l}}$} & \emph{\Large{}$_{0}\,\,\,\,\,\,^{0}$}\tabularnewline
\hline 
\multicolumn{3}{c}{Prisoner\textquoteright s Dilemma }\tabularnewline
\multicolumn{3}{c}{$G_{PD}$: $g,l>0$ ,~ $g<l+1$}\tabularnewline
\end{tabular}~~~~~~~%
\begin{tabular}{|c|c|c|}
\hline 
 & \emph{c} & \emph{d}\tabularnewline
\hline 
\emph{~~c~~} & \emph{\Large{}$_{\underset{\,}{1}}\,\,\,\,\,\,^{\overset{\,}{1}}$} & \emph{\Large{}$_{\underset{\,}{-3}}\,\,^{\overset{\,}{2}}$}\tabularnewline
\hline 
\emph{d} & \emph{\Large{}$_{\underset{\,}{2}}\,\,^{\overset{\,}{-3}}$} & \emph{\Large{}$_{0}\,\,\,\,\,\,^{0}$}\tabularnewline
\hline 
\multicolumn{3}{c}{Ex. 1: Defensive PD}\tabularnewline
\multicolumn{3}{c}{$G_{D}$: $1=g<l=3$}\tabularnewline
\end{tabular}~~~~~~~%
\begin{tabular}{|c|c|c|}
\hline 
 & \emph{c} & \emph{d}\tabularnewline
\hline 
\emph{~~c~~} & \emph{\Large{}$_{\underset{\,}{1}}\,\,\,\,\,\,^{\overset{\,}{1}}$} & \emph{\Large{}$_{\underset{\,}{-1.7}}\,\,^{\overset{\,}{3.3}}$}\tabularnewline
\hline 
\emph{d} & \emph{\Large{}$_{\underset{\,}{3.3}}\,\,^{\overset{\,}{-1.7}}$} & \emph{\Large{}$_{0}\,\,\,\,\,\,^{0}$}\tabularnewline
\hline 
\multicolumn{3}{c}{Ex. 2: Offensive PD}\tabularnewline
\multicolumn{3}{c}{$G_{O}$: $2.3=g>l=1.7$}\tabularnewline
\end{tabular}
\end{table}
\par\end{center}

Following \citet{dixit2003modes} we classify Prisoner's Dilemma games
into two kinds: offensive and defensive.\emph{}\footnote{\citet{takahashi2010community} calls offensive (defensive) Prisoner's
Dilemmas submodular (supermodular). }\emph{ }In an\emph{ offensive} Prisoner's Dilemma there is a stronger
incentive to defect against a cooperator than against a defector (i.e.,
$g>l$); in a \emph{defensive} PD the opposite holds (i.e., $l>g$).
If cooperating is interpreted as exerting high effort, then the defensive
PD exhibits strategic complementarity; increasing one's effort from
low to high is less costly if the opponent exerts high effort. 

\subsection{Stability of Defection\label{subsec:Stability-of-Defection}}

We begin by showing that defection is a regular perfect equilibrium
action in any Prisoner's Dilemma game and for any $k$. Formally:
\begin{prop}
\label{pro:defection-is-evol-stable}Let $E=\left(G_{PD},k\right)$
be an environment\emph{. }Defection is a regular perfect equilibrium
action.
\end{prop}
The intuition is straightforward. Consider any distribution of commitment
strategies. Consider the steady state in which all the normal incumbents
defect regardless of the observed signal. It is immediate that this
strategy is the unique best reply to itself. This implies that if
the share of committed agents is sufficiently small, then always defecting
is also the unique best reply in the slightly perturbed environment.

Our first main result shows that defection is the \emph{unique} regular
perfect equilibrium in offensive games. 
\begin{thm}
\label{thm:only-defection-is-stable}Let $E=\left(G_{PD},k\right)$
be an environment, where $G$ is an offensive Prisoner's Dilemma (i.e.,
$g>l$). If $\left(S^{*},\sigma^{*},\theta^{*}\right)$ is a regular
perfect equilibrium, then $S^{*}=\left\{ d\right\} $ and $\theta^{*}=k$.
\end{thm}
\begin{proof}[Sketch of Proof]
The payoff of a strategy can be divided into two components: (1)
a \emph{direct }component: defecting yields additional $g$ points
if the partner cooperates and additional $l$ points if the partner
defects, and (2) an \emph{indirect} component: the strategy's average
probability of defection determines the distribution of signals observed
by the partners, and thereby determines the partner's probability
of defecting. For each fixed average probability of defection $q$
the fact that the Prisoner's Dilemma is offensive implies that the
optimal strategy among all those who defect with an average probability
of $q$ is to defect, with the maximal probability, against the partners
who are most likely to cooperate. This implies that all agents who
follow incumbent strategies are more likely to defect against partners
who are more likely to cooperate. As a result, mutants who always
defect outperform incumbents because they both have a strictly higher
direct payoff (since defection is a dominant action) and a weakly
higher indirect payoff (since incumbents are less likely to defect
against them). 
\end{proof}

\paragraph{Discussion of Theorem \ref{thm:only-defection-is-stable}}

The proof of Theorem \ref{thm:only-defection-is-stable} relies on
the assumption that agents are limited to choosing only stationary
strategies. The stationarity assumption implies that a partner who
has been observed to defect more in the past is more likely to defect
in the current match. However, this may no longer be true in a non-stationary
environment. In Appendix \ref{sec:Conventional-Repeated-Game} we
analyze the classic setup of repeated games, in which agents can choose
non-stationary strategies and observe the opponent's recent actions.
In that setup we are able to prove a weaker version of Theorem \ref{thm:only-defection-is-stable}
(namely, Theorem \ref{thmLrepeated-game}) which states that \emph{full}
cooperation cannot be supported as a perfect equilibrium outcome in
offensive Prisoner's Dilemmas (i.e., cooperation is not a perfect
equilibrium action in offensive games). 

Several papers in the existing literature present various mechanisms
to support cooperation in any Prisoner's Dilemma game. \citet[Theorem 1]{kandori1992social}
and \citet{ellison1994cooperation} show that in large finite populations
cooperation can be supported by contagious equilibria even when an
agent does not observe any signal about her partner (i.e., $k=0$).
In these equilibria each agent starts the game by cooperating, but
she starts defecting forever as soon as any partner has defected against
her. As pointed out by \citet[p. 578]{ellison1994cooperation}, if
we consider a large population in which at least one ``crazy'' agent
defects with positive probability in all rounds regardless of the
observed signal, then \citeauthor{kandori1992social}'s and \citeauthor{ellison1994cooperation}'s
equilibria fail because agents assign high probability to the event
that the contagion process has already begun, even after having experienced
a long period during which no partner defected against them. Recently,
\citet{dilme2016helping} presented a novel ``tit-for-tat''-like
contagious equilibrium that is robust to the presence of committed
agents, but only for the borderline case of $g=l$ (as discussed in
Remark \ref{enu:The-threshold-case} below). 

\citet{sugden1986economics} and \citet[Theorem 2]{kandori1992social}
show that cooperation can be a perfect equilibrium in a setup in which
each player observes a binary signal about his partner, either a ``good
label'' or a ``bad label.'' All players start with a good label.
This label becomes bad if a player defects against a ``good'' partner.
The equilibrium strategy that supports full cooperation in this setup
is to cooperate against good partners and defect against bad partners.
Theorems \ref{thm:only-defection-is-stable} and \ref{thmLrepeated-game}
reveal that the presence of a small fraction of committed agents does
not allow the population to maintain such a simple binary reputation
under an observation structure in which players observe an arbitrary
number of past actions taken by their partners. The theorem shows
this indirectly, because if it were possible to derive binary reputations
from this information structure, then it should have been possible
to support cooperation as a perfect equilibrium action. Moreover,
Theorem \ref{thm:stable-cooperation-observing-action-profiles} shows
that cooperation is not a perfect equilibrium action in acute games
when players observe action profiles. This suggests that the presence
of a few committed agents does not allow us to maintain the seemingly
simple binary reputation mechanisms of \citet{sugden1986economics}
and \citet{kandori1992social}, even under observation structures
in which each agent observes the whole action profile of many of her
opponent's past interactions. 

The mild restriction to a regular perfect equilibrium is necessary
for Theorem \ref{thm:only-defection-is-stable} to go through. Example
\ref{exa:non-regular-perfect-equilbirium-partial-ccoperation} in
Appendix \ref{sec:Example-non-regular-perfect-partial-cooperation}
demonstrates the existence of a non-regular perfect equilibrium of
an offensive PD, in which players cooperate with positive probability.
This non-robust equilibrium is similar to the ``belief-free'' sequential
equilibria that support cooperation in offensive Prisoner's Dilemma
games in \citet{takahashi2010community}, which have the property
that players are always indifferent between their actions, but they
choose different mixed actions depending on the signal they obtain
about the partner. 

\subsection{Stability of Cooperation in Defensive Prisoner\textquoteright s Dilemmas\label{subsec:Stability-of-Cooperation} }

Our next result shows that if players observe at least two actions,
then cooperation is a regular perfect equilibrium action in any defensive
Prisoner's Dilemma. Moreover, it shows that there is essentially a
unique combination of strategies that supports full cooperation in
the Prisoner's Dilemma game, according to which: (a) all agents cooperate
when observing no defections, (b) all agents defect when observing
at least 2 defections, (3) sometimes (but not always) agents defect
when observing a single defection. 
\begin{thm}
\label{thm:cooperation-defensive-PDs}Let $E=\left(G_{PD},k\right)$
be an environment with observations of actions, where $G_{PD}$ is
a defensive Prisoner's Dilemma ($g<l$), and $k\geq2$.

\begin{enumerate}
\item If $\left(S^{*},\sigma^{*},\theta^{*}\equiv0\right)$ is a perfect
equilibrium then: (a) for each $s\in S^{*}$, $s_{0}\left(c\right)=1$
and $s_{m}\left(d\right)=1$ for each $m\geq2$; and (b) there exist
$s,s'\in S^{*}$ such that $s_{1}\left(d\right)<1$ and $s'_{1}\left(d\right)>0$. 
\item Cooperation is a regular perfect equilibrium action. 
\end{enumerate}
\end{thm}
\begin{proof}[Sketch of Proof]
 Suppose that $\left(S^{*},\sigma^{*},\theta^{*}\equiv0\right)$
is a perfect equilibrium. The fact that the equilibrium induces full
cooperation, in the limit when the mass of commitment strategies converges
to zero, implies that all normal agents must cooperate when they observe
no defections, i.e., $s_{0}\left(c\right)=1$ for each $s\in S^{*}$. 

Next we show that there is a normal strategy that induces the agent
to defect with positive probability when observing a single defection,
i.e., $s_{1}\left(d\right)>0$ for some $s\in S^{*}$. Assume to the
contrary that $s_{1}\left(c\right)=1$ for each $s\in S^{*}$. If
an agent (Alice) deviates and defects with small probability $\epsilon<<1$
when observing no defections, then she outperforms the incumbents.
On the one hand, the fact that she occasionally defects when observing
$m=0$ gives her a direct gain of at least $\epsilon\cdot g$. On
the other hand, the probability that a partner observes her defecting
twice or more is $O\left(\epsilon^{2}\right)$; therefore her indirect
loss from these additional $\epsilon$ defections is at most $O\left(\epsilon^{2}\right)\cdot\left(1+l\right)$,
and therefore for a sufficiently small $\epsilon>0$, Alice strictly
outperforms the incumbents.

The fact that $s_{1}\left(d\right)>0$ for some $s\in S^{*}$ implies
that defection is a best reply conditional on an agent observing $m=1$.
The direct gain from defecting is strictly increasing in the probability
that the partner defects (because the game is defensive), while the
indirect influence of defection on the behavior of future partners
is independent of the partner's play. This implies that defection
must be the unique best reply when an agent observes $m\geq2$ , since
such an observation implies a higher probability that the partner
is going to defect relative to the observation of a single defection.
This establishes that $s_{m}\left(d\right)=1$ for all $m\geq2$ and
all $s\in S^{*}$.

In order to demonstrate that there is a strategy $s$ such that $s_{1}\left(d\right)<1$,
assume to the contrary that $s_{1}\left(d\right)=1$ for each $s\in S^{*}$.
Suppose that the average probability of defection in the population
is $0<\Pr\left(d\right)$. Since there is full cooperation in the
limit we have $\Pr\left(d\right)=O\left(\epsilon\right)$. This implies
that a random partner is observed to defect at least once with a probability
of $k\cdot\Pr\left(d\right)+O\left(\epsilon^{2}\right)$. This in
turn induces the defection of a fraction $k\cdot\Pr\left(d\right)+O\left(\epsilon^{2}\right)$
of the normal agents (under the assumption that $s_{1}\left(d\right)=1$).
Since the normal agents constitute a fraction $1-O\left(\epsilon\right)$
of the population we must have $\Pr\left(d\right)=k\cdot\Pr\left(d\right)+O\left(\epsilon^{2}\right)$,
which leads to a contradiction for any $k\geq2$. Thus, if $s_{1}\left(d\right)=1$,
then defections are ``contagious,'' and so there is no steady state
in which only a fraction $O\left(\epsilon\right)$ of the population
defects. This completes the sketch of the proof of part 1.

To prove part 2 of the theorem, let $s^{1}$ and $s^{2}$ be the strategies
that defect iff $m\geq1$ and $m\geq2$, respectively. Consider the
state $\left(\left\{ s^{1},s^{2}\right\} ,\left(q^{*},1-q^{*}\right),\theta^{*}\equiv0\right)$.
The direct gain from defecting (relative to cooperating) when observing
a single defection is
\[
\Pr\left(m=1\right)\cdot\left(\left(l\cdot\Pr\left(d|m=1\right)\right)+g\cdot\Pr\left(c|m=1\right)\right),
\]
where $\Pr\left(d|m=1\right)$ ($\Pr\left(c|m=1\right)$) is the probability
that a random partner is going to defect (cooperate) conditional on
the agent observing $m=1$, and $\Pr\left(m=1\right)$ is the average
probability of observing signal $m=1$. The indirect loss from defection,
relative to cooperation, conditional on the agent observing a single
defection, is
\[
q^{*}\cdot\left(k\cdot\Pr\left(m=1\right)\right)\cdot\left(l+1\right)+O\left(\left(\Pr\left(m=1\right)\right)^{2}\right).
\]
To see this, note that a random partner defects with an average probability
of $q$ if he observes a single defection (which occurs with probability
$k\cdot\Pr\left(m=1\right)$ when the partner makes $k$ i.i.d. observations,
each of which has a probability of $\Pr\left(m=1\right)$ of being
a defection), and each defection induces a loss of $l+1$ to the agent
(who obtains $-l$ instead of 1). The fact that some normal agents
cooperate and others defect when observing a single defection implies
that in an equilibrium both actions have to be best replies conditional
on the agent observing $m=1$. This implies that the indirect loss
from defecting is exactly equal to the direct gain (up to $O\left(\left(\Pr\left(m=1\right)\right)^{2}\right)$),
i.e., 
\[
\Pr\left(m=1\right)\cdot\left(\left(l\cdot\Pr\left(d|m=1\right)\right)+g\cdot\Pr\left(c|m=1\right)\right)=q^{*}\cdot\left(k\cdot\Pr\left(m=1\right)\right)\cdot\left(l+1\right)
\]
\begin{equation}
\Rightarrow q^{*}=\frac{\left(l\cdot\Pr\left(d|m=1\right)\right)+g\cdot\Pr\left(c|m=1\right)}{k\cdot\left(l+1\right)}.\label{eq:q-indifference-equation}
\end{equation}
The probability $\Pr\left(d|m=1\right)$ depends on the distribution
of commitments. Yet, one can show that for every distribution of commitment
strategies $\left(S^{C},\lambda\right)$, there is a unique value
of $q^{*}\in\left(0,\frac{1}{k}\right)$ that solves Eq. (\ref{eq:q-indifference-equation})
and that, given this $q^{*}$, both $s^{1}$ and $s^{2}$ (and only
these strategies) are best replies. This means that the steady state
$\left(\left\{ s^{1},s^{2}\right\} ,\left(q^{*},1-q^{*}\right),\theta^{*}\equiv0\right)$
is a perfect equilibrium. 
\end{proof}

\paragraph{Discussion of Theorem \ref{thm:cooperation-defensive-PDs} }

We comment on a few issues related to Theorem \ref{thm:cooperation-defensive-PDs}.
\begin{enumerate}
\item In the formal proof of Theorem \ref{thm:cooperation-defensive-PDs}
we show that cooperation satisfies the stronger refinements of strict
perfection, evolutionary stability, and robustness (see Section \ref{subsec:Stronger-Refinements-of}
and Appendix \ref{sec:Evolutionary-Stability}).
\item Each distribution of commitment strategies induces a unique frequency
$q^{*}\in\left(0,\frac{1}{k}\right)$ of $s^{1}$-agents, which yields
a perfect equilibrium. One may wonder whether a population starting
from a different share $q_{0}\neq q^{*}$ of $s^{1}$-agents is likely
to converge to the equilibrium frequency $q^{*}$. It is possible
to show that the answer is affirmative. Specifically, given any initial
low frequency $q_{0}\in\left(0,q^{*}\right)$, the $s^{1}$-agents
achieve a higher payoff than the $s^{2}$-agents and, given any initial
high frequency $q_{0}\in\left(q^{*},\frac{1}{k}\right)$, the $s^{1}$-agents
achieve a lower payoff than the $s^{2}$-agents. Thus, under any smooth
monotonic dynamic process in which a more successful strategy gradually
becomes more frequent, the share of $s^{1}$-agents will shift from
any initial value in the interval $q_{0}\in\left(0,\frac{1}{k}\right)$
to the exact value of $q^{*}$ that induces a perfect equilibrium.
\item As discussed in the formal proof in Appendix \ref{subsec:Proof-of-Theorem-defensive},
some distributions of commitment strategies may induce a slightly
different perfect equilibrium, in which the population is homogeneous,
and each agent in the population defects with probability $q^{*}\left(\mu\right)$
when observing a single defection (contrary to the heterogeneous deterministic
behavior described above). 
\item \emph{Random number of observed action}s. \label{enu:random}Consider
a \emph{random environment} $\left(G_{PD},p\right)$, where $p\in\Delta\left(\mathbb{N}\right)$
is a distribution with a finite support, and each agent privately
observes $k$ actions of the partner with probability $p\left(k\right)$.
Theorem \ref{thm:cooperation-defensive-PDs} (and, similarly, Theorems
\ref{thm:stable-cooperation-observing-conflicts}\textendash \ref{thmtertiary-observation})
can be extended to this setup for any random environment in which
the probability of observing at least two interactions is sufficiently
high. The perfect equilibrium has to be adapted as follows. As in
the main model, all normal agents cooperate (defect) when observing
no (at least two) defections. In addition, there will be a value $\bar{k}\in supp\left(p\right)$
and a probability $q\in\left[0,1\right]$ (which depend on the distribution
of commitment strategies), such that all normal agents cooperate (defect)
when observing a single defection out of $k>\bar{k}$ ($k<\bar{k}$),
and a fraction $q$ of the normal agents defect when observing a single
defection out of $\bar{k}$ observations.
\item \emph{Cheap talk. }In Appendix \ref{subsec:Cheap-Talk} we discuss
the influence on Theorems \ref{thm:only-defection-is-stable}\textendash \ref{thm:cooperation-defensive-PDs}
of the introduction of pre-play (slightly costly) cheap-talk communication.
In this setup one can show that:
\begin{enumerate}
\item Offensive games: No stable state exists. Both defection and cooperation
are only ``quasi-stable''{}'' the population state occasionally
changes between theses two states, based on the occurrence of rare
random experimentations. The argument is adapted from \citet{wiseman2001cooperation}.
\item Defensive games (and $k\geq2$): The introduction of cheap talk destabilizes
all inefficient equilibria, leaving cooperation as the unique stable
outcome. The argument is adapted from \citet{robson1990efficiency}.
\end{enumerate}
\item \emph{\label{enu:General-Noise-Structures:}General Noise Structures}:
In the model described above we deal with perturbed environments that
include a single kind of noise, namely, committed agents who follow
commitment strategies. It is possible to extend our results to include
additional sources of noise: specifically, observation noise and/or
trembles. We redefine a perturbed environment as a tuple $E_{\epsilon,\delta}=\left(\left(G,k\right),\left(S^{C},\lambda\right),\alpha,\epsilon,\delta\right)$,
where $\left(G,k\right),\left(S^{C},\lambda\right),\epsilon$ are
defined as in the main model, $0<\delta<<1$ is the probability of
error in each observed action of a player, and $\alpha\in\Delta\left(A\right)$
is a totally mixed distribution according to which the observed error
is sampled from in the event of an observation error. Alternatively,
these errors can also be interpreted as actions played by mistake
by the partner due to trembling hands. One can show that\emph{ all}
of our results can be adapted to this setup in a relatively straightforward
way. In particular, our results hold also in environments in which
most of the noise is due to observation errors, provided that there
is a small positive share of committed agents (possibly much smaller
than the probability of an observation error).
\item \label{enu:The-threshold-case}\emph{The borderline case between defensiveness
and offensiveness}: $g=l$. Such a Prisoner's Dilemma can be interpreted
as a game in which each of the players simultaneously decides whether
to sacrifice a personal payoff of $g$ in order to induce a gain of
$1+g$ to her partner. One can show that cooperation is also a perfect
equilibrium action in this setup, and that it can be supported by
the same kind of perfect equilibrium as described above. However,
in this case the uniqueness result (part 1 of Theorem \ref{thm:cooperation-defensive-PDs})
is no longer true. The reason for this is that when $g=l$ an agent
has the same incentive to defect regardless of the signal she observes
about the partner (because the direct bonus of defection is equal
to $g=l$ regardless of the partner's behavior). This implies that
cooperation can be supported by a large variety of strategies (including
belief-free-like strategies as in \citealp{takahashi2010community,dilme2016helping}).
We note that none of these strategies satisfy the refinement of evolutionary
stability (Appendix \ref{sec:Evolutionary-Stability}). One can adapt
the proof of Theorem \ref{thm:only-defection-is-stable} to show that
defection is the unique evolutionarily stable outcome when $g=l$.
\end{enumerate}
The following example demonstrates the existence of a perfect equilibrium
that supports cooperation when the unique commitment strategy is to
play each action uniformly.
\begin{example}[Example \ref{exa:revisiting-example} revisited: illustration of the
perfect equilibrium that supports cooperation]
 Consider the perturbed environment $\left(G_{D},2,\left\{ s^{u}\equiv0.5\right\} ,\epsilon\right)$,
where $G_{D}$ is the defensive Prisoner's Dilemma game with the parameters
$g=1$ and $l=3$ (as presented in Table \ref{tab:Prisoner-Dilemma-detailed}
above). Consider the steady state $\left(\left\{ s^{1},s^{2}\right\} ,\left(\frac{1}{6},\frac{5}{6}\right),\theta^{*}\right)$,
where $\theta^{*}$ is defined as in (\ref{eq:eta-star-example})
in Example \ref{exa:revisiting-example} above. A straightforward
calculation shows that the average probability in which a normal agent
observes $m=1$ when being matched with a random partner is
\[
\Pr\left(m=1\right)=\epsilon\cdot0.5+3.5\cdot\epsilon\cdot\frac{1}{6}+0.5\cdot\epsilon\cdot\frac{5}{6}+O\left(\epsilon^{2}\right)=1.5\cdot\epsilon+O\left(\epsilon^{2}\right).
\]
The probability that the partner is a committed agent conditional
on observing a single defection is: 
\[
\Pr\left(s^{u}|m=1\right)=\frac{\epsilon\cdot0.5}{1.5\cdot\epsilon}=\frac{1}{3}\,\,\Rightarrow\,\,\Pr\left(d|m=1\right)=\frac{1}{3}\cdot0.5=\frac{1}{6},
\]
which yields the conditional probability that the partner of a normal
agent will defect. Next we calculate the direct gain from defecting
conditional on the agent observing a single defection ($m=1$): 
\[
\Pr\left(m=1\right)\cdot\left(\left(l\cdot\Pr\left(d|m=1\right)\right)+g\cdot\Pr\left(c|m=1\right)\right)=1.5\cdot\epsilon\cdot\left(3\cdot\frac{1}{6}+1\cdot\frac{5}{6}\right)+O\left(\epsilon^{2}\right)=2\cdot\epsilon+O\left(\epsilon^{2}\right).
\]
The indirect loss from defecting conditional on the agent observing
a single defection is:
\[
q\cdot\left(k\cdot\Pr\left(m=1\right)\right)\cdot\left(l+1\right)+O\left(\epsilon^{2}\right)=q\cdot2\cdot1.5\cdot\epsilon\cdot\left(3+1\right)=12\cdot q\cdot\epsilon+O\left(\epsilon^{2}\right).
\]
When taking $q=\frac{1}{6}$ the indirect loss from defecting is exactly
equal to the direct gain (up to $O\left(\epsilon^{2}\right)$).
\end{example}

\paragraph{Stability of Cooperation when Observing a Single Action}

We conclude this section by showing that in defensive Prisoner's Dilemmas
with $k=1$, cooperation is a regular perfect equilibrium action iff
$g<1$.
\begin{prop}
\label{Pro-stability-of-cooperation-when-k=00003D1}Let $E=\left(G_{PD},1\right)$
be an environment where $G_{PD}$ is a defensive Prisoner's Dilemma
($g<l$). Cooperation is a (regular) perfect equilibrium action iff
$g<1$.
\end{prop}
\begin{proof}[Sketch of Proof]
Similar arguments to those presented in part 1 of Theorem \ref{thm:cooperation-defensive-PDs}
imply that any distribution of commitment strategies induces a unique
average probability $q$ by which normal agents defect when observing
$m=1$, in any cooperative perfect equilibrium. This implies that
a deviator who always defects gets a payoff of $1+g$ in a fraction
$1-q$ of the interactions. One can show that such a deviator outperforms
the incumbents if\footnote{In environments with $k\geq2$, a deviator who always defects gets
a payoff of zero, regardless of the value of $q$ (because all agents
observe $m=k$ when being matched with such a deviator).} $g>1$ (whereas, if $g<1$, there are distributions of commitment
for which $1-q$ is sufficiently low such that the deviator is outperformed).
\end{proof}
Proposition \ref{Pro-stability-of-cooperation-when-k=00003D1} is
immediately implied by Proposition \ref{prop:Let--be-single-action-complex-result}
(Appendix \ref{subsec:Stability-of-Cooperation-single-action}), which
characterizes which distributions of commitments support cooperation
as a perfect equilibrium outcome in a defensive Prisoner's Dilemma
when $k=1$.

\section{General Observation Structures\label{sec:General-Observation-Structures}}

In this section we extend our analysis to general observation structures
in which the signal about the partner may also depend on the behavior
of other opponents against the partner.

\subsection{Definitions}

An \emph{observation structure }is a tuple\emph{ $\Theta=\left(k,B,o\right)$,
}where $k\in\mathbb{N}$ is the number of observed interactions, $B=\left\{ b_{1},..,b_{\left|B\right|}\right\} $
is a finite set of \emph{observations} that can be made in each interaction,
and the mapping $o:A\times A\rightarrow B$ describes the observed
signal as a function the action profile played in the interaction
(where the first action is the one played by the current partner,
and the second action is the one played by his opponent). Note that
observing actions (which was analyzed in the previous section) is
equivalent to having $B=A$ and $o\left(a,a'\right)=a$. 

In the results of this section we focus on three observation structures:
\begin{enumerate}
\item \emph{Observation of action profiles: }$B=A^{2}$ and $o\left(a,a'\right)=\left(a,a'\right).$
In this observation structure, each agent observes, in each sampled
interaction of her partner, both the action played by her partner
and the action played by her partner's opponent.
\item \emph{Observation of conflicts}: observing whether or not there was
mutual cooperation. That is, $B=\left\{ C,D\right\} $, $o\left(c,c\right)=C$,
and $o\left(a,a'\right)=D$ for any $\left(a,a'\right)\neq\left(c,c\right)$.
Such an observation structure (which we have not seen in the existing
literature) seems like a plausible way to capture non-verifiable feedback
about the partner's behavior. The agent can observe, in each sampled
past interaction of the partner, whether both partners were ``happy''
(i.e., mutual cooperation) or whether the partners complained about
each other (i.e., there was a conflict, at least one of the players
defected, and it is too costly for an outside observer to verify who
actually defected). 
\item Observation of actions against cooperation: $B=\left\{ CC,DC,*D\right\} $
and $o\left(c,c\right)=CC$, $o\left(d,c\right)=DC$, and $o\left(c,d\right)=o\left(d,d\right)=*D$.
That is, each agent (Alice) observes a ternary signal about each sampled
interaction of her partner (Bob): either both players cooperated,
or Bob unilaterally defected, or Bob's partner defected (and in this
latter case Alice cannot observe Bob's action). We analyze this observation
structure because it turns out to be an ``optimal'' observation
structure that allows cooperation to be supported as a perfect equilibrium
action in any Prisoner\textquoteright s Dilemma.
\end{enumerate}
In each of these cases, we let the mapping $o$\textbf{ }and the set
of signals\textbf{ }$B$ be implied by the context, and identify the
observation structure $\Theta$ with the number of observed interactions
$k$. 

In what follows we present the definitions of the main model (Sections
\ref{sec:Model} and \ref{sec:Solution-Concept}) that have to be
changed to deal with the general observation structure. Before playing
the game, each player independently samples $k$ independent interactions
of her partner. Let $M$ denote the set of feasible signals:
\[
M=\left\{ m\in\mathbb{N}^{\left|B\right|}\left|\sum_{i}m_{i}=k\right.\right\} ,
\]
where $m_{i}$ is interpreted as the number of times that observation
$b_{i}$ has been observed in the sample. When agents observe conflicts,
we simplify the notation by letting $M=\left\{ 1,...,k\right\} $,
and interpreting $m\in\left\{ 1,...,k\right\} $ as the number of
observed conflicts.

The definitions of a strategy and a perturbed environment remain the
same. Given a distribution of action profiles $\psi\in\Delta\left(A\times A\right)$,
let $\nu_{\psi}=\nu\left(\psi\right)\in\Delta\left(M\right)$ be the
multinomial distribution of signals that is induced by the distribution
of action profiles $\psi$, i.e., 
\[
\nu_{\psi}\left(m_{1},...,m_{\left|B\right|}\right)=\frac{k!}{m_{1}!\cdot...\cdot m_{\left|B\right|}!}\cdot\prod_{i=1}^{\left|B\right|}\left(\sum_{\left\{ \left(a,a'\right)\in A\times A|o\left(a,a'\right)=b_{i}\right\} }\psi\left(a,a'\right)\right)^{m_{i}}.
\]

The definition of a steady state is adapted as follows.
\begin{defn}[Adaptation of Def. \ref{def:state-1}]
\label{def:state-1}A \emph{steady state }(or \emph{state}) of a
perturbed environment \emph{$\left(\left(G,k\right),\left(S^{C},\lambda\right),\epsilon\right)$}
is a triple $\left(S^{N},\sigma,\theta\right)$, where $S^{N}\subseteq\mathcal{S}$
is a finite set of strategies, $\sigma\in\Delta\left(S^{N}\right)$
is a distribution, and $\theta:\left(S^{N}\cup S^{C}\right)\rightarrow\Delta\left(M\right)$
is a profile of signal distributions that satisfies for each signal
$m$ and each strategy $s$ the consistency requirement (\ref{eq:consistency-1})
below. Let $\psi_{s}\in\Delta\left(A\times A\right)$ be the (possibly
correlated) distribution of action profiles that is played when an
agent with strategy $s\in S^{N}\cup S^{C}$ is matched with a random
partner (given $\sigma$ and $\theta$); i.e., for each $\left(a,a'\right)\in A\times A$,
where $a$ is interpreted as the action of the agent with strategy
$s$, and $a'$ is interpreted as the action of her partner, let 
\begin{equation}
\psi_{s}\left(a,a'\right)=\sum_{s'\in S^{N}\cup S^{C}}\left(\left(1-\epsilon\right)\cdot\sigma\left(s'\right)+\epsilon\cdot\lambda\left(s'\right)\right)\cdot s\left(\theta_{s'}\right)\left(a\right)\cdot s'\left(\theta_{s}\right)\left(a'\right).\label{eq:psi-s}
\end{equation}
The \emph{consistency requirement} that the mapping $\theta$ has
to satisfy is 
\begin{equation}
\forall m\in M,\,\,s\in S^{N}\cup S^{C},\,\,\,\,\theta_{s}\left(m\right)=\nu\left(\psi_{s}\right)\left(m\right).\label{eq:consistency-1}
\end{equation}
\end{defn}
The definition of the long-run payoff of an incumbent agent remains
unchanged. We now adapt the definition of the payoff of an agent (Alice)
who deviates and plays a non-incumbent strategy. Unlike in the basic
model, in this extension there might be multiple consistent outcomes
following Alice's deviation, as demonstrated in Example \ref{exa:multiple-consistent-outcomes}.
\begin{example}
\label{exa:multiple-consistent-outcomes}Consider an unperturbed environment
$\left(G_{PD},3\right)$ with an observation of $k=3$ action profiles.
Consider a homogeneous incumbent population in which all agents play
the following strategy: $s^{*}\left(m\right)=d$ if $m\,\textrm{includes at least 2 interactions with\,}\left(d,d\right),$
and $s^{*}\left(m\right)=c$ otherwise. Consider the state $\left(\left\{ s^{*}\right\} ,\theta^{*}=0\right)$
in which everyone cooperates. Consider a deviator (Alice) who follows
the strategy of always defecting. Then there exist three consistent
post-deviation steady states (in all of which the incumbents continue
to cooperate among themselves): (1) all the incumbents defect against
Alice, (2) all the incumbents cooperate against Alice, and (3) all
the incumbents defect against Alice with a probability of 50\%.
\end{example}
Formally, we define a consistent distribution of signals for a deviator
as follows.
\begin{defn}
Given steady state $\left(S^{N},\sigma,\theta\right)$ and non-incumbent
strategy $\hat{s}\in\mathcal{S}\backslash\left(S^{N}\cup S^{C}\right)$,
we say that a distribution of signals $\theta_{\hat{s}}\in\Delta\left(M\right)$
is \emph{consistent} if 
\[
\forall m\in M,\,\,\,\,\theta_{\hat{s}}\left(m\right)=\nu\left(\psi_{\hat{s}}\right)\left(m\right),
\]
 where $\psi_{s}\in\Delta\left(A\times A\right)$ is defined as in
(\ref{eq:psi-s}) above. Let $\Theta_{\hat{s}}\subseteq\Delta\left(M\right)$
be the set of all \emph{consistent signal distributions} of strategy
$\hat{s}$.
\end{defn}
Given steady state $\left(S,\sigma,\theta\right)$, non-incumbent
strategy $\hat{s}\in\mathcal{S}\backslash\left(S^{N}\cup S^{C}\right)$,
and consistent signal distribution $\theta\left(s\right)\equiv\theta_{\hat{s}}\in\Delta\left(M\right)$,
let $\pi_{\hat{s}}\left(S,\sigma,\theta|\theta_{\hat{s}}\right)$
denote the deviator's (long-run) payoff given that in the post-deviation
steady state the deviator's distribution of signals is $\theta_{\hat{s}}$.
Formally: 
\[
\pi_{\hat{s}}\left(S,\sigma,\theta|\theta_{\hat{s}}\right)=\sum_{s'\in S^{N}\cup S^{C}}\left(\left(1-\epsilon\right)\cdot\sigma\left(s'\right)+\epsilon\cdot\lambda\left(s'\right)\right)\cdot\left(\sum_{\left(a,a'\right)\in A\times A}\hat{s}_{\theta\left(s'\right)}\left(a\right)\cdot s'_{\theta\left(\hat{s}\right)}\left(a'\right)\cdot\pi\left(a,a'\right)\right).
\]

Let $\pi_{\hat{s}}\left(S,\sigma,\theta\right)$ be the maximal (long-run)
payoff for a deviator who follows strategy $\hat{s}$ in a post-deviation
steady state: 
\begin{equation}
\pi_{\hat{s}}\left(S,\sigma,\theta\right):=_{\theta_{\hat{s}}\in\Theta_{\hat{s}}}\max\pi_{\hat{s}}\left(S,\sigma,\theta|\theta_{\hat{s}}\right).\label{eq:deviator-payoff-extension}
\end{equation}

\begin{rem}
Our results remain the same if one replaces the maximum function in
(\ref{eq:deviator-payoff-extension}) with a minimum function.
\end{rem}

\subsection{Acute and Mild Prisoner's Dilemma\label{subsec:Acute-and-Mild} }

In this subsection we present a novel classification of Prisoner's
Dilemma games that plays an important role in the results of this
section. Recall that the parameter $g$ of a Prisoner's Dilemma game
may take any value in the interval $\left[0,l+1\right]$ (if $g>l+1$,
then mutual cooperation is no longer the efficient outcome that maximizes
the sum of payoffs). We say that a Prisoner's Dilemma game is \emph{acute}
if $g$ is in the upper half of this interval (i.e., if $g>\frac{l+1}{2}$),
and \emph{mild} if it's in the lower half (i.e., if $g<\frac{l+1}{2}$).
The threshold, $g=\frac{l+1}{2}$, is characterized by the fact that
the gain from a single unilateral defection is exactly half the loss
incurred by the partner who is the sole cooperator. Hence, unilateral
defection is \emph{mildly tempting} in mild games and \emph{acutely
tempting} in acute games. An interpretation of this threshold comes
from a setup (which will be important for our results) in which an
agent is deterred from unilaterally defecting because it induces future
partners to unilaterally defect against the agent with some probability.
Deterrence in acute games requires this probability of being punished
to be more than 50\%, while a probability of below 50\% is enough
for mild games. Figure \ref{fig:Classification-of-Prisoner's} illustrates
the classification of games into offensive/defensive and mild/acute.

\begin{figure}[h]
\caption{\label{fig:Classification-of-Prisoner's}Classification of Prisoner's
Dilemma Games}

\centering{}\includegraphics[scale=0.07]{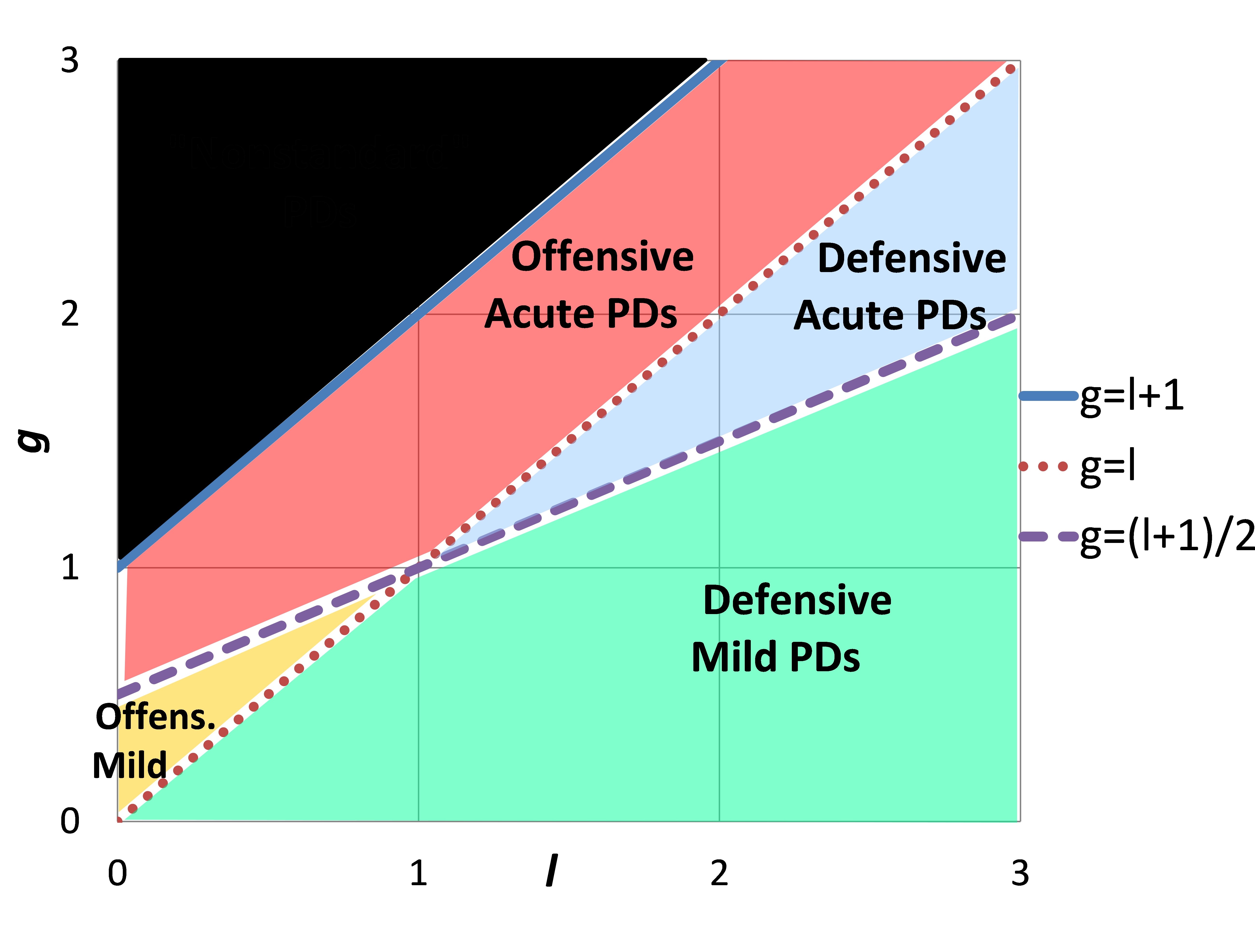}
\end{figure}

\begin{example}
Table \ref{tab:Prisoner-Dilemma-acute-mild} demonstrates the payoffs
of specific acute ($G_{A}$) and mild ($G_{M}$ ) Prisoner's Dilemma
games. In both examples $g=l$, i.e., the Prisoner's Dilemma game
is ``linear.'' This means that it can be described as a ``helping
game'' in which agents have to decide simultaneously whether to give
up a payoff of $g$ in order to create a benefit of $1+g$ for the
partner. In the acute game ($G_{A}$) on the left, $g=3$ and the
loss of a helping player amounts to more than half of of the benefit
to the partner who receives the help ($\frac{3}{3+1}=\frac{3}{4}>\frac{1}{2}$),
while in the mild game ($G_{M}$) on the right, $g=0.2$ and the loss
of the helping player is less than half of the benefit to the partner
who receives the help ($\frac{0.2}{0.2+1}=\frac{1}{6}<\frac{1}{2}$).
\end{example}
\begin{center}
\begin{table}[h]
\caption{\label{tab:Prisoner-Dilemma-acute-mild}Matrix Payoffs of Acute and
Mild Prisoner's Dilemma Games }

\centering{}%
\begin{tabular}{|c|c|c|}
\hline 
 & \emph{c} & \emph{d}\tabularnewline
\hline 
\emph{~~c~~} & \emph{\Large{}$_{\underset{\,}{1}}\,\,\,\,\,\,^{\overset{\,}{1}}$} & \emph{\Large{}$_{\underset{\,}{-l}}\,\,^{\overset{\,}{1+g}}$}\tabularnewline
\hline 
\emph{d} & \emph{\Large{}$_{\underset{\,}{1+g}}\,\,^{\overset{\,}{-l}}$} & \emph{\Large{}$_{0}\,\,\,\,\,\,^{0}$}\tabularnewline
\hline 
\multicolumn{3}{c}{General Prisoner\textquoteright s Dilemma }\tabularnewline
\multicolumn{3}{c}{$G_{PD}$: $g,l>0$ ,~ $g<l+1$}\tabularnewline
\end{tabular}~~~~~~~%
\begin{tabular}{|c|c|c|}
\hline 
 & \emph{c} & \emph{d}\tabularnewline
\hline 
\emph{~~c~~} & \emph{\Large{}$_{\underset{\,}{1}}\,\,\,\,\,\,^{\overset{\,}{1}}$} & \emph{\Large{}$_{\underset{\,}{-3}}\,\,^{\overset{\,}{4}}$}\tabularnewline
\hline 
\emph{d} & \emph{\Large{}$_{\underset{\,}{4}}\,\,^{\overset{\,}{-3}}$} & \emph{\Large{}$_{0}\,\,\,\,\,\,^{0}$}\tabularnewline
\hline 
\multicolumn{3}{c}{Ex. 3: Acute Prisoner\textquoteright s Dilemma}\tabularnewline
\multicolumn{3}{c}{$G_{A}$: $g=l=3>\frac{l+1}{2}=2$}\tabularnewline
\end{tabular}~~~~~~~%
\begin{tabular}{|c|c|c|}
\hline 
 & \emph{c} & \emph{d}\tabularnewline
\hline 
\emph{~~c~~} & \emph{\Large{}$_{\underset{\,}{1}}\,\,\,\,\,\,^{\overset{\,}{1}}$} & \emph{\Large{}$_{\underset{\,}{-0.2}}\,\,^{\overset{\,}{1.2}}$}\tabularnewline
\hline 
\emph{d} & \emph{\Large{}$_{\underset{\,}{1.2}}\,\,^{\overset{\,}{-0.2}}$} & \emph{\Large{}$_{0}\,\,\,\,\,\,^{0}$}\tabularnewline
\hline 
\multicolumn{3}{c}{Ex. 4: Mild Prisoner\textquoteright s Dilemma}\tabularnewline
\multicolumn{3}{c}{$G_{M}$: $g=l=0.2<\frac{l+1}{2}=0.6$}\tabularnewline
\end{tabular}
\end{table}
\par\end{center}

\subsection{Analysis of the Stability of Cooperation\label{subsec:Analysis-of-the-mild-acute}}

We first note that Proposition \ref{pro:defection-is-evol-stable}
is valid also in this extended setup, with minor adaptations to the
proof. Thus, always defecting is a perfect equilibrium regardless
of the observation structure. Next we analyze the stability of cooperation
in each of the three interesting observation structures.

The following two results show that under either \textbf{observation
of conflicts} or \textbf{observation of action profiles}, cooperation
is a perfect equilibrium iff the Prisoner's Dilemma is mild. Moreover,
in mild Prisoner's Dilemma games there is essentially a unique strategy
distribution that supports cooperation (which is analogous to the
essentially unique strategy distribution in Theorem \ref{thm:cooperation-defensive-PDs}).
Formally:
\begin{thm}
\label{thm:stable-cooperation-observing-conflicts}Let $E=\left(G_{PD},k\right)$
be an environment with observation of conflicts with $k\geq2$. 

\begin{enumerate}
\item If $G_{PD}$ is a mild PD ($g<\frac{l+1}{2}$), then:

\begin{enumerate}
\item If $\left(S^{*},\sigma^{*},\theta^{*}\equiv0\right)$ is a perfect
equilibrium then (1) for each $s\in S^{*}$, $s_{0}\left(c\right)=1$
and $s_{m}\left(d\right)=1$ for each $m\geq2$, and (2) there exist
$s,s'\in S^{*}$ such that $s_{1}\left(d\right)<1$ and $s'_{1}\left(d\right)>0$. 
\item Cooperation is a regular perfect equilibrium action.
\end{enumerate}
\item If $G_{PD}$ is an acute PD ($g>\frac{l+1}{2}$), then cooperation
is not a perfect equilibrium action.
\end{enumerate}
\end{thm}
\begin{proof}[Sketch of proof]
 The argument for part 1(a) is analogous to Theorem \ref{thm:cooperation-defensive-PDs}.
In what follows we sketch the proofs of part 1(b) and part 2. Fix
a distribution of commitments, and a commitment level $\epsilon\in\left(0,1\right)$.
Let $m$ denote the number of observed conflicts and define $s^{1}$
and $s^{2}$ as before, but with the new meaning of $m$. Consider
the following candidate for a perfect equilibrium $\left(\left\{ s^{1},s^{2}\right\} ,\left(q,1-q\right),\theta^{*}\equiv0\right)$.
Here, the probability $q$ will be determined such that both actions
are best replies when an agent observes a single conflict. That is,
the direct benefit from her defecting when observing $m=1$ (the LHS
of the equation below) must balance the indirect loss due to inducing
future partners who observe these conflicts to defect (the RHS, neglecting
terms of $O\left(\epsilon\right)$). The RHS is calculated by noting
that defection induces an additional conflict only if the current
partner has cooperated and that, on expectation, each such additional
conflict is observed by $k$ future partners, each of whom defects
with an average probability of $q$). Recall that $\Pr\left(d|m=1\right)$
($\Pr\left(c|m=1\right)$) is the probability that a random partner
is going to defect (cooperate) conditional on the agent observing
$m=1$.

\[
\Pr\left(m=1\right)\cdot\left(\left(l\cdot\Pr\left(d|m=1\right)\right)+g\cdot\Pr\left(c|m=1\right)\right)=\Pr\left(m=1\right)\cdot k\cdot q\cdot\Pr\left(c|m=1\right)\cdot\left(l+1\right)
\]
\begin{equation}
\Leftrightarrow q\cdot k=\frac{\left(l\cdot\Pr\left(d|m=1\right)\right)+g\cdot\Pr\left(c|m=1\right)}{\Pr\left(c|m=1\right)\cdot\left(l+1\right)}.\label{eq:q-k-conflict}
\end{equation}
One can see that the RHS is increasing in $\Pr\left(d|m=1\right)$.
The minimal bound on the value of $q$ is obtained when $\Pr\left(d|m=1\right)=0$.
In this case $q\cdot k=\frac{g}{l+1}$. 

Suppose that the game is acute. In this case $q\cdot k>0.5$. Suppose
that the average probability of defection in the population is $\Pr\left(d\right)$.
Since there is full cooperation in the limit we have $\Pr\left(d\right)=O\left(\epsilon\right)$.
This implies that a fraction $2\cdot\Pr\left(d\right)+O\left(\epsilon^{2}\right)$
of the population is involved in conflicts. This in turn induces the
defection of a fraction $2\cdot\Pr\left(d\right)\cdot k\cdot q+O\left(\epsilon^{2}\right)$
of the normal agents (because a normal agent defects with probability
$q$ upon observing at least one conflict in the $k$ sampled interactions).
Since the normal agents constitute a fraction $1-O\left(\epsilon\right)$
of the population we must have $\Pr\left(d\right)=2\cdot\Pr\left(d\right)\cdot k\cdot q+O\left(\epsilon^{2}\right)$.
However, in an acute game, $2\cdot k\cdot q>1$ leads to the contradiction
that $\Pr\left(d\right)<\Pr\left(d\right)$. Thus, if $2\cdot k\cdot q>1$,
then defections are contagious, and so there is no steady state in
which only a fraction $O\left(\epsilon\right)$ of the population
defects.

Suppose that the game is mild. One can show that $\Pr\left(d|m=1\right)$
is decreasing in $q$, and that it converges to zero when $k\cdot q\nearrow0.5$.
(The reason is that when $k\cdot q$ is close to 0.5 each defection
by a committed agent induces many defections by normal agents and,
conditional on observing $m=1$, the partner is likely to be normal
and to cooperate when being matched with a normal agent.) It follows
that the RHS of Eq. (\ref{eq:q-k-conflict}) is decreasing in $q$
and approaches the value $\frac{g}{l+1}$ when $k\cdot q\nearrow0.5$.
Since the game is mild, $\frac{g}{l+1}<0.5$. Hence there is some
$q\cdot k<0.5$ that solves Eq. (\ref{eq:q-k-conflict}), and in which
the normal agents defect with a low probability of ($O\left(\epsilon\right)$). 
\end{proof}

\begin{thm}
\label{thm:stable-cooperation-observing-action-profiles}Let $E=\left(G_{PD},k\right)$
be an environment with observation of action profiles and $k\geq2$. 
\begin{enumerate}
\item If $G_{PD}$ is a mild PD ($g<\frac{l+1}{2}$), then cooperation is
a regular perfect equilibrium action.
\item If $G_{PD}$ is an acute PD ($g>\frac{l+1}{2}$), then cooperation
is not a perfect equilibrium action.
\end{enumerate}
\end{thm}
\begin{proof}[Sketch of proof]
Using arguments that are familiar from above one can show that in
any perfect equilibrium that supports cooperation, normal agents have
to defect with an average probability of $q\in\left(0,1\right)$ when
observing a single unilateral defection (and $k-1$ mutual cooperations),
and defect with a smaller probability when observing a single mutual
defection (since this is necessary in order for a normal agent to
have better incentives to cooperate against a partner who is more
likely to cooperate). The value of $q$ is determined by Eq. (\ref{eq:q-k-conflict})
above, implying that both actions are best replies conditional on
an agent observing the partner to be the sole defector once, and to
be involved in mutual cooperation in the remaining $k-1$ observed
action profiles. Let $\epsilon$ be the share of committed agents,
and let $\varphi$ be the average probability that a committed agent
unilaterally defects. In order to simplify the sketch of the proof,
we will focus on the case in which the committed agents defect with
a small probability when observing the partner to have been involved
only in mutual cooperations, which implies, in particular, that $\varphi<<1$
(the formal proof in the Appendix does not make this simplifying assumption).
The unilateral defections of the committed agents induce a fraction
$\epsilon\cdot\varphi\cdot k\cdot q+O\left(\epsilon^{2}\right)+O\left(\varphi^{2}\right)$
of the normal agents to defect when being matched against committed
agents (because a normal agent defects with probability $q$ upon
observing a single unilateral defection in the $k$ sampled interactions).
These unilateral defections of normal agents against committed agents
induce a further $\left(\epsilon\cdot\varphi\cdot k\cdot q\right)\cdot k\cdot q+O\left(\epsilon^{2}\right)$
defections of normal agents against other normal agents. Repeating
this argument we come to the conclusion that the average probability
of a normal agent being the sole defector is (neglecting terms of
$O\left(\epsilon^{2}\right)$ and $O\left(\varphi^{2}\right)$): 
\[
\epsilon\cdot\varphi\cdot k\cdot q\cdot\left(1+k\cdot q+\left(k\cdot q\right)^{2}+...\right)=\epsilon\cdot\varphi\cdot\frac{k\cdot q}{1-k\cdot q.}
\]
As discussed above, in acute games, the value of $k\cdot q$ must
be larger than $0.5$, which implies that $\frac{k\cdot q}{1-k\cdot q}>1$.
This implies that conditional on an agent observing the partner to
be the sole defector once, the posterior probability that the partner
is normal is: 
\[
\frac{\epsilon\cdot\varphi\cdot\frac{k\cdot q}{1-k\cdot q}}{\epsilon\cdot\varphi+\epsilon\cdot\varphi\cdot\frac{k\cdot q}{1-k\cdot q}}=\frac{\frac{k\cdot q}{1-k\cdot q}}{1+\frac{k\cdot q}{1-k\cdot q}}>0.5.
\]
Thus, normal agents are more likely to unilaterally defect than committed
agents. One can show that when there is a mutual defection, it is
most likely that at least one of the agents involved is committed.
This implies that the partner is more likely to defect when he is
observed to be involved in mutual defection relative to being observed
to be the sole defector. This implies that defection is the unique
best reply when observing a single mutual defection, and this contradicts
the assumption that normal agents cooperate with positive probability
when observing a single mutual defection. When the game is mild, a
construction similar to the previous proofs supports cooperation as
a perfect equilibrium. 
\end{proof}
Our last result studies the observation of actions against cooperation,
and it shows that cooperation is a perfect equilibrium action in any
underlying Prisoner's Dilemma. Formally:
\begin{thm}
\label{thmtertiary-observation}Let $E=\left(G_{PD},k\right)$ be
an environment with observation of actions against cooperation and
$k\geq2$. Then cooperation is a regular perfect equilibrium action.
\end{thm}
The intuition behind the proof is as follows. Not allowing Alice to
observe Bob's behavior when his past opponent has defected helps to
sustain cooperation because it implies that defecting against a defector
does not have any negative indirect effect (in any steady state) because
it is never observed by future opponents. This encourages agents to
defect against partners who are more likely to defect, and allows
cooperation to be sustained regardless of the values of $g$ and $l$.

Table \ref{tab:Stability-of-Cooperation} summarizes our analysis
and shows the characterization of the conditions under which cooperation
can be sustained as a perfect equilibrium outcome in environments
in which agents observe at least 2 actions.

\begin{table}[h]
\caption{\label{tab:Stability-of-Cooperation}Summary of Key Results: When
Is Cooperation a Perfect Equilibrium Outcome?}
\begin{tabular}{|c|c|>{\centering}p{2.2cm}|>{\centering}p{2.2cm}|>{\centering}m{2.2cm}|>{\centering}p{2.4cm}|}
\hline 
\multirow{2}{*}{Category of PD} & \multirow{2}{*}{Parameters} & \multicolumn{4}{c|}{Observation Structure (any $k\geq2$)}\tabularnewline
\cline{3-6} \cline{4-6} \cline{5-6} \cline{6-6} 
 &  & Actions & Conflicts & Action profiles & Actions against cooperation\tabularnewline
\hline 
Mild \& Defensive & $\overset{\overset{\overset{\,}{\,}}{\,}}{\underset{}{g}}\underset{\,}{<}\min\left(l,\frac{l+1}{2}\right)$ & \textbf{\textcolor{green}{Y }} & \multirow{2}{2.2cm}{\textbf{\textcolor{green}{~~~~~$\begin{array}{c}
\\
\mathbf{Y}
\end{array}$ }}} & \multirow{2}{2.2cm}{\textbf{\textcolor{green}{~~~~~~$\begin{array}{c}
\\
\mathbf{Y}
\end{array}$ }}} & \multirow{4}{2.4cm}{\textbf{\textcolor{green}{~~~~~~$\begin{array}{c}
\begin{array}{c}
\\
\\
\\
\end{array}\\
\mathbf{Y}
\end{array}$ }}}\tabularnewline
\cline{1-3} \cline{2-3} \cline{3-3} 
Mild \& Offensive & $l<\overset{\overset{\overset{\,}{\,}}{\,}}{\underset{}{g}}<\frac{l+1}{2}$ & \textbf{\textcolor{red}{N }} &  &  & \tabularnewline
\cline{1-5} \cline{2-5} \cline{3-5} \cline{4-5} \cline{5-5} 
Acute \& Defensive & $\frac{l+1}{2}<\overset{\overset{\overset{\,}{\,}}{\,}}{\underset{}{g}}<l$  & \textbf{\textcolor{green}{Y }} & \multirow{2}{2.2cm}{\textbf{\textcolor{green}{~~~~~$\begin{array}{c}
\\
\mathbf{{\color{red}N}}
\end{array}$}}} & \multirow{2}{2.2cm}{\centering{}\textbf{\textcolor{green}{$\begin{array}{c}
\\
\mathbf{{\color{red}N}}
\end{array}$}}} & \tabularnewline
\cline{1-3} \cline{2-3} \cline{3-3} 
Acute \& Offensive & $\max\left(l,\frac{l+1}{2}\right)<\overset{\overset{\overset{\,}{\,}}{\,}}{\underset{}{g}}$ & \textbf{\textcolor{red}{N}}\textbf{ } &  &  & \tabularnewline
\hline 
\end{tabular}
\end{table}

\section{Related Literature\label{sec:Related-Literature}}

In what follows we discuss related literature that was not discussed
above. Related experimental literature is discussed in Appendix \ref{sec:Empirical-Predictions}.

\paragraph{Models with Rare Committed Types}

Various papers have shown that when a patient long-run agent (she)
plays a repeated game against partners who can observe her entire
history of play, and there is a small probability of the agent being
a commitment type, then the agent can guarantee herself a high payoff
in any equilibrium by mimicking an irrational type committed to Stackelberg-leader
behavior (e.g., \citealp{kreps1982rational,fudenberg1989reputation,celetani1996maintaining};
see \citealp{mailath2006repeated}, for a textbook analysis and survey).
When both sides of the game are equally patient, and, possibly, both
sides have a small probability of being a commitment type, then the
specific details about the set of feasible commitment types, the underlying
game, and the discount factor are important in determining whether
an agent can guarantee a high Stackelberg-leader payoff or whether
a folk theorem result holds and the set of equilibrium payoffs is
the same as in the case of complete information (see, e.g., \citealp{cripps1995reputation,chan2000non,cripps2005reputation,horner2009belief,Atakan2011,pkeski2014repeated}).

One contribution of our paper is to demonstrate that the introduction
of a small probability that an agent is committed may have qualitatively
different implications in repeated games with random matching.\footnote{We are aware of only one paper that introduces rare commitment types
to repeated games with random matching. \citet{dilme2016helping}
constructs cooperative ``tit-for-tat''-like equilibria that are
robust to the presence of committed agents, in the borderline case
in which $g=l$ is the underlying Prisoner's Dilemma (see the discussion
of the case of $g=l$ in Remark \ref{enu:The-threshold-case} in Section
\ref{subsec:Stability-of-Cooperation}). \citet{ghosh1996cooperation}
study a somewhat related setup (which is further discussed below)
in which the presence of a non-negligible share of agents who are
committed to always defecting allows cooperation to be sustained among
the normal agents in voluntarily separable interactions.} In defensive games, the presence of a few committed agents in the
population implies that there is a unique stationary strategy to sustain
full cooperation. In offensive games with observation of actions,
the presence of committed agents implies that the low payoff of zero
(of mutual defection) is the unique equilibrium payoff in the stationary
model (and it rules out the highest symmetric payoff of 1 in the conventional
model). \footnote{\citet{ely2008reputation} show a related result in a setup in which
a long-run player faces a sequence of short-run players. They show
that if the participation of the short-run players is optional, and
if every action of the long-run player that makes the short-run players
want to participate can be interpreted as a signal that the long-run
player is ``bad,'' then reputation uniquely chooses a low equilibrium
payoff to the long-run player. } 

\paragraph{Image Scoring}

In an influential paper, \citet{nowak1998evolution} present the mechanism
of \emph{image scoring} to support cooperation when agents from a
large community are randomly matched and each agent observes the partner's
past actions. In their setup, each agent observes the last $k$ past
actions of the partner, and she defects if and only if the partner
has defected at least $m$ times in the last $k$ observed actions.
A couple of papers have raised concerns about the stability of cooperation
under image-scoring mechanisms. Specifically, \citet{leimar2001evolution}
demonstrate in simulations that cooperation is unstable, and \citet{panchanathan2003tale}
analytically study the case in which each agent observes the last
action.\footnote{See \citet{berger2016stability} who study observation of $k$ actions,
but restrict agents to play only image-scoring-like strategies. } Our paper makes two key contributions to this literature. First,
we introduce a novel variant of image scoring that is essentially
the unique stationary way to support cooperation as a perfect equilibrium
outcome when agents observe actions. Second, we show that the classification
of Prisoner's Dilemma games into offensive and defensive games is
critical to the stability of cooperation when agents observe actions
(and image scoring fails in offensive Prisoner's Dilemma games).

\paragraph{Structured Populations and Voluntarily Separable Interactions}

A few papers have studied the scope of cooperation in the case where
players do not have any information about their current partner but
the matching of agents is not uniformly random. That is, the population
is assumed to have some structure such that some agents are more likely
to be matched to some partners than to other partners. \citet{van2012direct}
and \citet{alger2013homo} show that it is possible to sustain cooperation
with no information about the partner's behavior if matching is sufficiently
assortative, i.e., if cooperators are more likely to interact with
other cooperators. \citet{ghosh1996cooperation} and \citet{fujiwara2009voluntarily,fujiwara2017long}
show how to sustain cooperation in a related setup in which matching
is random, but each pair of matched agents may unanimously agree to
keep interacting without being rematched to other agents.\footnote{See also \citet{herold2012carrot} who studies a ``haystack'' model
in which individuals interact within separate groups.} Our paper shows that letting players observe the partner's behavior
in two interactions is sufficient to sustain cooperation without assuming
assortativity or repeated interactions with the same partner.

\paragraph{Models without Calendar Time.}

The present paper differs from most of the literature on community
enforcement by having a model without a global time zero. To the best
of our knowledge, \citet{rosenthal1979sequences} is the first paper
to present the notion of a steady-state Nash equilibrium in environments
in which each player observes the partner's last action, and apply
it to the study of the Prisoner's Dilemma. \citeauthor{rosenthal1979sequences}
focuses only on pure steady states (in which everyone uses the same
pure strategy), and concludes that defection is the unique pure stationary
Nash equilibrium action except in a few knife-edge cases. The methodology
is further developed in \citet{okuno1995social}. Other papers following
a related approach include \citet{Rubinstein-Wolinsky}, who study
bargaining, and \citet{Phelan_Skrzypacz_2006_private} who study repeated
games with private monitoring. Our methodological contribution to
the previous literature is that (1) we allow each agent to observe
the behavior of the partner in several past interactions with other
opponents, and (2) we combine the steady-state analysis with the presence
of a few committed agents and present a novel notion of a perfect
equilibrium to analyze this setup. 

\section{Conclusion\label{subsec:Conclusion-and-Directions}}

In many situations people engage in short-term interactions where
they are tempted to behave opportunistically but there is a possibility
that future partners will obtain some information about their behavior
today. We propose a new modeling approach based on the premises that
(1) an equilibrium has to be robust to the presence of a few committed
agents, and (2) the community has been interacting from time immemorial
(though this latter assumption is relaxed in Appendix \ref{sec:Conventional-Repeated-Game}). 

We develop a novel methodology that allows for a tractable analysis
of these seemingly complicated environments. We apply this methodology
to the study of Prisoner\textquoteright s Dilemma games, and we obtain
sharp testable predictions for the equilibrium outcomes, and the exact
conditions under which cooperation can be sustained as an equilibrium
outcome. Finally, we show that whenever cooperation is sustainable,
there is a unique (and novel) way to support it that has a few appealing
properties: (1) agents behave in an intuitive and simple way, and
(2) the equilibrium is robust, e.g., to deviations by a group of agents,
or to the presence of any kind of committed agents. We believe that
our modeling approach will be helpful in understanding various interactions
in future research. 

\bibliographystyle{econometrica}
\bibliography{SOB}

\appendix
\newpage{}

\clearpage \pagenumbering{arabic}

\section*{Online Appendices}

\section{Conventional Repeated Game Model\label{sec:Conventional-Repeated-Game}}

The main model of the paper relies on various simplifying assumptions,
and some unconventional modeling choices that distinguish it from
the existing literature: (1) the interactions within the community
do not have an explicit starting point, (2) agents live forever and
do not discount the future, (3) agents are only allowed to follow
stationary strategies, (4) agents (privately) observe the partner's
actions sampled from the entire infinite history of play of the partner.
In this Appendix we present a conventional repeated game model that
relaxes all of these assumptions. It differs from most of the existing
literature in only one respect: the presence of a small fraction of
committed agents in the population. We show that this difference is
sufficient to yield most of our key results. For brevity, we focus
only on the observations of actions. The adaptation of the results
on general observation structures is analogous. 

\subsection{Adaptations to the Model}

\paragraph{Environment as a Repeated Game}

We consider an infinite population (a continuum of mass one) interacting
in discrete time $t=0,1,2,3,...$. We redefine an\emph{ environment}
to be a triple $\left(G,k,\delta\right)$, where $G=\left(A,\pi\right)$
is the underlying symmetric game, $k\in\mathbb{N}$ is the number
of recent actions of an agent that are observed by her partner, and
$\delta\in\left(0,1\right)$ is the discount factor of the agents.
In each period the agents are randomly matched into pairs and, before
playing, each agent observes the most recent $\min\left(k,t\right)$
actions of her partner; i.e., an agent observes all past actions in
the early rounds when $t\leq k$, and she observes only the last $k$
rounds in later rounds when $t>k$. Let $M=\cup_{i\leq k}A^{i}$ denote
the set of all possible signals. 
\begin{rem}
Our results can be adapted to a more general setup in which each agent
observes $k$ actions randomly sampled from the partner's last $n\geq k$
actions. The case of $n>>k$ is the one closest to the main model.
We choose to focus on the opposite case of $n=k$ (i.e., observation
of the last $k$ actions) in order to demonstrate the robustness of
our results in the setup that is the ``furthest'' from the main
model. 
\end{rem}
A (private) history of an agent at round $\hat{t}$ is a tuple $h_{\hat{t}}=\left(\left(m_{t},a_{t},b_{t}\right)_{0\leq t<\hat{t}},m_{\hat{t}}\right)$,
where $m_{t}\in M$ is the signal observed by the agent about her
opponent at round $t$, $a_{t}\in A$ is the action played by the
agent at round $t$, and $b_{t}\in A$ is the action played by the
past partner at round $t$. Finally, $m_{\hat{t}}$ is the signal
the agent has observed about her current partner. Let $H_{\hat{t}}$
denote the set of all possible histories at round $\hat{t}$, and
let $H=\cup_{T\in\mathbb{N}}H_{T}$ denote the set of all (finite)
histories.

A\emph{ strategy} is a mapping $s:H\rightarrow\Delta\left(A\right)$
assigning a mixed action to each (private) history. We redefine $S$
to denote the set of all such strategies. Note, that unlike in the
main model we do not impose any restrictions on the set of feasible
strategies. In particular, we allow agents to follow non-stationary
strategies. A strategy is \emph{uniformly totally mixed} if there
exist $\gamma>0$ such that for each history $h_{\hat{t}}\in H$ and
each action $a\in A$ it is the case that $s_{h_{\hat{t}}}\left(a\right)>\gamma$.

\paragraph{Perturbed Environment and Population State}

A perturbed environment is a tuple consisting of: (1) an environment,
(2) a distribution $\lambda$ over a set of commitment strategies
$S^{C}$ that includes a uniformly totally mixed strategy, and (3)
a number $\epsilon$ representing how many agents are committed to
playing strategies in $S^{C}$ (\emph{committed agents}). The remaining
$1-\epsilon$ agents can play any strategy in $S^{N}$ (\emph{normal
agents}). Formally:
\begin{defn}
\label{def:A-perturbed-environment-1-1}A \emph{perturbed} \emph{environment}
is a tuple $E_{\epsilon}=\left(\left(G,k,\delta\right),\left(S^{C},\lambda\right),\epsilon\right)$,
where $\left(G,k,\delta\right)$ is an environment, $S^{C}$ is a
non-empty finite set of strategies (called \emph{commitment strategies})
that includes a uniformly totally mixed strategy, $\lambda\in\Delta\left(S^{C}\right)$
is a distribution with full support over the commitment strategies,
and $\epsilon\geq0$ is the mass of committed agents in the population.

A \emph{population state} is defined as a pair $\left(S^{N},\sigma\right)$,
where $S^{N}$ is the finite set of normal strategies in the population,
and $\sigma\in\Delta\left(S^{N}\right)$ is the distribution describing
the frequency of each normal strategy in the population of normal
agents. By standard arguments, a population state $\left(S^{N},\sigma\right)$
and a perturbed environment $E_{\epsilon}$ jointly induce a unique
sequence of distributions over the set of histories. Formally, there
exists a unique profile $\left(\mu_{s,t}\right)_{s\in S,t\in\mathbb{N}}$,
where each $\mu_{s,t}\in\Delta\left(H_{T}\right)$ is a distribution
over the histories of length $t$, such that $\mu_{s,t}\left(h_{t}\right)$
is the probability that an agent who follows strategy $s$ reaches
history $h_{t}\in H_{t}$ in round $t$.
\end{defn}

\paragraph{Expected Payoff and Equilibria }

In what follows we define the (ex-ante) expected payoff of an agent
who follows strategy $s$ and has discount factor $\delta$ , given
a population state $\left(S^{N},\sigma\right)$ of a perturbed environment
\emph{$E_{\epsilon}=\left(\left(G,k,\delta\right),\left(S^{C},\lambda\right),\epsilon\right)$}.
When $s\in S^{N}\cup S^{C}$ is an incumbent strategy, we define the
payoff as follows:
\begin{eqnarray}
\pi_{s}\left(S^{N},\sigma,E_{\epsilon}\right) & = & \left(1-\delta\right)\cdot\sum_{t\geq1}\delta^{t-1}\cdot\sum_{h_{t}=\left(\left(m_{t},a_{t},b_{t}\right)_{0\leq t<\hat{t}},m_{\hat{t}}\right)\in H^{t}}\mu_{s,t}\left(h_{t}\right)\cdot\pi\left(a_{t-1},b_{t-1}\right).\label{eq:pi_repeated_Game}
\end{eqnarray}

As in the stationary model, let $\pi\left(S^{N},\sigma,E_{\epsilon}\right)=\sum_{s\in S^{N}}\sigma\left(s\right)\cdot\pi_{s}\left(S^{N},\sigma,E_{\epsilon}\right)$
denote the mean payoff of the normal agents in the population. 

Next consider an agent (Alice) who deviates and plays a new strategy
$\hat{s}\in\mathcal{S}\backslash S^{N}$. Alice's strategy determines
her behavior against the incumbents. This determines the distribution
of signals that are observed by the partners when being matched with
Alice, and thus it determines the incumbents' play against Alice,
and this uniquely determines the sequence of distributions over the
set of histories of Alice. Formally, there exists a unique profile
$\left(\mu_{\hat{s},t}\right)_{t\in\mathbb{N}}$, where each $\mu_{\hat{s},t}\in\Delta\left(H_{T}\right)$
is a distribution over the histories of length $t$, such that $\mu_{\hat{s},t}\left(h_{t}\right)$
is the probability that Alice who follows strategy $\hat{s}$ reaches
history $h_{t}\in H_{t}$ in round $t$. We define Alice's payoff
$\pi_{\hat{s}}\left(S^{N},\sigma,E_{\epsilon}\right)$ in the same
way as in Eq. (\ref{eq:pi_repeated_Game}), with $\mu_{\hat{s},t}\left(h_{t}\right)$
replacing $\mu_{s,t}\left(h_{t}\right)$.

The definition of Nash equilibrium is standard:
\begin{defn}
A population state $\left(S^{N},\sigma\right)$ of the perturbed environment
$\left(\left(G,k,\delta\right),\left(S^{C},\lambda\right),\epsilon\right)$
is a \emph{Nash equilibrium} if for each strategy $s\in\mathcal{S}$,
it is the case that $\pi_{s}\left(S^{N},\sigma,E_{\epsilon}\right)\leq\pi\left(S^{N},\sigma,E_{\epsilon}\right)$.
\end{defn}
\begin{defn}
Fix an environment $\left(G,k,\delta\right)$. A sequence of strategies
$\left(s_{n}\right)_{n}$ converges to strategy $s$ (denoted by $\left(s_{n}\right)_{n}\rightarrow_{n\rightarrow\infty}s$)
if for each round $t\in\mathbb{N}$, each history $h_{t}\in H_{t}$,
and each action $a$, the sequence of probabilities $s\left(h_{t}\right)\left(a\right)$
converges to $s\left(h_{t}\right)\left(a\right).$ A sequence of population
states $\left(S^{N}{}_{n},\sigma_{n}\right)_{n}$ converges to a population
state $\left(S^{N},\sigma^{*}\right)$ if for each strategy $s\in supp\left(\sigma^{*}\right),$
there exists a sequence of sets of strategies $\left(S_{n}^{N}\right)_{n}$
such that: (1) $\sum_{s_{n}\in S_{n}^{N}}\sigma_{n}\left(s_{n}\right)\rightarrow\sigma^{*}\left(s\right)$,
and (2) for each sequence of elements of those sets (i.e., for each
sequence of strategies $\left(s_{n}\right)_{n}$ such that $s_{n}\in S_{n}^{N}$
for each $n$), $s_{n}\rightarrow_{n\rightarrow\infty}s$ .
\end{defn}
A perfect equilibrium is defined as the limit of a converging sequence
of Nash equilibria of a converging sequence of perturbed environments.
Formally:
\begin{defn}
A population state $\left(S^{N},\sigma^{*}\right)$ of the environment
$\left(G,k,\delta\right)$ is a\emph{ perfect equilibrium }if there
exist a distribution of commitments $\left(S^{C},\lambda\right)$
and converging sequences $\left(S_{n}^{N},\sigma_{n}\right)_{n}\rightarrow_{n\rightarrow\infty}\left(S^{N},\sigma^{*}\right)$
and $\left(\epsilon_{n}>0\right)_{n}\rightarrow_{n\rightarrow\infty}0$,
such that for each $n$, the state $\left(S_{n}^{N},\sigma_{n}\right)$
is a Nash equilibrium of the perturbed environment $\left(\left(G,k,\delta\right),\left(S^{C},\lambda\right),\epsilon_{n}\right)$.
If the underlying game is the Prisoner's Dilemma, we say that the
\emph{perfect equilibrium induces full cooperation} if $\lim_{n\rightarrow\infty}\pi\left(S_{n}^{N},\sigma,E_{\epsilon_{n}}\right)=\pi\left(c,c\right)$.

We say that cooperation is a perfect equilibrium outcome if there
exists a perfect equilibrium that induces full cooperation.
\end{defn}

\subsection{Adaptation of Main Results}

The following result adapts the main results of Section \ref{sec:Main-Results}.
Specifically, it shows that full cooperation is a perfect equilibrium
outcome iff the underlying Prisoner's Dilemma game is (weakly) defensive.
Moreover, we construct a perfect equilibrium that sustains full cooperation
and has the same qualitative properties as the strategy presented
in the stationary model. The intuition for the result is similar to
the intuition described in connection with the results of the stationary
model.
\begin{thm}
\label{thmLrepeated-game}Let $\left(G_{PD},k,\delta\right)$ be an
environment. 
\begin{enumerate}
\item Cooperation is not a perfect equilibrium outcome if $g>l$.
\item Cooperation is a perfect equilibrium outcome if $k\geq2$, $g\leq l$,
and $\delta>\frac{l}{l+1}$. Moreover, cooperation is sustained by
a strategy in which each normal agent (1) always cooperates if she
observes the partner always cooperating, (2) always defects if she
observes the partner defecting at least twice, and (3) sometimes defects
if she observes the partner defecting once.
\end{enumerate}
\end{thm}

\subsection{\label{subsec:Discussion-of-the-results-repeated-games}Discussion
of the Results in the Setup of Repeated Games}

Theorem \ref{thmLrepeated-game} adapts our main results from the
stationary model (Theorems \ref{thm:only-defection-is-stable} and
\ref{thm:cooperation-defensive-PDs}) to the conventional setup of
repeated games.\footnote{Similarly, one can adapt Corollary \ref{Pro-stability-of-cooperation-when-k=00003D1}
to the setup of repeated games, to show that when the underlying game
is defensive and each agent observes only the partner's last action,
there is a threshold $\bar{g}_{\delta}$ that depends on the discount
factor $\delta$, such that cooperation can be supported as a perfect
equilibrium outcome iff $g<\bar{g}_{\delta}$, and this threshold
converges to one as the discount factor converges to one, i.e., $lim_{\delta\rightarrow1}\bar{g}_{\delta}=1$. } The adaptation weakens our main results in three aspects:
\begin{enumerate}
\item While Theorem \ref{thm:only-defection-is-stable} shows that no level
of partial cooperation is sustainable in stationary environments,
Theorem \ref{thmLrepeated-game} merely shows that full cooperation
is not sustainable. The reason for this is as follows. In stationary
environments, if the partner has been observed to defect more often
in the past, it implies that he is more likely to defect in the current
match. Such an inference is not always valid in a non-stationary environment
in which an agent may condition his behavior on his own recent history
of play. In particular, we conjecture that partial cooperation may
be sustained in offensive games by a strategy according to which normal
agents sometimes cooperate, and an agent is more likely to cooperate
if (1) the agent has recently defected more often, and (2) if the
partner has recently cooperated more often.
\item While Theorem \ref{thm:cooperation-defensive-PDs} shows that there
is essentially a unique way to support full cooperation, Theorem \ref{thmLrepeated-game}
shows only that a very similar mechanism can also be used to support
full cooperation in standard repeated games. The fact that we allow
non-stationary strategies and that observed actions are ordered induces
a larger set of strategies, and does not allow us to show a similar
uniqueness property in this setup. We conjecture that some qualitative
properties of the unique stationary equilibrium hold in any equilibrium
sustaining full cooperation: (1) normal agents always cooperate after
observing no defections, (2) usually (though not necessarily in all
rounds) the average probability that a normal agent defects conditional
on observing a single defection of the partner is relatively low (less
than $\frac{1}{k}$), and (3) the average probability that a normal
agent defects conditional on observing many defections of the partner
is relatively high.
\item Theorem \ref{thm:cooperation-defensive-PDs} shows that cooperation
is a \emph{strictly} perfect equilibrium outcome; i.e., cooperation
can be sustained regardless of the behavior of the committed agents
(as formally defined in Appendix \ref{subsec:Strictly-Perfect}).
In this setup, as there is a much larger set of non-stationary strategies
that may be used by committed agents, we are not able to show a similar
strictness property. Specifically, we conjecture that full cooperation
cannot be sustained in a perturbed environment of an underlying defensive
game in which each committed agent defects with a high probability
if he has defected at most once in the last $k$ rounds, and he defects
with a low probability if he has defected at least twice in the last
$k$ rounds. This is so because in such environments normal agents
are incentivized to cooperate when observing the partner to defect
in all recent $k$ rounds, but this implies that a deviator who always
defects will outperform the incumbents. 
\end{enumerate}
\begin{rem}[Comparison with \citealp{takahashi2010community}]
\label{rem:Takahshi-repeated} The setup in this appendix is almost
identical to the setup of \citet{takahashi2010community}. The only
key difference between the two models is that we introduce a few committed
agents into the population (in addition, \citeauthor{takahashi2010community}
assumes that an agent observes all past actions of the partner, but
one can adapt his results to a setup in which an agent observes only
the recent $k$ actions of the partner). \citet[Prop. 2]{takahashi2010community}
constructs ``belief-free'' equilibria in which (1) each agent is
indifferent between the two actions after any history, and (2) each
agent chooses actions independent of her own record of past play.
\citeauthor{takahashi2010community} shows how these equilibria can
support any level of cooperation, and, in particular, can support
full cooperation in any Prisoner's Dilemma. \\
We show that the presence of a few committed agents substantially
changes this result when $g\neq l$ . When committed agents are present,
an agent can no longer be indifferent between the two actions after
all histories of play, and can no longer play in the current match
independently of her own record of past play.\footnote{\citet{heller2015instability} presents a related non-robustness
argument in the setup of repeated games played by the \emph{same}
two players, and shows that none of the belief-free equilibria are
robust against small perturbations in the behavior of potential opponents
(i.e., none of them satisfy a mild refinement in the spirit of evolutionary
stability).} We adapt \citeauthor{takahashi2010community}'s construction and
present an equilibrium in which each agent is indifferent between
the two actions after only one class of histories: that in which the
agent has cooperated in the previous $k-1$ rounds and she observes
the partner to defect only in the last round. In all other classes
of histories, the agents have strict incentives to either cooperate
or defect. 
\end{rem}

\section{Empirical Predictions and Experimental Verification \label{sec:Empirical-Predictions}}

In this appendix we discuss a few testable empirical predictions of
our model, comment on how to evaluate these predictions in lab experiments,
and discuss related experimental literature.

An experimental setup to evaluate our predictions would include a
large group of subjects (say, at least 10) who play a large number
of rounds (say, in expectation at least 50 rounds), and are rematched
in each period to play a Prisoner\textquoteright s Dilemma game with
a new partner. The experiment would include various treatments that
differ in terms of (1) the parameters of the underlying game, e.g.,
whether the game is offensive/defensive and mild/acute, and (2) the
information each agent observes about her partner: in particular,
the number of past interactions that each agent observes, and what
she observes in each interaction (e.g., actions, conflicts, or action
profiles.)

Our theoretical predictions deal with a ``pure'' setup in which
all agents maximize their material payoffs except for a vanishingly
small number of committed agents. An experimental setup (and, arguably,
real-life interactions) differs in at least two key respects: (1)
agents, while caring about their material payoffs, may consider other
non-material aspects, such as fairness and reciprocity, and (2) agents
occasionally make mistakes and the frequency of these mistakes, while
relatively low, is not negligible. In what follows, we describe our
key predictions in the ``pure'' setup, interpret its implications
in a ``noisy'' experimental setup, and describe the relevant existing
data. 

Our first prediction (Theorems \ref{thm:only-defection-is-stable}
and \ref{thm:cooperation-defensive-PDs}) deals with observation of
the partner's actions, and it states that cooperation can be sustained
only in defensive games. In an experimental setup we interpret this
to imply that, \emph{ceteris paribus}, the frequency of cooperation
will be higher in a defensive game than in an offensive game. \citet{engelmann2009indirect},
\citet{molleman2013personal}, and \citet{swakman2015reputation}
study the rate of cooperation in the borderline case of $g=l$ and
in the closely related donor-recipient game, in which at each interaction
only one of the players (the donor) chooses whether to give up $g$
of her own payoff to yield a gain of $1+g$ for the recipient. The
typical findings in these experiments are that observation of 3\textendash 6
past actions induces a relatively high level of cooperation (50\%\textendash 75\%,
where higher rates of cooperation are typically associated with environments
in which more past actions are observed, and environments in which
subjects can also observe second-order information about the behavior
of the partner's past opponent). We are aware of only a single experiment
that studies a setup in which $g\neq l$. \citet{gong2010reputation}
study the case of $g=0.8>l=0.4$, and present results that seem to
be consistent with our prediction. They observe an average rate of
cooperation of only 30\%\textendash 50\%, even though in their setup
players observe 10 past actions of the partner, and, in addition,
are also able to observe the signal observed by the partner in each
of these past interactions (``second-order information,'' which
facilitates cooperation relative to the model analyzed in this paper).

Our second prediction (Theorems \ref{thm:stable-cooperation-observing-conflicts}
and \ref{thm:stable-cooperation-observing-action-profiles}) deals
with observation of either past conflicts or past action profiles,
and it states that cooperation can be sustained only in mild games.
In an experimental setup it implies that, \emph{ceteris paribus},
the frequency of cooperation will be higher in mild games than in
acute games. We are unaware of any existing experimental data with
observation of either action profiles or conflicts.

It is interesting to compare our first two predictions to the comparative
statics recently developed for repeated Prisoner's Dilemma games played
by the same pair of players. \citet{blonski2011equilibrium}, \citet{bo2011evolution},
and \citet{breitmoser2013cooperation} present theoretical arguments
and experimental data to suggest that when a pair of players repeatedly
play the Prisoner\textquoteright s Dilemma, then the lower the values
of $g$ and $l$ are, the easier it is to sustain cooperation.\footnote{Specifically, the above papers show that cooperation is more likely
to be sustained in the infinitely repeated Prisoner's Dilemma if the
discount factor of the players is above $\frac{g+l}{g+l+1}$. Note
that this minimal threshold for cooperation is increasing in both
parameters. \citet{embrey2015cooperation} present similar comparative
statics evidence for the finitely repeated Prisoner's Dilemma.} However, our prediction is that when agents are randomly matched
in each round, then the lower the value of $g$, and the\emph{ higher}
the value of $l$, the easier it is to sustain cooperation. 

Our final prediction is that when communities succeed in sustaining
cooperation, it will be supported by the following behavior: most
subjects defect (resp., cooperate, mix) when observing 2+ (resp.,
0, 1) defections/conflicts. In an experimental setup we interpret
this to predict that the probability that an agent defects increases
with the number of times she observes the partner to be involved in
defections/conflict. In particular, we predict a substantial increase
in a subject's propensity to defect when moving from zero to two observations
of defection. The findings of \citet{engelmann2009indirect}, \citet{molleman2013personal},
\citet{gong2010reputation}, and \citet{swakman2015reputation} suggest
that subjects are indeed more likely to defect when they observe the
partner to defect more often in the past. 

\section{Technical Definitions and Additional Results\label{sec:Additional-Notation-and} }

The appendix presents technical definitions, which we have omitted
from the model in Sections \ref{sec:Model} and \ref{sec:Solution-Concept}
for expositional reasons. We also state simple results about the implementation
of (``trembling-hand'' perfect) Nash equilibria as (perfect) Nash
equilibria in our setup.

\subsection{Dynamic Mapping Between Signal Profiles\label{subsec:Dynamic-Mapping-Between}}

In this subsection we present a definition of a dynamic mapping between
signal profiles, which is useful in the definitions of the equilibrium
refinements in Appendix \ref{sec:Evolutionary-Stability} and in various
proofs in Appendix \ref{sec:Proofs}, and we prove a simple result,
namely, that any distribution of strategies admits a consistent signal
profile. 

Let $f_{\sigma}:O_{S}\rightarrow O_{S}$ be the \emph{mapping between
signal profiles} that is induced by $\sigma$. That is, $f_{\sigma}\left(\theta\right)$
is the ``new'' signal profile that is induced by players who follow
strategy distribution $\sigma$, and who observe signals about the
partners according to the ``current'' signal profile $\theta$.
Specifically, when Alice, who follows strategy $s$, is being matched
with a random partner whose strategy is sampled according to $\sigma$,
she observes a random signal according to the ``current'' average
distribution of signals in the population $\theta_{\sigma}$. As a
result her distribution of actions is $s\left(\theta_{\sigma}\right)$,
and thus her behavior induces the signal distribution $\nu\left(s\left(\theta_{\sigma}\right)\right)$.
Thus, we define this latter expression as her ``new'' distribution
of signals $\left(f_{\sigma}\left(\theta\right)\right)_{s}$. Formally:

\begin{equation}
\forall m\in M,\,\,s\in S,\,\,\,\,\left(f_{\sigma}\left(\theta\right)\right)_{s}\left(m\right)=\nu\left(s\left(\theta_{\sigma}\right)\right)\left(m\right).\label{eq:transformation-between-signal=00003Dprofiles-3}
\end{equation}

Observe that a signal profile $\theta:S\rightarrow\Delta\left(M\right)$
is c\emph{onsistent} with distribution of strategies $\sigma$ (as
defined in eq. (\ref{eq:transformation-between-signal=00003Dprofiles-2}))
if it is a fixed point of the mapping $f_{\sigma}\left(\theta\right)$,
i.e., if $f_{\sigma}\left(\theta\right)=\theta$. A standard fixed-point
argument shows that any distribution of strategies admits a consistent
signal profile.
\begin{lem}
\label{lem:existance_consitent_outcomes-1-1}Let $S$ be a finite
set of strategies and let $\sigma\in\Delta\left(S\right)$ be a distribution.
Then, there exists a \emph{consistent signal profile} $\theta:S\rightarrow\Delta\left(M\right)$
such that $\left(S,\sigma,\theta\right)$ is a steady state. 
\end{lem}
\begin{proof}
Observe that the space $O_{S}$ is a convex and compact subset of
a Euclidean space, and that the mapping $f_{\sigma}:O_{S}\rightarrow O_{S}$
(defined in (\ref{eq:transformation-between-signal=00003Dprofiles-3})
above) is continuous. Brouwer's fixed-point theorem implies that the
mapping $\sigma$ has a fixed point, which is a consistent outcome
by definition.
\end{proof}

\subsection{Steady State in a Perturbed Environment}

In this subsection we formally adapt the definitions of a consistent
signal profile and of a steady state to perturbed environments.

Let $f_{\left(\left(1-\epsilon\right)\cdot\sigma+\epsilon\cdot\lambda\right)}:O_{S}\rightarrow O_{S}$
be the \emph{mapping between signal profiles} that is induced by the
population's distribution over strategies $\left(\left(1-\epsilon\right)\cdot\sigma+\epsilon\cdot\lambda\right)$.
That is, $f_{\left(\left(1-\epsilon\right)\cdot\sigma+\epsilon\cdot\lambda\right)}\left(\theta\right)$
is the ``new'' signal profile that is induced by a population of
normal agents who follow strategy distribution $\sigma$ and committed
agents who follow strategy distribution $\lambda$, and who observe
signals about the partners according to the ``current'' signal profile
$\theta$. Specifically, when Alice, who follows strategy $s$, is
being matched with a random partner whose strategy is sampled according
to $\left(1-\epsilon\right)\cdot\sigma+\epsilon\cdot\lambda$, she
observes a random signal according to the ``current'' average distribution
of signals in the population $\theta_{\left(\left(1-\epsilon\right)\cdot\sigma+\epsilon\cdot\lambda\right)}$.
As a result her distribution of actions is $s\left(\theta_{\left(\left(1-\epsilon\right)\cdot\sigma+\epsilon\cdot\lambda\right)}\right)$,
and consequently her behavior induces the signal distribution $\nu\left(s\left(\theta_{\left(\left(1-\epsilon\right)\cdot\sigma+\epsilon\cdot\lambda\right)}\right)\right)$.
Thus, we define this latter expression as her ``new'' distribution
of signals $\left(f_{\left(\left(1-\epsilon\right)\cdot\sigma+\epsilon\cdot\lambda\right)}\left(\theta\right)\right)_{s}$.
Formally:

\begin{equation}
\forall m\in M,\,\,s\in S,\,\,\,\,\left(f_{\left(\left(1-\epsilon\right)\cdot\sigma+\epsilon\cdot\lambda\right)}\left(\theta\right)\right)_{s}\left(m\right)=\nu\left(s\left(\theta_{\left(\left(1-\epsilon\right)\cdot\sigma+\epsilon\cdot\lambda\right)}\right)\right)\left(m\right).\label{eq:transformation-between-signal=00003Dprofiles-1-1}
\end{equation}
Given a distribution of strategies $\left(\left(1-\epsilon\right)\cdot\sigma+\epsilon\cdot\lambda\right)$,
we say that a signal profile $\theta^{*}:S^{C}\cup S^{N}\rightarrow\Delta\left(M\right)$
is \emph{consistent} if it is a fixed point of the mapping $f_{\left(\left(1-\epsilon\right)\cdot\sigma+\epsilon\cdot\lambda\right)}$,
i.e., if $f_{\left(\left(1-\epsilon\right)\cdot\sigma+\epsilon\cdot\lambda\right)}\left(\theta^{*}\right)=\theta^{*}$.

We formally adapt the definition of a steady state as follows:
\begin{defn}
\label{def:state-2}A \emph{steady state }(or \emph{state} for short)
of a perturbed environment \emph{$\left(\left(G,k\right),\left(S^{C},\lambda\right),\epsilon\right)$}
is a triple $\left(S^{N},\sigma,\theta\right)$ where $S^{N}\subseteq\mathcal{S}$
is a finite set of strategies (called, \emph{normal strategies}),
$\sigma\in\Delta\left(S^{N}\right)$ is a distribution with a full
support over $S^{N}$, and $\theta:S^{N}\cup S^{C}\rightarrow\Delta\left(M\right)$
is a consistent signal profile. 
\end{defn}

\subsection{Convergence of Strategies, Distributions, and States}

Next we formally define the standard notions of convergence of strategies,
convergence of distributions, and convergence of states that are used
throughout the paper.
\begin{defn}[Convergence of strategies, of distributions, and of states]
 Fix environment $\left(G,k\right)$. A sequence of strategies $\left(s_{n}\right)_{n}$
converges to strategy $s$ (denoted by $\left(s_{n}\right)_{n}\rightarrow_{n\rightarrow\infty}s$)
if for each signal $m\in M$ and each action $a$, the sequence of
probabilities $\left(s_{n}\right)_{m}\left(a\right)$ converges to
$s_{m}\left(a\right).$ A distribution of signals $\left(\nu_{n}\right)_{n}$
converges to $\nu$ (denoted by $\left(\nu_{n}\right)_{n}\rightarrow_{n\rightarrow\infty}\nu)$
if the sequence of probabilities $\left(\nu_{n}\right)\left(m\right)$
converges to $\nu\left(m\right)$ for each signal $m$. A sequence
of states $\left(S_{n}^{N},\sigma_{n},\theta_{n}\right)_{n}$ converges
to a state $\left(S^{*},\sigma^{*},\theta^{*}\right)$ if for each
strategy $s\in supp\left(\sigma^{*}\right),$ there exists a sequence
of sets of strategies $\left(\hat{S}_{n}^{N}\right)_{n}$, with $\hat{S}_{n}^{N}\subseteq S_{n}^{N}$
for each $n$, such that (1) $\sum_{s_{n}\in\hat{S}_{n}^{N}}\sigma_{n}\left(s_{n}\right)\rightarrow\sigma^{*}\left(s\right)$,
and for each sequence of elements of those sets (i.e., for each sequence
of strategies $\left(s_{n}\right)_{n}$ such that $s_{n}\in\hat{S}_{n}^{N}$
for each $n$), (2) $s_{n}\rightarrow_{n\rightarrow\infty}s$ , and
(3) $\theta_{n}\left(s_{n}\right)\rightarrow\theta^{*}\left(s\right)$.
\end{defn}

\subsection{Implementing Nash and ``Trembling-Hand'' Perfect Equilibria\label{subsec:Implementing-a-Trembling-Hand}}

In what follows we show that any symmetric (``trembling-hand'' perfect)
Nash equilibrium $\alpha$ of the underlying game corresponds to a
(perfect) Nash equilibrium of the environment in which all normal
agents play $\alpha$ regardless of the observed signal.

The observation on Nash equilibria is immediate:
\begin{fact}
Let $\alpha\in\Delta\left(A\right)$ be a symmetric Nash equilibrium
strategy of the underlying game $G=\left(A,\pi\right)$. Then the
steady state $\left(S^{N}=\left\{ \alpha\right\} ,\alpha\right)$
in which everyone plays $\alpha$ regardless of the observed signal
is a Nash equilibrium in the unperturbed environment $\left(G,k\right)$
for any $k\in\mathbb{N}$.
\end{fact}
In what follows we state and prove the  result on perfect equilibria
(note, that the result hold also for games with more than two actions):
\begin{prop}
\label{pro:perfect-is-perfect-1-1}Let $\alpha\in\Delta\left(A\right)$
be a symmetric perfect equilibrium action of the underlying game $G=\left(A,\pi\right)$.
Then the state $\left(S=\left\{ \alpha\right\} ,\nu_{\alpha}\right)$
is a perfect equilibrium in the environment $\left(G,k\right)$ for
any $k\in\mathbb{N}$. Moreover, if the distribution $\alpha$ is
not totally mixed, then $\left(S^{N}=\left\{ \alpha\right\} ,\nu_{\alpha}\right)$
is a regular perfect equilibrium.
\end{prop}
\begin{proof}
If $\alpha$ is a totally mixed strategy, then it is immediate that
the state $\left(\left\{ \alpha\right\} ,\nu_{\alpha}\right)$ is
a Nash equilibrium of the perturbed environment$\left(\left(G,k\right),\left\{ \alpha\right\} ,\epsilon\right)$
for any $\epsilon>0$, which implies that the state $\left(\left\{ \alpha\right\} ,\nu_{\alpha}\right)$
is a perfect equilibrium. Assume now that $\alpha$ is not totally
mixed. The fact that $\alpha\in\Delta\left(A\right)$ is a symmetric
perfect equilibrium of the underlying game implies (see \citealp[Theorem 7]{selten1975reexamination})
that there is a sequence of totally mixed strategies $\left(\alpha_{n}\right)\rightarrow_{n\rightarrow\infty}\alpha$,
such that $\alpha$ is a best reply to each $\alpha_{n}$. The fact
that $\alpha$ is a best reply both to itself and to $\alpha_{1}$
(the first element in the sequence $\left(\alpha_{n}\right)$) implies
that the state $\left(\left\{ \alpha\right\} ,\nu_{\alpha}\right)$
is a Nash equilibrium of the regular perturbed environment $\left(\left(G,k\right),\left(\left\{ \alpha_{1},\alpha\right\} ,\left(0.5,0.5\right)\right),\epsilon\right)$
for any $\epsilon>0$, which implies that $\left(\left\{ \alpha\right\} ,\nu_{\alpha}\right)$
is a regular perfect equilibrium. 
\end{proof}

\subsection{Stability of Cooperation when Observing a Single Action\label{subsec:Stability-of-Cooperation-single-action}}

In what follows we characterize which distributions of commitments
support cooperation as a perfect equilibrium outcome in a defensive
Prisoner's Dilemma when $k=1$.

Given a distribution of commitments $\left(S^{C},\lambda\right)$,
we define $\beta_{\left(S^{C},\lambda\right)}\in\left(0,1\right)$
as follows: 
\begin{equation}
\beta_{\left(S^{C},\lambda\right)}=\frac{\mathbf{E}_{\lambda}\left(\left(s_{0}\left(d\right)\right)^{2}\right)}{\mathbf{E}_{\lambda}\left(\left(s_{0}\left(d\right)\right)\right)}=\frac{\sum_{s\in S^{C}}\lambda\left(s\right)\cdot\left(s_{0}\left(d\right)\right)^{2}}{\sum_{s\in S^{C}}\lambda\left(s\right)\cdot s_{0}\left(d\right)}.\label{eq:beta-C-lambda-1-1}
\end{equation}
The value of $\beta_{\left(S^{C},\lambda\right)}$ is the ratio between
the mean of the square of the probability of defection of a random
committed agent who observes $m=0$ and the mean of the same probability
without squaring it. In particular, when the set of commitments is
a singleton, $\beta_{\left(S^{C},\lambda\right)}$ is equal to the
probability that a committed agent defects when she observes $m=0$
(i.e., $\beta_{\left(S^{C},\lambda\right)}=s_{0}\left(d\right)$).

The following result shows that if the game is defensive and agents
observe a single action, then cooperation is a perfect equilibrium
action with respect to the distribution of commitments $\left(S^{C},\lambda\right)$
iff $g\leq\beta_{\left(S^{C},\lambda\right)}$. 
\begin{prop}
\label{prop:observing-single-action-1}\label{prop:Let--be-single-action-complex-result}Let
$E=\left(G_{PD},1\right)$ be an environment, where $G_{PD}$ is a
defensive Prisoner's Dilemma ($g<l$). Let $\left(S^{C},\lambda\right)$
be a distribution of commitments. There exists a regular perfect equilibrium
$\left(S^{*},\sigma^{*},\theta^{*}\equiv0\right)$ with respect to
$\left(S^{C},\lambda\right)$ iff $g\leq\beta_{\left(S^{C},\lambda\right)}$. 

\emph{The proof of Prop. \ref{prop:observing-single-action-1} is
given in Appendix \ref{subsec:Stability-of-Cooperation-single-action}.
Observe that Prop. \ref{prop:observing-single-action-1} immediately
implies that cooperation is a perfect equilibrium outcome in a defensive
Prisoner's Dilemma with $k=1$ iff $g<1$ (i.e., Proposition \ref{prop:observing-single-action-1}
implies Proposition \ref{Pro-stability-of-cooperation-when-k=00003D1}).}
\end{prop}

\section{Stronger Equilibrium Refinements\label{sec:Evolutionary-Stability}}

In the main text we dealt with the notion of perfect equilibrium.
In this appendix we present three stronger refinements of this solution
concept: strict perfection, evolutionary stability, and robustness.

\subsection{Strictly Perfect Equilibrium Action\label{subsec:Strictly-Perfect}}

The notion of perfect equilibrium might be considered too weak because
it may crucially depend on a specific set of commitment strategies.
In what follows we present the refinement of strict perfection that
requires the equilibrium outcome to be sustained regardless of which
commitment strategies are present in the population.

In most of our results we focus on pure perfect equilibria in which
there exists an action $a^{*}$ that is played with probability one
in the limit in which the frequency of committed agents converges
to zero. In order to simplify the notation, we define the refinement
of strict perfection only with respect to pure equilibrium outcomes.
We say that an action $a\in A$ is strictly perfect if it is the limit
behavior of Nash equilibria with respect to \emph{all} distributions
of commitment strategies. Formally:\footnote{\citet{okada1981stability} deals with normal-form games and presents
the related notion of a strict perfect equilibrium as the limit of
Nash equilibria for any ``trembling-hand'' perturbation. In our
setup different strategies might be equivalent in the sense that they
induce the same observable behavior, as the frequency of the commitment
agents converges to zero. Our notion focuses on the observed behavior
(i.e., everyone playing action $a^{*}$), but allows for the choice
of strategy that induces the pure action $a^{*}$ to depend on the
distribution of commitments. This approach is in the spirit of other
set-wise solution concepts in the literature, such as evolutionarily
stable sets (\citealp{thomas1985evolutionarily}) and hyperstable
sets (\citealp{kohlberg1986strategic}).}
\begin{defn}
\label{def-strict-perfect-outcome}Action $a^{*}\in A$ is a \emph{strictly
perfect} equilibrium action in the environment $E=\left(G,k\right)$
if, for any distribution of commitment strategies \emph{$\left(S^{C},\lambda\right)$},
there exist a steady state $\left(S^{*},\sigma^{*},\theta^{*}\equiv\nu_{a^{*}}\right)$
and converging sequences $\left(S_{n}^{N},\sigma_{n},\theta_{n}\right)_{n}\rightarrow_{n\rightarrow\infty}\left(S^{*},\sigma^{*},\theta^{*}\right)$
and $\left(\epsilon_{n}>0\right)_{n}\rightarrow_{n\rightarrow\infty}0$,
such that for each $n$, the state $\left(S_{n}^{N},\sigma_{n},\theta_{n}\right)$
is a Nash equilibrium of the perturbed environment $\left(\left(G,k\right),\left(S^{C},\lambda\right),\epsilon_{n}\right)$. 
\end{defn}

\paragraph{Equilibrium actions that satisfy strict perfection}

The formal proofs of Proposition \ref{pro:defection-is-evol-stable}
show that defection always satisfies the refinement of strict perfection.
The formal proofs of Theorems \ref{thm:cooperation-defensive-PDs}
and \ref{thm:stable-cooperation-observing-conflicts} show that cooperation
satisfies the refinement of strict perfection when agents observe
at least two actions in defensive games, or observe at least two conflicts
in mild games.

\paragraph{Equilibrium actions that do not satisfy strict perfection}

The proof of Proposition \ref{prop:Let--be-single-action-complex-result}
shows that when agents observe a single action, cooperation is a perfect
equilibrium action only with respect to some distributions of commitment
strategies (namely, those in which the value of $\beta_{\left(S^{C},\lambda\right)}$
is sufficiently large), and, thus, it is not strictly perfect. 

Cooperation is also not a strictly perfect equilibrium in Theorems
\ref{thm:stable-cooperation-observing-action-profiles} and \ref{thmtertiary-observation}
(dealing with observation of action profiles and observation of actions
against cooperation, respectively). Specifically, cooperation is not
a perfect equilibrium action with respect to distributions of commitments
in which the committed agents defect with high probability. The reason
is that committed agents who defect with high probability induce normal
partners to defect against them with probability one. This implies
that when observing a partner to be involved in either side of a unilateral
defection (either as the sole defector or as the sole cooperator),
the partner is most likely to be normal. As a result the agents' incentives
to defect are the same when observing mutual cooperation as when observing
unilateral defection, and this does not allow cooperation to be supported
in a perfect equilibrium, as such cooperation relies on agents who
have better incentives to defect when observing a unilateral defection.

\subsection{Evolutionary Stability }

The notion of perfect equilibrium requires that no agent be able to
achieve a better payoff than the incumbents by unilateral deviation.
In what follows we present the refinement of evolutionary stability,
which requires stability also against small groups of agents who deviate
together. 

\subsubsection{Definitions}

In a seminal paper, \citet{smith1973lhe} define a symmetric Nash
equilibrium strategy $\alpha^{*}$ to be evolutionarily stable if
the incumbents achieve a strictly higher payoff when being matched
with any other best-reply strategy $\beta$ (i.e., $\pi\left(\beta,\alpha^{*}\right)=\pi\left(\alpha^{*},\alpha^{*}\right)\,\Rightarrow\,$$\pi\left(\alpha^{*},\beta\right)>\pi\left(\beta,\beta\right)$.
The motivation is that if $\beta$ is a best reply to $\alpha^{*}$,
then a single deviator who plays $\beta$ will be as successful as
the incumbents. This may induce a few other agents to mimic her behavior,
until a small positive mass of agents follow $\beta$. The above inequality
implies that at this stage the followers of $\beta$ will be strictly
outperformed, and thus will disappear from the population. 

Our setup with environments is similar to the standard setup of a
repeated game in that it rarely admits evolutionarily stable strategies.
Typically, not all the actions will be played by normal agents in
equilibrium, and as a result some signals will never be observed.
Deviators who differ in their behavior only after such zero probability
signals will get the same payoff as the incumbents both against the
incumbents and against other deviators. This violates the above inequality. 

Following \citeauthor{selten1983evolutionary}'s \citeyearpar{selten1983evolutionary}
notion of ``limit ESS,'' (see also \citealp{heller2014stability})
we solve this issue by requiring evolutionary stability in a converging
sequence of perturbed environments, in which all signals are observed
on the equilibrium path, instead of simply requiring evolutionary
stability in the unperturbed environment.

This is formalized as follows. Given a steady state $\left(S,\sigma,\theta\right)$
in a perturbed environment $\left(\left(G,k\right),\left(S^{C},\lambda\right),\epsilon\right)$,
we define $\pi_{\hat{s}}\left(\hat{s}\right)$ as the (long-run average)
payoff of strategy $\hat{s}$ against itself, and $\pi_{\left(S,\sigma\right)}\left(\hat{s}\right)$
as the mean (long-run average) payoff of the incumbents against strategy
$\hat{s}$. Specifically, if $\hat{s}\in S\cup S^{C}$, then 
\[
\pi_{\hat{s}}\left(\hat{s}|S,\sigma,\theta\right)=\sum_{\left(a,a'\right)\in A^{2}}\hat{\theta}_{\hat{s}}\left(\hat{s}\right)\left(a\right)\cdot\hat{\theta}_{\hat{s}}\left(\hat{s}\right)\left(a'\right)\cdot\pi\left(a,a'\right),
\]
\[
\pi_{\left(S,\sigma\right)}\left(\hat{s}|S,\sigma,\theta\right)=\sum_{s\in S\cup S^{C}}\sum_{\left(a,a'\right)\in A^{2}}\left(\left(1-\epsilon\right)\cdot\sigma\left(s\right)+\epsilon\cdot\lambda\left(s\right)\right)\cdot\theta_{s}\left(\hat{s}\right)\left(a\right)\cdot\hat{\theta}_{\hat{s}}\left(s\right)\left(a'\right)\cdot\pi\left(a,a'\right),
\]
and if $\hat{s}\notin S\cup S^{C}$, then we define $\pi_{\hat{s}}\left(\hat{s}\right)$
and $\pi_{\left(S,\sigma\right)}\left(\hat{s}\right)$ as the respective
payoffs in the post-deviation steady state $\left(S\cup\left\{ \hat{s}\right\} ,\hat{\sigma},\hat{\theta}\right)$:
\[
\pi_{\hat{s}}\left(\hat{s}|S,\sigma,\theta\right)=\sum_{\left(a,a'\right)\in A^{2}}\hat{\theta}_{\hat{s}}\left(\hat{s}\right)\left(a\right)\cdot\hat{\theta}_{\hat{s}}\left(\hat{s}\right)\left(a'\right)\cdot\pi\left(a,a'\right),
\]
\[
\pi_{\left(S,\sigma\right)}\left(\hat{s}|S,\sigma,\theta\right)=\sum_{s\in S\cup S^{C}}\sum_{\left(a,a'\right)\in A^{2}}\left(\left(1-\epsilon\right)\cdot\hat{\sigma}\left(s\right)+\epsilon\cdot\lambda\left(s\right)\right)\cdot\hat{\theta}_{s}\left(\hat{s}\right)\left(a\right)\cdot\hat{\theta}_{\hat{s}}\left(s\right)\left(a'\right)\cdot\pi\left(a,a'\right).
\]

\begin{defn}
A steady state $\left(S^{*},\sigma^{*},\theta^{*}\right)$ of a perturbed
environment $\left(\left(G,k\right),\left(S^{C},\lambda\right),\epsilon\right)$
is \emph{evolutionarily stable} if (1) $\left(S^{*},\sigma^{*},\theta^{*}\right)$
is a Nash equilibrium, and (2) for any best-reply strategy $\hat{s}$
(i.e., $\pi_{\hat{s}}\left(S^{*},\sigma^{*},\theta^{*}\right)=\pi\left(S^{*},\sigma^{*},\theta^{*}\right)$),
such that $\sigma^{*}\left(\hat{s}\right)<1$ (i.e., $\hat{s}$ is
not the only normal strategy) the following inequality holds: $\pi_{\left(S,\sigma\right)}\left(\hat{s}|S,\sigma,\theta\right)>\pi_{\hat{s}}\left(\hat{s}|S,\sigma,\theta\right)$.
\end{defn}

\begin{defn}
A steady state $\left(S^{*},\sigma^{*},\theta^{*}\right)$ of the
environment $\left(G,k\right)$ is\emph{ a perfect evolutionarily
stable state} if there exist a distribution of commitments $\left(S^{C},\lambda\right)$
and converging sequences $\left(S_{n}^{N},\sigma_{n},\theta_{n}\right)_{n}\rightarrow_{n\rightarrow\infty}\left(S^{*},\sigma^{*},\theta^{*}\right)$
and $\left(\epsilon_{n}>0\right)_{n}\rightarrow_{n\rightarrow\infty}0$,
such that for each $n$, the state $\left(S_{n}^{N},\sigma_{n},\theta_{n}\right)$
is an evolutionarily stable state in the perturbed environment $\left(\left(G,k\right),\left(S^{C},\lambda\right),\epsilon_{n}\right)$.
If the outcome assigns probability one to one of the actions, i.e.,
$\theta^{*}\equiv a$, then we say that this action is a perfect evolutionarily
stable outcome. 
\end{defn}
Finally, we define a strictly perfect evolutionarily stable outcome
as a pure action that is an outcome of a perfect evolutionarily stable
state for any distribution of commitments (similar to the notion of
strict limit ESS in \citealp{heller2014Three}).
\begin{defn}
Action $a^{*}\in A$ is a \emph{strictly perfect} \emph{evolutionarily
stable outcome} in the environment $E=\left(\left(A,\pi\right),k\right)$
if, for any distribution of commitment strategies \emph{$\left(S^{C},\lambda\right)$},
there exist a steady state $\left(S^{*},\sigma^{*},\theta^{*}\equiv a^{*}\right)$
and converging sequences $\left(S_{n}^{N},\sigma_{n},\theta_{n}\right)_{n}\rightarrow_{n\rightarrow\infty}\left(S^{*},\sigma^{*},\theta^{*}\right)$
and $\left(\epsilon_{n}>0\right)_{n}\rightarrow_{n\rightarrow\infty}0$,
such that for each $n$, the state $\left(S_{n}^{N},\sigma_{n},\theta_{n}\right)$
is an evolutionarily stable state in the perturbed environment $\left(\left(G,k\right),\left(S^{C},\lambda\right),\epsilon_{n}\right)$.
\end{defn}

\subsubsection{Adaptation of Results}

All of our results hold with respect to the refinement of evolutionary
stability. In particular, the fact that always defecting is a strict
equilibrium (i.e., the unique best reply to itself) in any slightly
perturbed environment implies that defection is a strictly perfect
evolutionarily stable outcome. 

One can adapt the results about sustaining cooperation as an equilibrium
action (Theorems \ref{thm:cooperation-defensive-PDs}\textendash \ref{thmtertiary-observation}).
Specifically, minor modifications to the proofs can show that cooperation
is a strictly perfect evolutionarily stable outcome in defensive games
with observation actions and in mild games with observation of conflicts
(when $k\geq2)$, and that cooperation is a perfect evolutionarily
stable outcome in mild games with observation of action profiles,
and in any game with observation of actions against defectors.

A sketch of the argument why the results apply also to the refinement
of evolutionary stability is as follows. There are two kinds of steady
states that sustain cooperation in the proofs in this paper:
\begin{enumerate}
\item Steady state $\psi'_{n}=\left(\left\{ s^{q_{n}}\right\} ,\theta_{n}\right)$
that has a single normal strategy in its support. The arguments in
the proofs show that each such strategy is the unique best reply to
itself in the $n^{\textrm{th}}$ perturbed environment (i.e., $\pi_{s'}\left(\psi'_{n}\right)<\pi\left(\psi'_{n}\right)$
for each $s'\neq s^{q_{n}}$), which shows that $\psi'_{n}$ is an
evolutionarily stable state in the $n^{\textrm{th}}$ perturbed environment.
\item Steady state $\psi_{n}=\left(\left\{ s^{1},s^{2}\right\} ,\left(q_{n},1-q_{n}\right),\theta_{n}\right)$
that consists of two normal strategies in its support. The arguments
in the proofs show that these two strategies are the only best replies
to this steady state (i.e., $\pi_{s'}\left(\psi_{n}\right)<\pi\left(\psi_{n}\theta_{n}\right)$
for each $s'\notin\left\{ s^{1},s^{2}\right\} $ ). Moreover, the
arguments in the proof (see, in particular, Remark \ref{rem:evolutionarily-stability}
at the end of the proof of Theorem \ref{thm:cooperation-defensive-PDs})
imply that each of these two normal strategies obtains a relatively
low payoff when being matched against itself, i.e.,: $\pi\left(s^{1}|\psi_{n}\right)>\pi_{s^{1}}\left(s^{1}|S,\sigma,\theta\right)$
and $\pi\left(s^{2}|\psi_{n}\right)>\pi_{s^{2}}\left(s^{2}|\psi_{n}\right)$,
which implies that $\psi_{n}$ is evolutionarily stable. 
\end{enumerate}

\subsection{Robustness\label{subsec:Robustness}}

The outcome of a perfect equilibrium may be unstable in the sense
that small perturbations of the distribution of observed signals may
induce a change of behavior that moves the population away from the
consistent signal profile. We address this issue by introducing a
robustness refinement (in the spirit of the notion of Lyapunov stability
in dynamic environments) that requires that if we slightly perturb
the distribution of observed signals, then the agents converge back
to playing the equilibrium outcome.

In order to simplify the notation, we define the refinement of robustness
only with respect to pure equilibrium outcomes. We say that a pure
perfect equilibrium with outcome $a^{*}$ is robust if there exists
a bounded sequence of parameters $\left(\kappa_{n}\right)_{n}$ such
that for each perturbed environment with $\epsilon_{n}$ committed
agents: (1) the normal agents play action $a^{*}$with a probability
greater than $1-\kappa_{n}\cdot\epsilon_{n}$ in the steady state,
and (2) if one perturbs the initial distribution of signals to any
other (possibly inconsistent) signal profile in which the normal agents
are observed to play action $a^{*}$with a probability of at least
$1-\kappa_{n}\cdot\epsilon_{n}$, then agents continue to play action
$a^{*}$ with a probability of at least $1-\kappa_{n}\cdot\epsilon_{n}$
in the new signal profile that is induced by the agents' behavior
and the perturbed signal profile. 

Let $\Delta^{bn}\left(M\right)\subseteq\Delta\left(M\right)$ be the
set of binomial distributions of signals that are induced by distributions
of actions, i.e., 
\[
\Delta^{bn}\left(M\right)=\left\{ \tilde{\nu}\in\Delta\left(M|\exists\alpha\in\Delta\left(A\right)\,s.t.\,\tilde{\nu}=\nu\left(\alpha\right)\right)\right\} ,
\]
Let $\alpha\left(\theta_{s}\right)\in\Delta\left(A\right)$ is the
distribution of actions that induce signals distributed according
to $\theta_{s}\in\Delta^{bn}\left(M\right)$, i.e., $\nu\left(\alpha\left(\theta_{s}\right)\right)=\theta_{s}$.
Given a distribution of normal strategies $\left(S^{N},\sigma\right)$
and a (possibly inconsistent) signal profile $\theta$, let $\alpha_{\sigma}\left(\theta\right)\in\Delta\left(A\right)$
be the ($\sigma$-weighted) population average of the distributions
of actions that induce signals distributed according to the signal
profile $\theta\in\Delta^{nm}\left(M\right)$ for the normal agents;
i.e., for each action $a\in A$,

\[
\alpha_{\sigma}\left(\theta\right)\left(a\right)=\sum_{s\in S^{N}}\sigma\left(s\right)\cdot\alpha\left(\theta_{s}\right)\left(a\right).
\]
That is, $\left(\alpha\left(\theta_{s}\right)\right)_{s\in S^{N}}$
is the profile of distributions of actions that generate the profile
of signal distributions for the normal agents $\theta=\left\{ \theta_{s}\right\} _{s\in S^{N}}$,
and $\alpha_{\sigma}\left(\theta\right)$ is the ($\sigma$-weighted)
average of the distributions of actions in this profile. 

The formal definition of robust perfection is as follows.
\begin{defn}
\label{def-strict-perfect-outcome-1}Let $\left(S^{*},\sigma^{*},\theta^{*}\equiv a^{*}\right)$
be a perfect equilibrium with respect to the distribution of commitments
$\left(S^{C},\lambda\right)$ and the converging sequences $\left(S_{n}^{N},\sigma_{n},\theta_{n}\right)_{n}\rightarrow_{n\rightarrow\infty}\left(S^{*},\sigma^{*},\theta^{*}\right)$
and $\left(\epsilon_{n}>0\right)_{n}\rightarrow_{n\rightarrow\infty}0$.
The equilibrium $\left(S^{*},\sigma^{*},\theta^{*}\equiv a^{*}\right)$
is \emph{robust} if there exists $\kappa>0$ and a bounded sequence
$0<\left(\kappa_{n}\right)_{n}<\kappa$, such that for each $n$,
(1) $\alpha_{\sigma_{n}}\left(\theta_{n}\right)\left(a^{*}\right)>1-\kappa_{n}\cdot\epsilon_{n}$,
and (2) for each signal profile $\theta\in O_{\left(S_{n}^{N}\cup S^{C}\right)}$,
\[
\alpha_{\sigma_{n}}\left(\theta\right)\left(a^{*}\right)\geq1-\kappa_{n}\cdot\epsilon_{n}\,\,\Rightarrow\,\,\alpha_{\sigma}\left(f_{\left(\left(1-\epsilon\right)\cdot\sigma_{n}+\epsilon\cdot\lambda\right)}\theta\right)\left(a^{*}\right)>1-\kappa_{n}\cdot\epsilon_{n}.
\]
The proof of Part 2 of Theorem \ref{thm:cooperation-defensive-PDs}
contains a detailed argument as to why the cooperative equilibrium
of Theorem \ref{thm:cooperation-defensive-PDs} is robust. The argument
as to why all other cooperative equilibria in Theorems \ref{thm:stable-cooperation-observing-conflicts}\textendash \ref{thmtertiary-observation}
are robust is analogous. (It is immediate that the defective perfect
equilibrium of Proposition \ref{pro:defection-is-evol-stable} satisfies
robustness because the behavior of the normal agents is independent
of what these agents observe.) 
\end{defn}

\section{Proofs\label{sec:Proofs}}

\subsection{Proof of Proposition \ref{pro:defection-is-evol-stable} (Defection
is Perfect)\label{subsec:Proof-of-Theorem-strict-is-strict-stable}}

We will prove a stronger result, namely, that defection is a strictly
perfect equilibrium action (as defined in Appendix \ref{subsec:Strictly-Perfect}),
i.e., that it is a perfect equilibrium action with respect to all
distributions of commitment strategies. Let $\zeta=\left(\mathcal{S^{C}},\lambda\right)$
be a distribution of commitments. Let $s_{d}\equiv d$ be the strategy
that always defects. Let $\left(\left\{ s_{d}\right\} ,\theta_{n}\right)$
be a steady state of the perturbed environment $\left(\left(G,k\right),\left(S^{C},\lambda\right),\epsilon_{n}\right)$.
The fact that an agent who follows $s_{d}\equiv d$ always defects
implies that $\left(\theta_{n}\right)_{s_{d}}\left(k\right)=1$ (i.e.,
the agent is always observed to defect in all $k$ interactions).
Consider a deviating agent (Alice) who follows any strategy $s\neq s_{d}$.
We show that Alice is strictly outperformed in any post-deviation
steady state.

The facts that $s\neq s_{d}$ and that all signals are observed with
positive probability in any perturbed environment imply that Alice
cooperates with an average probability of $\alpha>0$. We now compare
the payoff of Alice to the payoff of an incumbent (Bob) who follows
$s_{d}$. Alice obtains a direct loss of at least $\alpha\cdot\textrm{min}\left(g,l\right)$
due to cooperating with probability $\alpha$. The maximal indirect
benefit that she might achieve due to these cooperations (by inducing
committed agents to cooperate against her with higher probability
relative to their cooperation probability against Bob) is $\epsilon_{n}\cdot k\cdot\alpha\cdot\left(l+1\right)$
because there are $\epsilon_{n}$ committed agents, each of whom observes
Alice cooperate at least once in the $k$ sampled actions with a probability
of at most $k\cdot\alpha$, and each committed partner can yield Alice
a benefit of at most $l+1$ by cooperating when the partner observes
$m\geq1$. If $\epsilon_{n}$ is sufficiently small ($\epsilon_{n}<\frac{1}{k\cdot\left(l+1\right)}$),
then the direct loss is larger than the indirect maximal benefit ($\alpha>\epsilon_{n}\cdot k\cdot\alpha\cdot\left(l+1\right)$).
This implies that $\left(\left\{ s_{d}\right\} ,\theta_{n}\right)$
is a (strict) Nash equilibrium in any environment with $\epsilon_{n}<\frac{1}{k\cdot\left(l+1\right)}$,
which proves defection is a strictly perfect equilibrium action.

\subsection{Proof of Theorem \ref{thm:only-defection-is-stable} (Defection is
the Unique Equilibrium in Offensive PDs)\label{subsec:Proof-of-Theorem-defection-unique-stable-outcome}}

Let $\left(S^{*},\sigma^{*},\theta^{*}\right)$ be a regular perfect
equilibrium. That is, there exists a regular distribution of commitments
$\left(\mathcal{S^{C}},\lambda\right)$, a converging sequence $\left(\epsilon_{n}\right)_{n}\rightarrow0$,
and a converging sequence of steady states $\left(S_{n}^{N},\sigma_{n},\theta_{n}\right)\rightarrow\left(S^{*},\sigma^{*},\theta^{*}\right)$,
such that for each $n$ the state $\left(S_{n}^{N},\sigma_{n},\theta_{n}\right)$
is a Nash equilibrium of $\left(\left(G,k\right),\left(S^{C},\lambda\right),\epsilon_{n}\right)$.
We assume to the contrary that $S^{*}\neq\left\{ d\right\} $. 

Recall that any signal $m\in M=\left\{ 0,...,k\right\} $ is observed
with positive probability in any perturbed environment. Given a state
$\left(S_{n}^{N},\sigma_{n},\theta_{n}\right)$, an environment $\left(\left(G,k\right),\left(S^{C},\lambda\right),\epsilon_{n}\right)$,
a signal $m\in M$, and a strategy $s\in S_{n}^{N}$, let $q\left(m,s\right)$
denote the probability that a randomly drawn partner of a player defects,
conditional on the player following strategy $s$ and observing signal
$m$ about the partner.

We say that a strategy is ``defector-favoring'' if the strategy
is to defect against partners who are likely to cooperate, and to
cooperate against partners who are likely to defect. Specifically,
a strategy is defector-favoring if there is some threshold such that
the strategy is to cooperate (defect) when the partner's conditional
probability of defecting is above (below) this threshold. Formally:
\begin{defn}
Strategy $s\in S_{n}^{N}$ is \emph{defector-favoring}, given state
$\left(S_{n}^{N},\sigma_{n},\theta_{n}\right)$ and environment $\left(\left(G,k\right),\left(S^{C},\lambda\right),\epsilon_{n}\right)$,
if there is some $\bar{q}\in\left[0,1\right]$ such that, for each
$m\in M$, ~$q\left(m,s\right)>\bar{q}\,\Rightarrow\,s_{m}\left(d\right)=0,\,\,\,\textrm{and}\,\,\,q\left(m,s\right)<\bar{q}\,\Rightarrow\,s_{m}\left(d\right)=1.$
\end{defn}
The rest of the proof consists of the following four steps.

First, we show that all normal strategies are defector-favoring. Assume
to the contrary that there is a strategy $s\in S_{n}^{N}$ that is
not defector-favoring. Let $s'$ be a defector-favoring strategy that
has the same average defection probability as $s$ in the post-deviation
steady state. The fact that both strategies prescribe defection with
the same average probability implies that they induce the same behavior
from the partners (since these partners observe identical distributions
of signals when facing $s$ and when facing $s'$), and hence $q\left(m,s\right)=q\left(m,s'\right)$.
Agents who follow strategy $s'$ defect more often against partners
who are more likely to cooperate relative to strategy $s$. Since
the underlying game is offensive this implies that strategy $s'$
strictly outperforms strategy $s$, which contradicts that $\left(S_{n}^{N},\sigma_{n},\theta_{n}\right)$
is a Nash equilibrium.

Second, we show that all the normal strategies lead agents to defect
with the same average probability in $\left(S_{n}^{N},\sigma_{n},\theta_{n}\right)$.
Assume to the contrary that there are strategies $s,s'\in S_{n}^{N}$
such that agents following the former strategy have a higher average
probability of defection, i.e., $\alpha\left(\theta_{s}\right)\left(d\right)>\alpha\left(\theta_{s'}\right)\left(d\right)$.
Let $\beta=\alpha\left(\theta_{s}\right)\left(d\right)-\alpha\left(\theta_{s'}\right)\left(d\right)$.
Note that agents who follow strategy $s$ have a strictly higher payoff
than agents who follow $s'$ when being matched with normal partners.
This is because strategy $s$ yields: (1) a strictly higher direct
payoff of at least $\beta\cdot l$ due to playing more often the dominant
action $d$, and (2) a weakly higher payoff against normal agents,
because the fact that agents who follow it defect more often and all
normal agents follow defector-favoring strategies implies that normal
partners defect with a weakly smaller probability when being matched
with agents who follow strategy $s$ (relative to $s'$). We also
need to consider what happens when normal agents are matched with
committed agents. The maximal indirect gain that followers of strategy
$s'$ have relative to followers of strategy $s$, due to inducing
a higher probability of cooperation from committed partners, is at
most $\epsilon_{n}\cdot\left(l+1\right)\cdot k\cdot\beta$. This implies
that if $\epsilon_{n}<\frac{l}{\left(l+1\right)\cdot k}$, then followers
of strategy $s$ have a strictly higher payoff than followers of $s'$,
which contradicts that $\left(S_{n}^{N},\sigma_{n},\theta_{n}\right)$
is a Nash equilibrium. 

Third, we argue that for any normal agent it is the case that the
probability that the partner defects conditional on the agent observing
signal $m=k$ is weakly larger than the probability that the partner
defects conditional on the agent observing any signal $m<k$. To see
why, note that the regularity of the set of commitments implies that
not all commitment strategies have the same defection probabilities,
and thus the signal about the partner yields some information about
the partner's probability of defecting. The previous step shows that
all normal agents defect with the same probability, which implies
that they induce the same signal distribution, and thus they induce
the same behavior from all partners. Combining this fact with the
fact that not all commitment strategies have the same defection probability
implies (for a sufficiently small $\epsilon_{n})$ that if a player
observes a signal that includes only defections, then the partner
is more likely to have a higher average defection probability against
normal agents (i.e., $q\left(m,s\right)<q\left(k,s\right)$ for any
normal strategy $s$ and any $m<k$). 

Thus, any normal agent (who follows a defector-favoring strategy due
to the first step) defects with a weakly lower probability after observing
signal $m=k$. This implies that if $\epsilon_{n}$ is sufficiently
small, then a deviator who always defects outperforms the incumbents.
The deviator achieves a direct higher payoff by defecting more often,
as well as a weakly higher indirect gain by inducing the incumbents
to cooperate more often. 

\subsection{Proof of Theorem \ref{thm:cooperation-defensive-PDs} (Cooperation
Is Perfect in Defensive PDs) \label{subsec:Proof-of-Theorem-defensive}}

\paragraph{Part 1:}

Let $\left(S^{*},\sigma^{*},\theta^{*}\equiv0\right)$ be a perfect
equilibrium. This implies that there exist a distribution of commitments
$\left(S^{C},\lambda\right)$, a converging sequence of strictly positive
commitment levels $\epsilon_{n}\rightarrow_{n\rightarrow\infty}0$,
and a converging sequence of steady states $\left(S_{n}^{N},\sigma_{n},\theta_{n}\right)\rightarrow_{n\rightarrow\infty}\left(S^{*},\sigma^{*},\theta^{*}\right)$,
such that for each $n$ the state $\left(S_{n}^{N},\sigma_{n},\theta_{n}\right)$
is a Nash equilibrium of the perturbed environment $\left(\left(G,k\right),\left(S^{C},\lambda\right),\epsilon_{n}\right)$.
The fact that the equilibrium induces full cooperation (in the limit
when $\epsilon_{n}\rightarrow_{n\rightarrow\infty}0$) implies that
all normal agents must cooperate when they observe no defections,
i.e., $s_{0}\left(c\right)=1$ for each $s\in S^{*}$.

Next we show that $s_{1}\left(d\right)>0$ for some $s\in S^{*}$.
Assume to the contrary that $s_{1}\left(d\right)=0$ for every $s\in S^{*}$.
This implies that for any $\delta>0$, if $n$ is sufficiently large
then $\sum_{s\in S_{n}^{N}}\sigma_{n}\left(s\right)\cdot s_{1}\left(d\right)<\delta$.
Consider a deviator (Alice) who follows a strategy $s'$ that defects
with a small probability $\alpha$, satisfying$\epsilon_{n},\delta<<\alpha<<1$,
when observing no defections (i.e., $s'_{0}\left(d\right)=\alpha$).
It turns out that Alice will outperform the incumbents. To see this
note that since she occasionally defects when observing $m=0$ she
obtains a direct gain of at least $\alpha\cdot g$$\cdot\Pr\left(m=0\right)$,
where $\Pr\left(m=0\right)$ is the probability of observing $m=0$
given the steady state $\left(S_{n}^{N},\sigma_{n},\theta_{n}\right)$.
The probability that a partner observes her defecting twice or more
is $\sum_{i=2}^{k}\left(\begin{array}{c}
k\\
i
\end{array}\right)\alpha^{i}\cdot\left(1-\alpha\right)^{k-i}$. This implies that her indirect loss from these defections is at
most $\left(\sum_{i=2}^{k}\left(\begin{array}{c}
k\\
i
\end{array}\right)\alpha^{i}\cdot\left(1-\alpha\right)^{k-i}+\delta+\epsilon_{n}\right)\cdot\left(1+l\right)$ and, thus, for sufficiently small values $\epsilon_{n},\delta<<\alpha<<1$,
Alice strictly outperforms the incumbents.

We now show that $s_{m}\left(d\right)=1$ for all $s\in S^{*}$ and
all $m\geq2$. The fact that $\theta^{*}\equiv0$ implies that for
a sufficiently large $n$, all normal agents cooperate with an average
probability very close to one and, thus, the average probability of
defection by an agent who follows a strategy $s\in S\cup S^{C}$ is
very close to $s_{0}\left(d\right)$. Hence the distribution of signals
induced by such an agent is very close to $\nu_{s_{0}\left(d\right)}$.
Recall that we assume that the distribution of commitments contains
at least one strategy $s$ with $s_{0}\left(d\right)>0$. This implies
that the posterior probability that the partner is going to defect
is strictly increasing in the signal $m$ that the agent observes
about the partner. Note that the direct gain from defecting is strictly
increasing in the probability that the partner defects as well (due
to the game being defensive), while the indirect influence of defection
(on the behavior of future partners who may observe the current defection)
is independent of the partner's play. From the previous paragraph
we know that defection is a best reply conditional on an agent observing
$m=1$. This implies that defection must be the unique best reply
when an agent observes at least two defections (i.e., when $m\geq2$). 

It remains to show that there is a normal incumbent strategy to cooperate
with positive probability after observing a single defection, i.e.,
$s_{1}\left(d\right)<1$ for some $s\in S^{*}$. Assume to the contrary
that $s_{1}\left(d\right)=1$ for every $s\in S^{*}$. Let $r_{n}$
denote the average probability that a normal agent defects after observing
$m\geq1$. Since $\left(S_{n}^{N},\sigma_{n},\theta_{n}\right)\rightarrow_{n\rightarrow\infty}\left(S^{*},\sigma^{*},\theta^{*}\right),$
the assumption that $s_{1}\left(d\right)=1$ for all $s\in S^{*}$
implies that $r_{n}>0.6$ for a sufficiently large $n$. Let $Pr\left(m\geq1|S_{n}^{N}\right)$
denote the probability of observing $m\geq1$ conditional on being
matched with a normal partner. Note that the assumption that $\hat{s}_{0}\left(d\right)>0$
for some committed strategy $\hat{s}$ and the assumption that $s_{1}\left(d\right)>0$
for some normal strategy together imply that $Pr\left(m\geq1|S_{n}^{N}\right)>0$.
Note that $\theta^{*}\equiv c$ implies that $\lim_{n\rightarrow\infty}\left(Pr\left(m=1|S_{n}^{N}\right)\right)=0$.
Hence $Pr\left(m=1|S_{n}^{N}\right)$ is $O\left(\epsilon_{n}\right)$.
We can calculate $Pr\left(m\geq1|S_{n}^{N}\right)$ as follows:
\[
Pr\left(m\geq1|S_{n}^{N}\right)=k\cdot\left(\left(1-\epsilon_{n}\right)\cdot r_{n}\cdot Pr\left(m\geq1|S_{n}^{N}\right)+\epsilon_{n}\cdot\lambda\left(\hat{s}\right)\cdot\left(\hat{s}_{0}\left(d\right)+O\left(\epsilon_{n}\right)\right)\right)-O\left(\epsilon_{n}^{2}\right)-O\left(\left(Pr\left(m\geq1|S_{n}^{N}\right)\right)^{2}\right).
\]
The reason for this equation is as follows. The observed signal induced
by a normal agent (Bob) describes his actions in $k$ interactions.
In each of these interactions Bob's partner was normal with a probability
of $1-\epsilon_{n}$, and was committed with a probability of $\epsilon_{n}$.
If Bob's partner in an interaction was normal then she defected with
a probability of $r_{n}$ when she observed $m\geq1$ (which happened
with a probability of $Pr\left(m\geq1|S_{n}^{N}\right)$). If Bob's
partner in an interaction was committed then she followed strategy
$\hat{s}$ with a probability of $\lambda\left(\hat{s}\right)$ and
defected with a probability of $\hat{s}_{0}\left(d\right)+O\left(\epsilon_{n}\right)$
(as argued above, the average defection probability of an agent following
strategy $s$ should be close to $s_{0}\left(d\right)$). Finally,
the terms $-O\left(\epsilon_{n}^{2}\right)-O\left(\left(Pr\left(m\geq1|S_{n}^{N}\right)\right)^{2}\right)$
are subtracted to avoid ``double-counting'' cases in which Bob has
defected more than once. Rearranging and simplifying the above equation
by using the fact that $\left(Pr\left(m\geq1|S_{n}^{N}\right)\right)^{2}$
is $O\left(\epsilon_{n}^{2}\right)$ yields 
\[
\left(1-k\cdot\left(1-\epsilon_{n}\right)\cdot r_{n}\right)\cdot Pr\left(m\geq1|S_{n}^{N}\right)=k\cdot\left(\epsilon_{n}\cdot\lambda\left(\hat{s}\right)\cdot\hat{s}_{0}\left(d\right)\right).
\]
Then use $r_{n}>0.6$ to infer that the LHS is negative. This contradicts
the fact that the RHS is positive. 

\paragraph{Part 2:}

We prove a stronger result, namely, that cooperation is a strictly
perfect equilibrium action (as defined in Appendix \ref{subsec:Strictly-Perfect}),
i.e., that it is a perfect equilibrium action with respect to all
distributions of commitment strategies. Recall that $s^{1}$ ($s^{2}$)
is the strategy that induces an agent to defect iff the agent observes
$m\geq1$ ($m\geq2$). Let $0<q<\frac{1}{k\cdot\left(l+1\right)}$
be a probability that will be defined later. Let \emph{$s^{q}$} be
the strategy that induces an agent to defect with a probability of
$q$ iff the agent observes $m=1$, to defect for sure if she observes
$m\geq2$, and to cooperate for sure if she observes $m=0$. Let $\left(S^{C},\lambda\right)$
be an arbitrary distribution of commitments. We will show that there
exist a converging sequence of commitment levels $\epsilon_{n}\rightarrow0$
and converging sequences of steady states 
\[
\psi{}_{n}\equiv\left(\left\{ s^{1},s^{2}\right\} ,\sigma_{n}=\left(q_{n},1-q_{n}\right),\theta_{n}\right)\rightarrow_{n\rightarrow\infty}\psi^{*}\equiv\left(\left\{ s^{1},s^{2}\right\} ,\left(q,1-q\right),\theta\equiv0\right),
\]
and 
\[
\psi'_{n}\equiv\left(\left\{ s^{q_{n}}\right\} ,\theta_{n}'\right)\rightarrow_{n\rightarrow\infty}\psi'^{*}\equiv\left(\left\{ \emph{\ensuremath{s^{q}}}\right\} ,\theta'\equiv0\right),
\]
 such that either (1) for each $n$ the steady state $\psi{}_{n}$
is a Nash equilibrium of $\left(\left(G,k\right),\left(S^{C},\lambda\right),\epsilon_{n}\right)$,
or (2) for each $n$ the steady state $\psi'_{n}$ is a Nash equilibrium
of $\left(\left(G,k\right),\left(S^{C},\lambda\right),\epsilon_{n}\right)$. 

Fix an $n\geq1$ such that $\epsilon_{n}$ is sufficiently small.
(Exactly what counts as sufficiently small will become clear below.)
In what follows, we calculate a number of probabilities while relying
on the fact that $\epsilon_{n}<<1$. Thus we neglect terms of $O\left(\epsilon_{n}\right)$
(resp., $O\left(\epsilon_{n}^{2}\right)$) when the leading term is
$O\left(1\right)$ (resp., $O\left(\epsilon_{n}\right)$). The calculations
give the same results for $\psi_{n}$ as for $\psi'_{n}$. Since we
are looking for consistent signal profiles $\theta_{n}$ and $\theta'_{n}$
such that $\theta_{n}\rightarrow_{n\rightarrow\infty}\theta\equiv0$
and $\theta'_{n}\rightarrow_{n\rightarrow\infty}\theta'\equiv0$,
we assume that $\left(\theta_{n}\right)_{s_{i}}\left(0\right)=1-O\left(\epsilon_{n}\right)$
for each $s_{i},s_{j}\in\left\{ s^{1},s^{2}\right\} $ in $\psi_{n}$
and assume that $\left(\theta'_{n}\right)_{\emph{\ensuremath{s^{q}}}}\left(0\right)=1-O\left(\epsilon_{n}\right)$
in $\psi'_{n}$. 

We begin by confirming that indeed there exist consistent signal profiles
$\theta_{n}$ and $\theta'_{n}$ in which the normal agents almost
always cooperate (the argument also implies that the steady states
$\psi{}_{n}$ and $\psi'{}_{n}$ satisfy the robustness refinement
defined in Appendix \ref{subsec:Robustness}). Consider a perturbed
signal profile $\theta\in O_{\left(S_{n}^{N}\cup S^{C}\right)}$.
Recall (Appendix \ref{subsec:Robustness}) that $\alpha_{\sigma_{n}}\left(\theta\right)\left(d\right)$
is the ($\sigma_{n}$-weighted) average of the distributions of actions
that induce signals distributed according to the signal profile $\theta$
for the normal agents, i.e., 

\[
\alpha_{\sigma_{n}}\left(\theta\right)\left(d\right)=q_{n}\cdot\alpha\left(\theta{}_{s^{1}}\right)\left(d\right)+\left(1-q_{n}\right)\cdot\alpha\left(\theta_{s^{2}}\right)\left(d\right)\,\,\,\,\,\,(\alpha_{\sigma_{n}}\left(\theta\right)\left(d\right)=\alpha\left(\theta_{s^{q}}\right)\left(d\right)).
\]

The (possibly inconsistent) ``old'' perturbed signal profile $\theta$
and the strategy distribution of the incumbents jointly induce a ``new''
signal profile $f_{\left(1-\epsilon_{n}\right)\cdot\sigma_{n}+\epsilon_{n}\cdot\lambda}\left(\theta\right)$
(where the dynamic mapping between states, $f$, is as defined in
Appendix \ref{subsec:Dynamic-Mapping-Between}) . The average defection
probability of a normal agent in this ``new'' signal profile is
bounded by the following inequality: 
\begin{equation}
\alpha_{\sigma_{n}}\left(f_{\left(1-\epsilon_{n}\right)\cdot\sigma_{n}+\epsilon_{n}\cdot\lambda}\left(\theta\right)\right)\left(d\right)<\left(1-\epsilon_{n}\right)\cdot\left(q_{n}\cdot k\cdot\alpha_{\sigma_{n}}\left(\theta\right)\left(d\right)+\left(\begin{array}{c}
k\\
2
\end{array}\right)\cdot\left(\alpha_{\sigma_{n}}\left(\theta\right)\left(d\right)\right)^{2}\right)+\epsilon_{n}.\label{eq:eta-bound}
\end{equation}
This is so because a normal agent, when being matched with a normal
partner (which happens with a probability of $\left(1-\epsilon_{n}\right)$)
defects with an average probability of $q_{n}$ when she observes
a single defection (which happens with a probability strictly less
than $k\cdot\alpha_{\sigma_{n}}\left(\theta\right)$ ), and defects
for sure when she observes at least two defections (which happens
with a probability strictly less than $\left(\begin{array}{c}
k\\
2
\end{array}\right)\cdot\left(\alpha_{\sigma_{n}}\left(\theta\right)\right)^{2}$). Consider the parabolic equation, which is based on substituting
\[
x=\alpha_{\sigma_{n}}\left(\theta\right)\left(d\right)=\alpha_{\sigma_{n}}\left(f_{\left(1-\epsilon_{n}\right)\cdot\sigma_{n}+\epsilon_{n}\cdot\lambda}\left(\theta\right)\right)\left(d\right)
\]
 in (\ref{eq:eta-bound}), and changing the inequality into an equality:
\[
x=\left(1-\epsilon_{n}\right)\cdot\left(q_{n}\cdot k\cdot x+\left(\begin{array}{c}
k\\
2
\end{array}\right)\cdot x^{2}\right)+\epsilon_{n}\,\Leftrightarrow\,0=\left(\begin{array}{c}
k\\
2
\end{array}\right)\cdot x^{2}-\left(1-\left(1-\epsilon_{n}\right)\cdot q_{n}\cdot k\right)\cdot x+\epsilon_{n}.
\]
Recall that a parabolic equation $A\cdot x^{2}-B\cdot x+C=0$ with
$A,B,C>0$ and $C<<A,B$ has two positive solutions, the smaller of
which is 
\begin{align*}
x_{1} & =\frac{B-\sqrt{B^{2}-4\cdot A\cdot C}}{2\cdot A}\approx\frac{B-\sqrt{B^{2}-2\cdot B\cdot\frac{2\cdot A\cdot C}{B}+\left(\frac{2\cdot A\cdot C}{B}\right)^{2}}}{2\cdot A}=\\
 & \frac{B-\left(B-\frac{2\cdot A\cdot C}{B}\right)}{2\cdot A}=\frac{\frac{2\cdot A\cdot C}{B}}{2\cdot A}=\frac{C}{B}=\frac{\epsilon_{n}}{1-\left(1-\epsilon_{n}\right)\cdot q_{n}\cdot k}=\kappa_{n}\cdot\epsilon_{n},
\end{align*}
where the penultimate equality is derived by substituting $C=\epsilon_{n}$
and $B=1-\left(1-\epsilon_{n}\right)\cdot q_{n}\cdot k$, and the
last equality is derived by defining $\kappa_{n}=\frac{1}{1-\left(1-\epsilon_{n}\right)\cdot q_{n}\cdot k}$.
Let $\kappa=sup_{n}\kappa_{n}<\infty$. The upper bound $\kappa$
is finite due to the fact that $q_{n}\rightarrow q$, $k\cdot q<\frac{1}{l+1}$
and $\epsilon_{n}\rightarrow\epsilon$. The definition of $x_{1}=\kappa_{n}\cdot\epsilon_{n}$
implies that 
\[
\alpha_{\sigma_{n}}\left(\theta\right)\left(d\right)\leq\kappa_{n}\cdot\epsilon_{n}\,\,\Rightarrow\,\,\alpha_{\sigma_{n}}\left(f_{\left(1-\epsilon_{n}\right)\cdot\sigma_{n}+\epsilon_{n}\cdot\lambda}\left(\theta\right)\right)\left(d\right)<\kappa_{n}\cdot\epsilon_{n},
\]
which immediately implies the following inequality (which implies
the robustness property of the steady state as defined in Appendix
\ref{subsec:Robustness}):
\[
\alpha_{\sigma_{n}}\left(\theta\right)\left(c\right)\geq1-\kappa_{n}\cdot\epsilon_{n}\,\,\Rightarrow\,\,\alpha_{\sigma_{n}}\left(f_{\left(1-\epsilon_{n}\right)\cdot\sigma_{n}+\epsilon_{n}\cdot\lambda}\left(\theta\right)\right)\left(c\right)>1-\kappa_{n}\cdot\epsilon_{n}.
\]

Let $O_{\left(\left\{ s^{1},s^{2}\right\} \cup S_{C},x_{1}\right)}$
($O_{\left(s_{q_{n}}\cup S_{C},x_{1}\right)}$) be the set of signal
profiles $\theta$ defined over $\left\{ s^{1},s^{2}\right\} \cup S_{C}$
($s_{q_{n}}\cup S_{C}$) and satisfying $\alpha_{\sigma_{n}}\left(\theta\right)\left(d\right)\leq x_{1}$.
Observe that $O_{\left(\left\{ s^{1},s^{2}\right\} \cup S_{C},x_{1}\right)}(O_{\left(s_{q_{n}}\cup S_{C},x_{1}\right)})$
is a convex and compact subset of a Euclidean space, and that the
mapping $f_{\left(1-\epsilon_{n}\right)\cdot\sigma_{n}+\epsilon_{n}\cdot\lambda}\left(\theta\right)$
is continuous. Brouwer's fixed-point theorem implies that the mapping
$f_{\left(1-\epsilon_{n}\right)\cdot\sigma_{n}+\epsilon_{n}\cdot\lambda}\left(\theta\right)$
has a fixed point $\theta_{n}$ ($\theta'_{n}$) satisfying $\alpha_{\sigma_{n}}\left(\theta_{n}\right)\left(d\right)<x_{1}=O\left(\epsilon_{n}\right)$
($\alpha_{\sigma_{n}}\left(\theta'_{n}\right)\left(d\right)\leq x_{1}=O\left(\epsilon_{n}\right)$),
which is a consistent signal profile in which the normal agents almost
always cooperate. 

For each incumbent strategy $s$, let $Pr\left(m=1|s\right)$ ($Pr\left(m\geq2|s\right)$)
denote the probability of observing exactly one defection (at least
two defections) conditional on the partner following strategy $s$.
Let $Pr\left(m=1\right)$ and $Pr\left(m\geq2\right)$ be the corresponding
unconditional probabilities. 

The assumption that $\theta_{n}\rightarrow_{n\rightarrow\infty}\theta\equiv0$
and $\theta'_{n}\rightarrow_{n\rightarrow\infty}\theta'\equiv0$ implies
that agents are very likely to observe the signal $m=0$ (i.e., zero
defections) when being matched with a random partner. Formally: 
\[
Pr\left(m=0\right)=\left(1-O\left(\epsilon_{n}\right)\right)^{k}=1-O\left(\epsilon_{n}\right).
\]
The conditional probabilities of observing $m=0$, $m=1$, and $m\geq2$,
for all $s\in S_{n}^{N}\cup S^{C}$, are
\[
Pr\left(m=0|s\right)=\left(s_{0}\left(c\right)\right)^{k}+O\left(\epsilon_{n}\right),
\]
\[
Pr\left(m=1|s\right)=k\cdot s_{0}\left(d\right)\cdot\left(s_{0}\left(c\right)\right)^{k-1}+O\left(\epsilon_{n}\right),
\]
\[
Pr\left(m\geq2|s\right)=1-Pr\left(m=0|s\right)-Pr\left(m=1|s\right).
\]

Let $S_{n}^{N}=\left\{ s^{1},s^{2}\right\} $ in $\psi_{n}$ and $S_{n}^{N}=\left\{ s^{q_{n}}\right\} $
in $\psi'_{n}$. Given signal $m$, let $Pr\left(m|S_{n}^{N}\right)$
denote the probability of observing signal $m$, conditional on the
partner following a normal strategy. Specifically, in the heterogeneous
state $\psi_{n}$ (with two normal strategies), this conditional probability
is given by 
\[
Pr\left(m|S_{n}^{N}\right)=q\cdot Pr\left(m|s^{1}\right)+\left(1-q\right)\cdot Pr\left(m|s^{2}\right).
\]
Furthermore, it follows (from the expressions for $Pr\left(m=0|s\right)$,
$Pr\left(m=1|s\right)$, and $Pr\left(m\geq2|s\right)$) that 
\[
Pr\left(m=0|S_{n}^{N}\right)=1-O\left(\epsilon_{n}\right),\,\,\,\,\,\,Pr\left(m=1|S_{n}^{N}\right)=O\left(\epsilon_{n}\right)\,\,\,\,\,Pr\left(m\geq2|S_{n}^{N}\right)=O\left(\epsilon_{n}^{2}\right).
\]
Next we calculate the probability that a normal agent (Alice) generates
a signal that contains a single defection. This happens with probability
one if exactly one of the $k$ interactions sampled from Alice's past
was such that Alice observed her partner in that interaction to have
defected at least twice (which implies that her partner is most likely
to have been a committed agent). This happens with probability $q_{n}$
if exactly one of the $k$ interactions sampled from Alice's past
was such that Alice observed her partner (who might have been either
a committed or a normal agent) to have defected exactly once: 
\begin{eqnarray*}
Pr\left(m=1|S_{n}^{N}\right) & = & k\cdot\sum_{s\in S^{C}}\epsilon_{n}\cdot\lambda\left(s\right)\cdot\left(Pr\left(m\geq2|s\right)+q_{n}\cdot Pr\left(m=1|s\right)\right).\\
 &  & +k\cdot\left(1-\epsilon_{n}\right)\cdot\left[q_{n}\cdot Pr\left(m=1|S_{n}^{N}\right)+Pr\left(m\geq2|S_{n}^{N}\right)\right]\\
 &  & +O\left(\epsilon_{n}^{2}\right).
\end{eqnarray*}
The final term $O\left(\epsilon_{n}^{2}\right)$ comes from the very
small probability of the partner observing a normal agent to defect
twice. Since $Pr\left(m=1|S_{n}^{N}\right)=O\left(\epsilon_{n}\right)$
and $Pr\left(m\geq2|S_{n}^{N}\right)=O\left(\epsilon_{n}^{2}\right)$,
this can be simplified (neglecting $O\left(\epsilon_{n}^{2}\right)$)
and rearranged to obtain 
\begin{equation}
Pr\left(m=1|S_{n}^{N}\right)=\frac{k\cdot\epsilon_{n}\cdot\sum_{s\in\mathcal{\mathit{S^{C}}}}\lambda\left(s\right)\cdot\left(Pr\left(m\geq2|s\right)+q_{n}\cdot Pr\left(m=1|s\right)\right)}{1-k\cdot q},\label{eq:1-kq}
\end{equation}
which is well defined and $O\left(\epsilon_{n}\right)$ as long as
$q_{n}<\nicefrac{1}{k}$. We can now calculate the unconditional probabilities:
\[
Pr\left(m=1\right)=\epsilon_{n}\cdot\sum_{s\in S^{C}}\lambda\left(s\right)\cdot Pr\left(m=1|s\right)+Pr\left(m=1|S_{n}^{N}\right)+O\left(\epsilon_{n}^{2}\right),
\]
\begin{eqnarray*}
Pr\left(m\geq2\right) & = & \epsilon_{n}\cdot\sum_{s\in S^{C}}\lambda\left(s\right)\cdot Pr\left(m\geq2|s\right)+\left(1-\epsilon_{n}\right)\cdot Pr\left(m\geq2|S_{n}^{N}\right)\\
 & = & \epsilon_{n}\cdot\sum_{s\in S^{C}}\lambda\left(s\right)\cdot Pr\left(m\geq2|s\right)+O\left(\epsilon_{n}^{2}\right).
\end{eqnarray*}
By using Bayes' rule we can calculate the conditional probability
that the partner uses strategy $s\in S^{C}$ as a function of the
observed signal:
\begin{eqnarray*}
Pr\left(s|m=0\right) & = & \frac{\epsilon_{n}\cdot\lambda\left(s\right)\cdot Pr\left(m=0|s\right)}{Pr\left(m=0\right)},\\
Pr\left(s|m=1\right) & = & \frac{\epsilon_{n}\cdot\lambda\left(s\right)\cdot Pr\left(m=1|s\right)}{Pr\left(m=1\right)},\\
Pr\left(s|m\geq2\right) & = & \frac{\epsilon_{n}\cdot\lambda\left(s\right)\cdot Pr\left(m\geq2|s\right)}{Pr\left(m\geq2\right)}.
\end{eqnarray*}
Note that 
\begin{eqnarray*}
\sum_{s\in S^{C}}Pr\left(s|m=0\right) & = & \frac{\epsilon_{n}\cdot\sum_{s\in S^{C}}\lambda\left(s\right)\cdot\left(s_{0}\left(c\right)\right)^{k}}{1-O\left(\epsilon_{n}\right)}=O\left(\epsilon_{n}\right).
\end{eqnarray*}
From Eq. (\ref{eq:1-kq}) we have 
\begin{eqnarray*}
\sum_{s\in S_{n}^{N}}\sigma\left(s\right)\cdot Pr\left(m=1|s\right) & = & Pr\left(m=1|S_{n}^{N}\right)=\frac{k\cdot\epsilon_{n}\cdot\sum_{s\in\mathcal{\mathit{S^{C}}}}\lambda\left(s\right)\cdot\left(Pr\left(m\geq2|s\right)+q\cdot Pr\left(m=1|s\right)\right)}{1-k\cdot q_{n}}.
\end{eqnarray*}
We use this to obtain, by Bayes' rule, 
\begin{eqnarray*}
\sum_{s\in S^{C}}Pr\left(s|m=1\right) & = & \frac{\epsilon_{n}\cdot\sum_{s\in S^{C}}\lambda\left(s\right)\cdot Pr\left(m=1|s\right)}{\epsilon_{n}\cdot\sum_{s\in S^{C}}\lambda\left(s\right)\cdot Pr\left(m=1|s\right)+\frac{k\cdot\epsilon_{n}\cdot\sum_{s\in\mathcal{\mathit{S^{C}}}}\lambda\left(s\right)\cdot\left(Pr\left(m\geq2|s\right)+q_{n}\cdot Pr\left(m=1|s\right)\right)}{1-k\cdot q_{n}}+O\left(\epsilon_{n}^{2}\right)}\\
 & = & \frac{\sum_{s\in S^{C}}\lambda\left(s\right)\cdot Pr\left(m=1|s\right)}{\sum_{s\in S^{C}}\lambda\left(s\right)\cdot Pr\left(m=1|s\right)+\frac{k\cdot\sum_{s\in\mathcal{\mathit{S^{C}}}}\lambda\left(s\right)\cdot\left(Pr\left(m\geq2|s\right)+q_{n}\cdot Pr\left(m=1|s\right)\right)}{1-k\cdot q_{n}}+O\left(\epsilon_{n}^{2}\right).}
\end{eqnarray*}
Note that the terms $\sum_{s\in S^{C}}\lambda\left(s\right)\cdot Pr\left(m=1|s\right)$
and $\sum_{s\in S^{C}}\lambda\left(s\right)\cdot\left(Pr\left(m\geq2|s\right)\right)$
do not vanish as $\epsilon_{n}\rightarrow0$. Moreover, we will see
below (Eqs. (\ref{eq:mu-calculation}) and (\ref{eq:mu-q-equation}))
that this implies that $q_{n}$ also does not vanish as $\epsilon_{n}\rightarrow0$.
Together, these observations imply that there are numbers $a,b\in\left(0,1\right)$
such that, for all $n$, it is the case that 
\begin{equation}
0<a<\sum_{s\in S^{C}}Pr\left(s|m=1\right)<b<1.\label{eq:a<Pr(Sc|m=00003D1)<b}
\end{equation}
Furthermore

\[
\sum_{s\in S^{C}}Pr\left(s|m\geq2\right)=\frac{1}{1+\frac{\sum_{s\in S_{n}^{N}}\sigma\left(s\right)\cdot Pr\left(m\geq2|s\right)}{\epsilon_{n}\cdot\sum_{s\in S^{C}}\lambda\left(s\right)\cdot Pr\left(m\geq2|s\right)}}=\frac{1}{1+\frac{O\left(\epsilon_{n}^{2}\right)}{O\left(\epsilon_{n}\right)}}=\frac{1}{1+O\left(\epsilon_{n}\right)}.
\]
Hence for a sufficiently large $n$, the more defections there are
in the observed signal, the higher is the conditional probability
that the partner is committed: 
\[
\sum_{s\in S^{C}}Pr\left(s|m=0\right)<\sum_{s\in S^{C}}Pr\left(s|m=1\right)<\sum_{s\in S^{C}}Pr\left(s|m\geq2\right).
\]

Let $Pr\left(S_{n}^{N}|m=1\right)=\sum_{s\in S_{n}^{N}}Pr\left(s|m=1\right)$
denote the conditional probability that the partner follows a normal
strategy conditional on the agent observing signal $m=1$. Eq. (\ref{eq:a<Pr(Sc|m=00003D1)<b})
implies that there are numbers $a',b'\in\left(0,1\right)$ such that,
for all $n$, it is the case that $0<a'<Pr\left(S_{n}^{N}|m=1\right)<b'<1$
(because $Pr\left(S_{n}^{N}|m=1\right)+\sum_{s\in S^{C}}Pr\left(s|m=1\right)=1$).

Let $\mu_{n}$ be the probability that a random partner defects conditional
on a player observing signal $m=1$ about the partner, and conditional
on the partner observing the signal $m=0$: 
\begin{equation}
\mu_{n}=\sum_{s\in S^{C}}Pr\left(s|m=1\right)\cdot s_{0}\left(d\right)+O\left(\epsilon_{n}\right).\label{eq:mu-calculation}
\end{equation}

Eq. (\ref{eq:mu-calculation}) defines $\mu_{n}$ as a strictly decreasing
function of $q_{n}.$ To see this, note that the term $s_{0}\left(d\right)$
does not depend on $q_{n}$, and in $Pr\left(s|m=1\right)=\frac{\epsilon_{n}\cdot\lambda\left(s\right)\cdot Pr\left(m=1|s\right)}{Pr\left(m=1\right)}$
the numerator does not depend on $q_{n}$, whereas the term $Pr\left(m=1\right)$
is increasing in $q_{n}$.

Next we calculate the value of $q_{n}$ that balances the payoff of
both actions after a player observes a single defection (neglecting
terms of $O\left(\epsilon_{n}^{2}\right)$). The LHS of the following
equation represents the player's direct gain from defecting when she
observes a single defection, while the RHS represents the player's
indirect loss induced by partners who defect as a result of observing
these defections:
\begin{equation}
\textrm{Pr}\left(m=1\right)\cdot\left(\mu_{n}\cdot l+\left(1-\mu_{n}\right)\cdot g\right)=\textrm{Pr}\left(m=1\right)\cdot\left(k\cdot q\cdot\left(l+1\right)+O\left(\epsilon_{n}\right)\right)\Rightarrow q_{n}=\frac{\mu_{n}\cdot l+\left(1-\mu_{n}\right)\cdot g}{k\cdot\left(l+1\right)}+O\left(\epsilon_{n}\right).\label{eq:mu-q-equation}
\end{equation}
Note that Eq. (\ref{eq:mu-q-equation}) defines $q_{n}$ as a strictly
increasing function of $\mu_{n}$. This implies that there are unique
values of $q_{n}$ and $\mu_{n}$ satisfying $\frac{g}{k\cdot\left(l+1\right)}<q_{n}<\frac{l}{k\cdot\left(l+1\right)}<\frac{1}{k}$
and $0<\mu_{n}<1$, which jointly solve Eqs. (\ref{eq:mu-calculation})
and (\ref{eq:mu-q-equation}). This pair of parameters balances the
payoff of both actions when a player observes a signal $m=1$. Note
that sequences of $\left(q_{n}\right)_{n}\rightarrow q$ and $\left(\mu_{n}\right)_{n}\rightarrow\mu$
converge to the values that solve the above equations when ignoring
the terms that are $O\left(\epsilon_{n}\right)$. 

Observe that defection is the unique best reply when a player observes
at least two defections. The direct gain from defecting is larger
than the LHS of Eq. (\ref{eq:mu-q-equation}), and the indirect loss
is still given by the RHS of Eq. (\ref{eq:mu-q-equation}). The reason
that the direct gain is larger is that normal partners almost never
defect twice or more (the probability is $O\left(\epsilon_{n}^{2}\right)$),
and thus the partner is most likely committed and will defect  with
a probability that is higher than $\mu_{n}$ (since $\mu_{n}$ also
gives weight to normal strategies that are most likely to cooperate).
More generally, note that given that the normal agents almost always
cooperate, the average probability of defection of each agent who
follows strategy $s$ is $s_{0}\left(d\right)+O\left(\epsilon_{n}\right)$.
This implies that for a sufficient small $\epsilon_{n}$, the higher
$m$ is, the higher the partner's value $s_{0}\left(d\right)$ is
likely to be. Hence the higher $m$ is, the higher the probability
is that the partner will defect against a normal agent. Thus the direct
gain from defection is increasing in the signal $m$ that the normal
agent observes about her partner. (A formal detailed proof of this
statement is available upon request.)

Next, consider a deviator (Alice) who defects with a probability of
$\alpha>0$ after she observes $m=0$. In what follows we calculate
Alice's expected payoff as a function of $\alpha$ in any post-deviation
stable state, neglecting terms of $O\left(\epsilon_{n}\right)$ throughout
the calculation. Note that Alice's partner observes signal $m=1$
with a probability of $k\cdot\alpha\cdot\left(1-\alpha\right)^{k-1}$,
and observes signal $m\ge2$ with a probability of $1-\left(1-\alpha\right)^{k}-k\cdot\alpha\cdot\left(1-\alpha\right)^{k-1}$.
This implies that the mean probability that a normal partner defects
against a mutant is 
\[
h\left(\alpha\right):=\left(k\cdot\alpha\cdot\left(1-\alpha\right)^{k-1}\right)\cdot q+1-\left(1-\alpha\right)^{k}-k\cdot\alpha\cdot\left(1-\alpha\right)^{k-1}=1-\left(1-\alpha\right)^{k-1}\left(1-\alpha+k\cdot\alpha\cdot\left(1-q\right)\right).
\]
Thus the expected payoff of the mutant is
\begin{eqnarray*}
\pi\left(\alpha\right): & = & \left(1-h\left(\alpha\right)\right)\cdot\alpha\cdot\left(1+g\right)+\left(1-h\left(\alpha\right)\right)\cdot\left(1-\alpha\right)-h\left(\alpha\right)\cdot\left(1-\alpha\right)\cdot l\\
 & = & 1+\alpha\cdot g-h\left(\alpha\right)\cdot\left(1+\left(1-\alpha\right)\cdot l+\alpha\cdot g\right).
\end{eqnarray*}
Direct numeric calculation of $\frac{\partial\pi\left(\alpha\right)}{\partial\alpha}$
reveals that $\pi\left(\alpha\right)$ is strictly decreasing in $\alpha$
for each $q>\frac{g}{k\cdot\left(l+1\right)}$. Thus any deviator
with $\alpha>0$ earns strictly less than the incumbents (who have
$\alpha=0$). 

We have now shown that the best reply is $c$ after observing $m=0$
and $d$ after observing $m\geq2$. After observing $m=1$ both $c$
and $d$ are best replies provided that $q$ has the required value.
That is, we know what the aggregate probability of defection after
a player observes $m=1$ has to be in equilibrium. However, we do
not know whether mixing will occur at the individual level. We now
turn to this question. 

Let $\chi$ be the probability that a random partner defects conditional
on both the agent and the partner observing a single defection (in
the limit as $\epsilon_{n}\rightarrow0$): 
\[
\chi=\lim_{n\rightarrow\infty}\left(\sum_{s\in S^{C}}Pr\left(s|m=1\right)\cdot s^{1}\left(d\right)+Pr\left(S_{n}^{N}|m=1\right)\cdot q\right).
\]
We conclude by showing that if $\chi>\mu$ ($\chi<\mu$), then $\psi^{*}$
($\psi'^{*}$) is a perfect equilibrium. This is so because if $\chi>\mu$
($\chi<\mu$), then conditional on a normal agent observing a single
defection, the partner is more (less) likely to defect the higher
the probability with which the agent defects when she observes a single
defection (because then it is more likely that the partner observes
a single defection rather than only cooperation). This implies that
when a player observes a single defection, the higher the agent's
own defection probability is, the more profitable defection is (recall
that the higher the probability is of the defection of the partner,
the higher the direct gain from defection, whereas the indirect loss
is independent of the partner's behavior). That is, an agent's payoff
is a strictly convex (concave) function of the agent's defection probability
conditional on him observing a single defection. This implies that
a deviator who mixes on the individual level (i.e., defects with probabilities
different from $q$) is outperformed when $\chi>\mu$ ($\chi<\mu$)). 

Note that the normal agents are more likely to defect against a partner
who is more likely to defect when she observes a single defection.
This implies that when focusing only on normal partners, the induced
level of $\chi$ is larger than the induced level of $\mu$. It is
only the committed agents who may induce the opposite inequality (namely,
$\chi<\mu$). Thus, if in the limit as $\epsilon\rightarrow0$ the
equality $\chi=\mu$ holds, then it must be that for any positive
small share of committed agents $\epsilon_{n}$, it is the case that
$\chi_{n}<\mu_{n}$, which implies by the argument above that the
state $\psi'_{n}$ is a Nash equilibrium.
\begin{rem}
\label{rem:evolutionarily-stability}The above argument shows that
when $\chi<\mu$, each state $\psi'_{n}$ is a \emph{strictly} perfect
equilibrium (any deviator who follows a strategy different from $s^{q_{n}}$
obtains a strictly lower payoff). In the opposite case of $\chi>\mu$,
one can show that an agent who follows strategy $s_{i}$ achieves
a higher payoff than an agent who follows $s_{-i}$, conditional on
the partner following $s_{i}$. This implies that the mixed equilibrium
between the strategies of $s^{1}$ and $x_{2}$ is Hawk-Dove-like,
and that the state $\psi_{n}$ is evolutionarily stable (see Appendix
\ref{sec:Evolutionary-Stability}). This shows that cooperation is
robust also to joint deviation of a small group of agents, and that
it satisfies the refinement of evolutionary stability defined in Appendix
\ref{sec:Evolutionary-Stability} (namely, cooperation is a strictly
perfect evolutionarily stable action). 
\end{rem}

\subsection{Proof of Proposition \ref{prop:Let--be-single-action-complex-result}
(Observing a Single Action)\label{subsec:Proof-of-Proposition}}

Arguments and pieces of notation that are analogous to the ones used
in the proof of Theorem \ref{thm:cooperation-defensive-PDs} are presented
in brief or skipped. Let $s^{c}\equiv c$ be the strategy that always
cooperates. The same arguments as in Theorem \ref{thm:cooperation-defensive-PDs}
show that the only possible candidates for perfect equilibria that
support full cooperation are steady states of the form $\psi=\left(\left\{ s^{1},s^{c}\right\} ,\left(q,1-q\right),\theta\equiv0\right)$
or $\psi'=\left(\left\{ \emph{\ensuremath{s^{q}}}\right\} ,\theta'\equiv0\right)$. 

Consider a perturbed environment $\left(\left(G_{PD},k\right),\left(S^{C},\lambda\right),\epsilon\right)$
where $\epsilon>0$ is sufficiently small. In what follows: (1) for
the case of $g\leq\beta_{C,\lambda}$ we characterize a Nash equilibrium
of this perturbed environment that is within a distance of $O\left(\epsilon\right)$
from either $\psi$ or $\psi'$, and (2) we show that no such Nash
equilibrium exists for the case of $g>\beta_{C,\lambda}$. 

Consider a steady state that is within a distance of $O\left(\epsilon\right)$
from either $\psi$ or $\psi'$. The fact that the behavior in the
steady state is close to always cooperating (i.e., to $\theta\equiv0$)
implies that the probability of observing $m=1$ conditional on the
partner following a commitment strategy $s\in S^{C}$ is: 
\[
Pr\left(m=1|s\right)=s_{0}\left(d\right)+O\left(\epsilon\right).
\]
 Similarly, the probability of observing $m=1$ conditional on the
partner being normal is

\[
Pr\left(m=1|S_{n}^{N}\right)=q\cdot\left(\epsilon\cdot\sum_{s\in S^{C}}\lambda\left(s\right)\cdot s_{0}\left(d\right)+\left(1-\epsilon\right)\cdot Pr\left(m=1|S_{n}^{N}\right)\right)+O\left(\epsilon^{2}\right)\Rightarrow
\]

\[
Pr\left(m=1|S_{n}^{N}\right)=\frac{\epsilon\cdot q\cdot\sum_{s\in S^{C}}\lambda\left(s\right)\cdot s_{0}\left(d\right)}{1-q}+O\left(\epsilon^{2}\right).
\]
 By using Bayes' rule we can calculate the probability that the partner
uses strategy $s\in S^{C}$ conditional on observing $m=1$:
\[
Pr\left(s|m=1\right)=\frac{\epsilon_{n}\cdot\lambda\left(s\right)\cdot Pr\left(m=1|s\right)}{Pr\left(m=1\right)}=\frac{\epsilon\cdot\lambda\left(s\right)\cdot s_{0}\left(d\right)}{\epsilon\cdot\left(\sum_{s\in S^{C}}\lambda\left(s\right)\cdot s_{0}\left(d\right)+\frac{q\cdot\sum_{s\in S^{C}}\lambda\left(s\right)\cdot s_{0}\left(d\right)}{1-q}\right)}+O\left(\epsilon\right)\Rightarrow
\]

\[
Pr\left(s|m=1\right)=\frac{\left(1-q\right)\cdot\lambda\left(s\right)\cdot s_{0}\left(d\right)}{\sum_{s\in S^{C}}\lambda\left(s\right)\cdot s_{0}\left(d\right)}+O\left(\epsilon\right).
\]
Let $\mu$ be the probability that a random partner defects conditional
on an agent observing signal $m=1$ about the partner, and conditional
on the partner observing the signal $m=0$ about the agent. (Note
that only committed partners defect with positive probability when
observing $m=0$.) 
\begin{equation}
\mu=\sum_{s\in S^{C}}Pr\left(s|m=1\right)\cdot s_{0}\left(d\right)+O\left(\epsilon\right)=\left(1-q\right)\cdot\frac{\sum_{s\in S^{C}}\lambda\left(s\right)\cdot\left(s_{0}\left(d\right)\right)^{2}}{\sum_{s\in S^{C}}\lambda\left(s\right)\cdot s_{0}\left(d\right)}+O\left(\epsilon\right)=\left(1-q\right)\cdot\beta_{\left(S^{C},\lambda\right)}+O\left(\epsilon\right).\label{eq:mu-calculation-1}
\end{equation}
Next we calculate the value of $q$ that balances the payoff of both
actions after a player observes a single defection. The LHS of the
following equation represents the player's direct gain from defecting
when she observes a single defection, while the RHS represents the
player's indirect loss induced by future partners who defect as a
result of observing these defections:
\begin{eqnarray}
\textrm{Pr}\left(m=1\right)\cdot\left(\mu\cdot l+\left(1-\mu\right)\cdot g\right) & +O\left(\epsilon\right)= & \textrm{Pr}\left(m=1\right)\cdot\left(q\cdot\left(l+1\right)+O\left(\epsilon\right)\right)\Rightarrow\label{eq:mu-q-equation-1}\\
 &  & q=\frac{\mu\cdot l+\left(1-\mu\right)\cdot g}{l+1}+O\left(\epsilon\right)=\frac{g+\mu\cdot\left(l-g\right)}{l+1}+O\left(\epsilon\right).\label{eq:q-as-function-of-mu}
\end{eqnarray}
Substituting (\ref{eq:mu-calculation-1}) in (\ref{eq:q-as-function-of-mu})
yields
\[
q=\frac{g+\left(1-q\right)\cdot\left(l-g\right)\cdot\beta_{\left(S^{C},\lambda\right)}}{l+1}+O\left(\epsilon\right)\Rightarrow q\cdot\left(l+1\right)=g+\left(1-q\right)\cdot\left(l-g\right)\cdot\beta_{\left(S^{C},\lambda\right)}+O\left(\epsilon\right)
\]
\[
\Rightarrow q=\frac{g+\left(l-g\right)\cdot\beta_{\left(S^{C},\lambda\right)}}{l+1+\left(l-g\right)\cdot\beta_{\left(S^{C},\lambda\right)}}+O\left(\epsilon\right).
\]

Consider a deviator (Alice) who always defects. Normal partners of
Alice cooperate with a probability of $1-q$. This implies that Alice
gets an expected payoff of $\left(1+g\right)\cdot\left(1-q\right)$,
while the normal agents each get a payoff of $1+O\left(\epsilon\right)$.
Alice is outperformed iff (neglecting terms of $O\left(\epsilon\right)$):
\[
\left(1+g\right)\cdot\left(1-q\right)\leq1\Leftrightarrow q\geq\frac{g}{1+g}\Leftrightarrow\frac{g+\left(l-g\right)\cdot\beta_{\left(S^{C},\lambda\right)}}{l+1+\left(l-g\right)\cdot\beta_{\left(S^{C},\lambda\right)}}\geq\frac{g}{1+g}
\]
\[
\Leftrightarrow\left(1+g\right)\cdot\left(g+\left(l-g\right)\cdot\beta_{\left(S^{C},\lambda\right)}\right)\geq g\cdot\left(l+1+\left(l-g\right)\cdot\beta_{\left(S^{C},\lambda\right)}\right)
\]
\[
\Leftrightarrow g^{2}+\left(l-g\right)\cdot\beta_{\left(S^{C},\lambda\right)}+\geq g\cdot l\,\Leftrightarrow\,g\cdot\left(l-g\right)\leq\left(l-g\right)\cdot\beta_{\left(S^{C},\lambda\right)}\,\Leftrightarrow\,g\leq\beta_{\left(S^{C},\lambda\right)}.
\]
Thus, the steady state can be a Nash equilibrium only if $g\leq\beta_{\left(S^{C},\lambda\right)}$.
It is relatively straightforward to show that if $g\leq\beta_{\left(S^{C},\lambda\right)}$,
then a deviator who defects with probability $\alpha$ when observing
$m=0$ is outperformed. The remaining steps of the proof are as in
the proof of part 2 of Theorem \ref{thm:cooperation-defensive-PDs},
and are omitted for brevity.

\subsection{Proof of Theorem \ref{thm:stable-cooperation-observing-conflicts}
(Observing Conflicts) \label{subsec:Proof-of-Theorem-stable-cooperation-action-profile}}

The proof of part 1(a) is analogous to Theorem \ref{thm:cooperation-defensive-PDs}
and is omitted for brevity. We now prove a stronger version of part
1(b), namely, that cooperation is a strictly perfect equilibrium action
in any mild game (i.e., it is a perfect equilibrium action with respect
to all distributions of commitment strategies, as defined in Appendix
\ref{subsec:Strictly-Perfect}). Arguments and notations that are
analogous to the proof of Theorem \ref{thm:cooperation-defensive-PDs}
are presented in brief. Let $s^{1}$ ($s^{2}$) be the strategy that
instructs a player to defect if and only if she receives a signal
containing one or more (two or more) conflicts. Consider the following
candidate for a perfect equilibrium $\left(\left\{ s^{1},s^{2}\right\} ,\left(q,1-q\right),\theta^{*}=0\right)$.
Here, the probability $q$ will be determined such that both actions
are best replies when an agent observes a single conflict. 

Let $\left(\mathcal{S^{C}},\lambda\right)$ be a distribution of commitments.
We show that there exists a converging sequence of levels $\epsilon_{n}\rightarrow0$,
and converging sequences of steady states $\left(\left\{ s^{1},s^{2}\right\} ,\left(q_{n},1-q_{n}\right),\theta_{n}\right)\rightarrow\left(\left\{ s^{1},s^{2}\right\} ,\left(q,1-q\right),\theta\equiv0\right)$
and $\left(\left\{ s^{q_{n}}\right\} ,\theta_{n}'\right)\rightarrow\left(\left\{ \emph{\ensuremath{s^{q}}}\right\} ,\theta'\equiv0\right)$
such that either (1) each steady state $\psi{}_{n}\equiv\left(\left\{ s^{1},s^{2}\right\} ,\sigma_{n}\equiv\left(q_{n},1-q_{n}\right),\theta_{n}\right)$
is a Nash equilibrium of $\left(\left(G,k\right),\left(S^{C},\lambda\right),\epsilon_{n}\right)$,
or (2) each steady state $\psi'_{n}\equiv\left(\left\{ s^{q_{n}}\right\} ,\theta_{n}'\right)$
is a Nash equilibrium of $\left(\left(G,k\right),\left(S^{C},\lambda\right),\epsilon_{n}\right)$. 

Fix $n\geq1$. Assume that $\epsilon_{n}$ is sufficiently small.
We calculate the probability $Pr\left(m=1|S_{n}^{N}\right)$ that
a normal agent (Alice) induces a signal $m=1$. Since we focus on
the steady states in which the incumbents defect very rarely (i.e.,
$\theta_{n}$ and $\theta'_{n}$ converge to $\theta^{*}\equiv0$),
we can assume that $Pr\left(m=1|S_{n}^{N}\right)$ is $O\left(\epsilon_{n}\right)$.
(The proof of the existence  of consistent signal profiles in which
the normal agents almost always cooperate in mild PDs is analogous
to the argument presented in the proof of Theorem \ref{thm:cooperation-defensive-PDs},
and is omitted for brevity). Alice may be involved in a conflict if
one of her $k$ partners is committed, which happens with a probability
of $O\left(\epsilon_{n}\right)$. If all of the $k$ partners are
normal, then at each interaction both Alice and her partner defect
with a probability of $Pr\left(m=1|S_{n}^{N}\right)$, which implies
that the probability of a conflict is $2\cdot Pr\left(m=1|S_{n}^{N}\right)-\left(Pr\left(m=1|S_{n}^{N}\right)^{2}\right)$.
Therefore:

\[
Pr\left(m=1|S_{n}^{N}\right)=k\cdot\left(O\left(\epsilon_{n}\right)+2\cdot q_{n}\cdot Pr\left(m=1|S_{n}^{N}\right)-O\left(\left(Pr\left(m=1|S_{n}^{N}\right)\right)^{2}\right)\right).
\]
Solving this equation, while neglecting terms that are $O\left(\epsilon_{n}^{2}\right)$
(including $Pr\left(m=1|S_{n}^{N}\right)^{2}$), yields
\begin{equation}
Pr\left(m=1|S_{n}^{N}\right)=\frac{k\cdot O\left(\epsilon_{n}\right)}{1-2\cdot k\cdot q_{n}},\label{eq:1-kq-1}
\end{equation}
which is well defined and $O\left(\epsilon_{n}\right)$ as long as
$q_{n}<\frac{1}{2\cdot k}$. Note that as $q_{n}$ approaches $\frac{1}{2\cdot k}$,
the value of $Pr\left(m=1|S_{n}^{N}\right)$ ``explodes'' (becomes
arbitrarily larger than terms that are $O\left(\epsilon_{n}\right)$).

By Bayes' rule we can calculate the conditional probability $Pr\left(s|m=1\right)$
of being matched with each strategy $s\in S^{C}$ (same calculations
as detailed in the proof of Theorem \ref{thm:cooperation-defensive-PDs}).
Note that these conditional probabilities are decreasing in $Pr\left(m=1|S_{n}^{N}\right)$,
and thus decreasing in $q_{n}$. Let $\mu_{n}$ be the probability
that a random partner defects conditional on a player observing signal
$m=1$ about the partner, and conditional on the partner observing
the signal $m=0$: 
\begin{equation}
\mu_{n}=\sum_{s\in S^{C}}Pr\left(s|m=1\right)\cdot s_{0}\left(d\right)+O\left(\epsilon_{n}\right).\label{eq:mu-calculation-2-1}
\end{equation}

Note that $\mu_{n}$ is decreasing in $q_{n}$. Moreover, as $q_{n}\nearrow\frac{1}{2\cdot k}$,
we have $\mu_{n}\left(q_{n}\right)\searrow0$, because $Pr\left(m=1|S_{n}^{N}\right)$
``explodes'' as we approach the threshold of $k\cdot q=0.5$. 

Next, we calculate the value of $q_{n}$ that balances the payoffs
of both actions when a player observes a single conflict (neglecting
terms of $O\left(\epsilon_{n}\right)$). The LHS of the following
equation represents a player's direct gain from defecting when observing
a single conflict, while the RHS represents the player's indirect
loss from defecting in this case, which is induced by normal partners
who defect as a result of observing these defections. Note that the
cost is paid only if the partner cooperated, because otherwise a future
partner would observe a conflict regardless of the agent's own action.
\begin{equation}
\textrm{Pr}\left(m=1\right)\cdot\left(\mu_{n}\cdot l+\left(1-\mu_{n}\right)\cdot g\right)=\textrm{Pr}\left(m=1\right)\cdot\left(1-\mu_{n}\right)\cdot k\cdot q\cdot\left(l+1\right)+O\left(\epsilon_{n}\right)\,\,\,\Leftrightarrow\,\,\,q_{n}=\frac{\mu_{n}\cdot l+\left(1-\mu_{n}\right)\cdot g}{\left(1-\mu\right)\cdot k\cdot\left(l+1\right)}+O\left(\epsilon_{n}\right).\label{eq:mu-q-equation-conflict}
\end{equation}
In connection with Eq. (\ref{eq:mu-q-equation-conflict}) it was noted
that $q\left(\mu\right)$ is increasing in $\mu_{n}$, and since the
game is mild we have $q_{n}\left(0\right)=\frac{g}{k\cdot\left(l+1\right)}<\frac{1}{2\cdot k}$.
This implies that there is a unique pair of values of $q_{n}\in\left(\frac{g}{k\cdot\left(l+1\right)},\frac{1}{2\cdot k}\right)$
and $\mu_{n}\in\left(0,1\right)$ that jointly solve Eqs. (\ref{eq:mu-calculation-2-1})
and (\ref{eq:mu-q-equation-conflict}). This pair of values balances
the payoff of both actions when a player observes a signal $m=1$.
Note that sequences of $\left(q_{n}\right)_{n}\rightarrow q$ and
$\left(\mu_{n}\right)_{n}\rightarrow\mu$ converge to the values that
solve the above equations when one ignores the terms that are $O\left(\epsilon_{n}\right)$.
The remaining arguments of part 1 are analogous to those in the final
part of the proof of Theorem \ref{thm:cooperation-defensive-PDs},
and are omitted for brevity.

Next, we deal with Part (2), namely, the case of an acute Prisoner's
Dilemma\emph{ }($g>0.5\cdot\left(l+1\right)$). Assume (in order to
obtain a contradiction) that the environment admits a perfect equilibrium
$\left(S^{*},\sigma^{*},\theta^{*}\equiv c\right)$. That is, there
exists a converging sequence of strictly positive commitment levels
$\epsilon_{n}\rightarrow_{n\rightarrow\infty}0$, and a converging
sequence of steady states $\left(S_{n}^{N},\sigma_{n},\theta_{n}\right)\rightarrow_{n\rightarrow\infty}\left(S^{*},\sigma^{*},\theta^{*}\right)$,
such that each state $\left(S_{n}^{N},\sigma_{n},\theta_{n}\right)$
is a Nash equilibrium of the perturbed environment $\left(\left(G,k\right),\left(S^{C},\lambda\right),\epsilon_{n}\right)$.
By the arguments of part 1 (and the arguments of part 1(a) of Theorem
\ref{thm:cooperation-defensive-PDs}), the average probability $q_{n}$
by which a normal agent defects when observing $m=1$ in the steady
state $\left(S_{n}^{N},\sigma_{n},\theta_{n}\right)$ (for a sufficiently
small $\epsilon_{n}$) should be at least equal to the minimal solution
of Eq. (\ref{eq:mu-q-equation-conflict}): $q_{n}\left(\mu_{n}=0\right)=\frac{g}{k\cdot\left(l+1\right)}+O\left(\epsilon_{n}\right)$.
However, if the game is acute, then this minimal solution is larger
than $\frac{1}{2\cdot k}$, and Eq. (\ref{eq:1-kq-1}) cannot be satisfied
by $Pr\left(m=1|S_{n}^{N}\right)<<1$, which yields a contradiction.

\subsection{Proof of Theorem \ref{thm:stable-cooperation-observing-action-profiles}
(Observing Action Profiles)}

Recall that a signal $m\in M$ consists of information about the
number of times in which each of the possible four action profiles
have been played in the sampled $k$ interactions. Let $u\left(m\right)$
be the number of sampled interactions in which the partner has been
the sole defector, and let $d\left(m\right)$ denote the number of
sampled interactions in which at least of one of the players has defected.
Let $s^{1}$ and $s^{2}$ be defined as follows: 
\[
s^{1}\left(m\right)=\begin{cases}
d & u\left(m\right)=1\,\textrm{or}\,d\left(m\right)\geq2\\
c & \textrm{otherwise}
\end{cases}\,\,\,\,\,s^{2}\left(m\right)=\begin{cases}
d & d\left(m\right)\geq2\\
c & \textrm{otherwise.}
\end{cases}
\]
That is, both strategies induce agents to defect if the partner has
been involved in at least two interactions in which the outcome has
not been mutual cooperation. In addition, agents who follow $s^{1}$
defect also when observing the partner to be the sole defector in
a single interaction. 

Assume first that $G_{PD}$ is mild (i.e., $g\leq\frac{l+1}{2}$).
Fix a small probability of $0<\alpha<<\frac{1}{k}$. Let $s^{\alpha}\equiv\alpha$
be the strategy to defect with a probability of $\alpha$ regardless
of the signal. In what follows, we show that there exist a converging
sequence of commitment levels $\epsilon_{n}\rightarrow0$ and converging
sequences of steady states $\psi{}_{n}\equiv\left(\left\{ s^{1},s^{2}\right\} ,\left(q_{n},1-q_{n}\right),\theta_{n}\right)\rightarrow\left(\left\{ s^{1},s^{2}\right\} ,\left(q,1-q\right),\theta^{*}\equiv\overrightarrow{\left(c,c\right)}\right)$,
such that each steady state $\psi{}_{n}$ is a Nash equilibrium of
$\left(\left(G,k\right),\left(\left\{ s^{\alpha}\right\} ,1_{s^{\alpha}}\right),\epsilon_{n}\right)$. 
\begin{rem}
To simplify the notations below we focus on the non-regular distribution
of commitments $\left(\left\{ s^{\alpha}\right\} \right)$. Note,
however, that our arguments can be adapted in a straightforward way
to deal with the regular distribution of commitments $\left(\left\{ s^{\alpha-\delta},s^{\alpha+\delta}\right\} ,\left(\frac{1}{2},\frac{1}{2}\right)\right)$
for any $0<\delta<<\alpha$, in which the each committed agent defects
with a probability very close to $\alpha$. The same is also true
for the proof of Theorem \ref{thmtertiary-observation}below.
\end{rem}
Fix a sufficiently small $\epsilon_{n}<<1$. Let $\mu_{n}$ be the
probability that the partner defects conditional on (1) the agent
observing a single unilateral defection and $k-1$ mutual cooperations,
i.e., $\hat{m}=\left\{ \left(d,c\right),\overrightarrow{\left(c,c\right)}\right\} $
$\left(u\left(m\right)=d\left(m\right)=1\right)$, and (2) the partner
observing $k$ mutual cooperations. The parameter $q_{n}$ is defined
such that it balances the direct gain of defection (LHS of the equation)
and its indirect loss (RHS) for a normal agent who almost always cooperates:
\begin{equation}
\Pr\left(\hat{m}\right)\cdot\mu_{n}\cdot l+\left(1-\mu\right)\cdot g=\Pr\left(\hat{m}\right)\cdot\left(1-\mu_{n}\right)\cdot k\cdot q_{n}\cdot\left(l+1\right)+O\left(\epsilon_{n}\right)\,\,\,\Leftrightarrow\,\,\,q_{n}=\frac{\mu_{n}\cdot l+\left(1-\mu_{n}\right)\cdot g}{\left(1-\mu_{n}\right)\cdot k\cdot\left(l+1\right)}+O\left(\epsilon_{n}\right).\label{eq:q-mu-action-profiles}
\end{equation}
The equation is the same as in the case of observation of conflicts;
see Eq. (\ref{eq:mu-q-equation-conflict}) above. In particular, note
that the indirect cost of defection when the current partner cooperates
is only $O\left(\epsilon_{n}\right)$, because it influences only
the behavior of normal future partners if they observe an additional
interaction different from $\left(c,c\right)$ in the $k$ sampled
interactions, which happens only with a probability of $O\left(\epsilon_{n}\right)$.
Next, note that $\mu_{n}=\alpha\cdot Pr\left(s^{\alpha}|\left(d,c\right),\overrightarrow{\left(c,c\right)}\right)+O\left(\epsilon_{n}\right)$
because the only agents who follow $s^{\alpha}$ defect with positive
probability when observing $k$ mutual cooperations. Substituting
this in (\ref{eq:q-mu-action-profiles}) yields

\[
q_{n}=\frac{g+\alpha\cdot Pr\left(s^{\alpha}|\left(d,c\right),\overrightarrow{\left(c,c\right)}\right)\cdot\left(l-g\right)}{\left(1-\alpha\cdot Pr\left(s^{\alpha}|\left(d,c\right),\overrightarrow{\left(c,c\right)}\right)\right)\cdot k\cdot\left(l+1\right)}+O\left(\epsilon_{n}\right)=\frac{g}{k\cdot\left(l+1\right)}+O\left(\alpha\right)+O\left(\epsilon_{n}\right).
\]
 The mildness of the game ($g<\frac{l+1}{2}$) implies that $k\cdot q_{n}<0.5$. 

Let $p_{n}$ be the average probability with which the normal agents
defect when being matched with committed agents. When $\alpha<<\frac{1}{k}$,
the \emph{$s^{2}$}-agents rarely ($O\left(\alpha^{2}\right)$) defect
against the committed agents, because it is rare to observe these
committed agents defecting more than once. The \emph{$s^{1}$}-agents
defect against the committed agents with a probability of $k\cdot q_{n}\cdot\alpha+O\left(\alpha^{2}\right)+O\left(\epsilon_{n}\right)$
because each rare defection of the committed agents is observed with
a probability of $k\cdot q$ by $s^{1}$-agents. Since $\alpha,p_{n}<<1$,
bilateral defections are very rare ($O\left(\alpha^{2}\right)$).
This implies that $p_{n}=\alpha\cdot k\cdot q_{n}+O\left(\alpha^{2}\right)+O\left(\epsilon_{n}\right)<\frac{\alpha}{2}.$

Let $r_{n}$ be the probability that an \emph{$s^{1}$}-agent defects
against a fellow \emph{$s^{1}$}-agent. In each observed interaction,
the \emph{$s^{1}$} partner interacts with a committed (resp., \emph{$s^{1}$},
\emph{$s^{2}$}) opponent with a probability of $\epsilon_{n}$ (resp.,
$q_{n}$, 1-$q_{n}$) and the partner unilaterally defects with a
probability of $\alpha\cdot k\cdot q_{n}+O\left(\epsilon_{n}\right)+O\left(\alpha^{2}\right)$
(resp., $r_{n}+O\left(r_{n}^{2}\right)$, $O\left(\epsilon_{n}\cdot\alpha^{2}\right)$).
This implies that $r_{n}$ solves the following equation:
\[
r_{n}=k\cdot\left(\alpha\cdot q\cdot\delta_{n}+q\cdot r_{n}\right)+O\left(\epsilon_{n}^{2}\right)\,\,\Rightarrow\,\,r_{n}=\frac{\alpha\cdot k\cdot q_{n}}{1-k\cdot q_{n}}\cdot\epsilon_{n}+O\left(\epsilon_{n}^{2}+\alpha^{2}\cdot\epsilon_{n}\right)<0.5\cdot\alpha\cdot\epsilon_{n},
\]
where the latter inequality is because $k\cdot q_{n}<0.5$. The above
calculations show that the total frequency with which committed agents
unilaterally defect ($\alpha\cdot\epsilon_{n}$) is higher than the
total frequency with which normal agents unilaterally defect ($q_{n}+p_{n}\cdot\delta_{n}<\alpha\cdot\epsilon_{n}$).
This implies that the probability that an agent is committed, conditional
on his being the sole defector in an interaction, is higher than 50\%,
and that it is higher than this probability conditional on her being
the sole cooperator. Next, note that mutual defections between a committed
agent and an \emph{$s^{1}$-}agent have a frequency of $O\left(\epsilon_{n}\right)$,
while mutual defections between two committed agents (or two normal
agents) are very rare ($O\left(\epsilon_{n}^{2}\right)$), which implies
that the probability that the partner follows a committed strategy
conditional on the player observing mutual defection is 50\%$+O\left(\epsilon_{n}\right)$.
This implies that 
\[
Pr\left(s^{\alpha}|\left(d,c\right),\overrightarrow{\left(c,c\right)}\right)>\max\left(Pr\left(s^{\alpha}|\left(d,d\right),\overrightarrow{\left(c,c\right)}\right),Pr\left(s^{\alpha}|\left(c,d\right),\overrightarrow{\left(c,c\right)}\right)\right),
\]
and thus while both actions are best replies after the player observes
the signal $\left(\left(d,c\right),\overrightarrow{\left(c,c\right)}\right)$,
only cooperation is a best reply after the player observes $\left(\left(d,d\right),\overrightarrow{\left(c,c\right)}\right)$
and $\left(\left(c,d\right),\overrightarrow{\left(c,c\right)}\right)$.
Next note that conditional on a player observing a signal with at
most $k-2$ mutual cooperations, the partner is most likely to be
committed (because normal agents have two outcomes different from
mutual cooperation with a probability of only $O\left(\epsilon_{n}^{2}\right)$).
This implies that the normal agents play the unique best reply after
any signal other than $\left(\left(d,c\right),\overrightarrow{\left(c,c\right)}\right)$,
and thus any deviator who behaves differently in these cases will
be outperformed. 

Let $\chi_{n}$ be the probability that a random partner defects conditional
on both the agent and the partner observing signal $\left(\left(d,c\right),\overrightarrow{\left(c,c\right)}\right)$.
The definitions of strategies $s^{\alpha},s^{1}$, and $s^{2}$ immediately
imply that $\chi_{n}>\mu_{n}$, and analogous arguments to those presented
at the end of the proof of Theorem \ref{thm:cooperation-defensive-PDs}
show that deviators who defect with a probability strictly between
zero and one after observing $\left(\left(d,c\right),\overrightarrow{\left(c,c\right)}\right)$
are outperformed (because an agent's payoff is a strictly convex function
of the agent's defection probability when observing signal $\left(\left(d,c\right),\overrightarrow{\left(c,c\right)}\right)$).

Next assume that the $G_{PD}$ is acute. We have to show that cooperation
is not a perfect equilibrium action. Assume to the contrary that $\left(S^{*},\sigma^{*},\theta^{*}\equiv0\right)$
is a perfect equilibrium with respect to distribution of commitments
$\left(\mathcal{S^{C}},\lambda\right)$. Let $\psi_{n}=\left(S_{n}^{N},\sigma_{n},\theta_{n}\right)\rightarrow\left(S^{*},\sigma^{*},0\right)$
be a converging sequence of Nash equilibria in the converging sequence
of perturbed environments $\left(\left(G_{PD},k\right),\left(\mathcal{S^{C}},\lambda\right),\epsilon_{n}\right)$.
Analogous arguments to the proof of part 1(a) of Theorem \ref{thm:cooperation-defensive-PDs}
show that any perfect equilibrium that implements full cooperation
$\left(S^{*},\sigma^{*},\theta^{*}\equiv0\right)$ must satisfy (1)
$s_{\overrightarrow{\left(c,c\right)}}=c$ for each $s\in S^{*}$,
(2) if $d\left(m\right)\geq2$ then $s_{m}=d$ for each $s\in S^{*}$,
and (3) there are $s,s'\in S^{*}$ such that $s_{\left\{ \left(d,c\right),\overrightarrow{\left(c,c\right)}\right\} }\left(d\right)>0$
and $s_{\left\{ \left(d,c\right),\overrightarrow{\left(c,c\right)}\right\} }\left(d\right)<1$. 

Let $0<q_{n}<1$ be the average probability according to which a normal
agent defects when she observes $\left\{ \left(d,c\right),\overrightarrow{\left(c,c\right)}\right\} $.
By analogous arguments to those presented above (see Eq. (\ref{eq:q-mu-action-profiles})),
$q_{n}$ is an increasing function of $\mu_{n}$, and $q_{n}\left(\mu_{n}=0\right)=\frac{g}{k\cdot\left(l+1\right)}$.
The acuteness of the game implies that $k\cdot q_{n}>\frac{g}{\left(l+1\right)}>\frac{1}{2}$. 

Let $s_{\beta}\in S^{C}$ be a committed strategy that induces an
agent who follows it (called an $s_{\beta}$\emph{-agent}) to defect
with a probability of $\beta>0$ when he observes $\left(\overrightarrow{\left(c,c\right)}\right)$.
In what follows, we show that the presence of strategy $s_{\beta}$
induces the normal agents to unilaterally defect more often than $s_{\beta}$-agents.
Let $p_{n}$ be the average probability that normal agents defect
against $s^{\alpha}$-agents in state $\psi_{n}$. This probability
$p_{n}$ must solve the following inequality:
\begin{eqnarray}
1-p_{n} & \geq & \left(\left(1-\beta\right)\cdot\left(1-p_{n}\right)\right)^{k}+k\cdot\left(\left(1-\beta\right)\cdot\left(1-p_{n}\right)\right)^{k-1}\cdot\left(1-\left(1-\beta\right)\cdot\left(1-p_{n}\right)\right)\label{eq:1-p}\\
 &  & +\left(1-q_{n}\right)\cdot k\cdot\left(\left(1-\beta\right)\cdot\left(1-p_{n}\right)\right)^{k-1}\cdot\beta\cdot\left(1-p_{n}\right)+O\left(\epsilon_{n}\right).\nonumber 
\end{eqnarray}
The LHS of Eq. (\ref{eq:1-p}) is the average probability that normal
agents cooperate against $s_{\beta}$-agents (recall that normal agents
always defect when they observe at most $k-2$ mutual cooperations).
The normal agents cooperate with probability one (resp., at most one,
$q_{n}$) if they observe $\left(\overrightarrow{\left(c,c\right)}\right)$
(resp., $\left(\left(d,d\right),\overrightarrow{\left(c,c\right)}\right)$
or $\left(\left(c,d\right),\overrightarrow{\left(c,c\right)}\right)$
, $\left(\left(d,c\right),\overrightarrow{\left(c,c\right)}\right)$),
which happens with a probability of $\left(\left(1-\beta\right)\cdot\left(1-p_{n}\right)\right)^{k}$
(resp.,\\
 $k\cdot\left(\left(1-\beta\right)\cdot\left(1-p_{n}\right)\right)^{k-1}\cdot\left(1-\left(1-\beta\right)\cdot\left(1-p_{n}\right)\right)$,
$k\cdot\left(\left(1-\alpha\right)\cdot\left(1-p_{n}\right)\right)^{k-1}\cdot\alpha\cdot\left(1-p\right)$). 

Direct numerical analysis of Eq. (\ref{eq:1-p}) shows that the minimal
$p_{n}$ that solves this inequality (given that $q_{n}>\frac{1}{2\cdot k}$)
is greater than $\frac{\beta}{2-\beta}$ for any $\beta\in\left(0,1\right)$.
The total frequency of interactions in which the $s_{\beta}$-agents
unilaterally defect is $\beta\cdot\left(1-p_{n}\right)\cdot\epsilon_{n}\cdot\lambda\left(s_{\beta}\right)+O\left(\epsilon_{n}^{2}\right)$.
The total frequency of interactions in which normal agents unilaterally
defect against the $s_{\beta}$-agents is $p_{n}\cdot\left(1-\beta\right)\cdot\epsilon_{n}\cdot\lambda\left(s_{\beta}\right)+O\left(\epsilon_{n}^{2}\right)$.
Eq. (\ref{eq:q-mu-action-profiles}) shows that these unilateral defections
against $s_{\beta}$-agents induce the normal agents to unilaterally
defect among themselves with a total frequency of $\frac{p_{n}\cdot\left(1-\beta\right)\cdot\epsilon_{n}\cdot\lambda\left(s_{\beta}\right)}{1-k\cdot q_{n}}+O\left(\epsilon_{n}^{2}\right)>p_{n}\cdot\left(1-\beta\right)\cdot\epsilon_{n}\cdot\lambda\left(s_{\beta}\right)$.
Finally, note that $p_{n}>\frac{\beta}{2-\beta}$$\Leftrightarrow$
$2\cdot p_{n}\cdot\left(1-\beta\right)>\beta\cdot\left(1-p_{n}\right)$
implies that normal agents unilaterally defect (as the indirect result
of the presence of the $s_{\beta}$-agents) more often than $s_{\beta}$-agents. 

Next, observe that bilateral defections are most likely to occur in
interactions between normal and committed agents. This is because
the probability that both normal agents defect against each other
is only $O\left(\epsilon_{n}^{2}\right)$. Thus, when a player observes
bilateral defection the partner is more likely to be a committed agent
than when the player observes a unilateral defection by the partner.
This implies that all the normal agents defect with probability one
when they observe $\left(\left(d,d\right),\overrightarrow{\left(c,c\right)}\right)$
because in this case defection is the unique best reply. 

Let $w_{n}$ be the (average) probability that normal agents defect
when they observe $\left(\left(c,d\right),\overrightarrow{\left(c,c\right)}\right)$.
If $w_{n}<0.5$, then cooperation is the unique best reply for a normal
agent who faces a partner who is likely to defect (e.g., when the
normal agent observes fewer than $k-1$ mutual cooperations), and
so we get a contradiction. This is because defecting against a defector
yields a direct gain of $l$ and an indirect loss of at least $0.5\cdot k\cdot\left(l+1\right)\geq l+1>l$
(because this bilateral defection will be observed on average $k$
times, and in at least half of these cases it will induce the partner
to defect, whereas if the agent were cooperating, then he would have
induced the partner to cooperate).

Thus, $w_{n}\geq0.5$$\Rightarrow k\cdot w_{n}>1$. However, in this
case, an analogous argument to the one at the end of the proof of
Theorem \ref{thm:stable-cooperation-observing-conflicts} implies
that an arbitrarily small group of mutants who defect with small probability
will cause the incumbents to unilaterally defect with high probability,
and thus no focal post-entry population exists, which contradicts
the assumption that cooperation is neutrally stable.

\subsection{Proof of Theorem \ref{thmtertiary-observation} (Observing Actions
against Cooperation)}

The construction of the distribution of commitments $\left(\left\{ s^{\alpha}\right\} \right)$
(or the regular distribution of commitments $\left(\left\{ s^{\alpha-\delta},s^{\alpha+\delta}\right\} ,\left(\frac{1}{2},\frac{1}{2}\right)\right)$
for $0<\delta<<\alpha$) and of the perfect equilibrium $\left(\left\{ s^{1},s^{2}\right\} ,\left(q,1-q\right),\theta\equiv\overrightarrow{\left(c,c\right)}\right)$
and most of the arguments are the same as in the proof of Theorem
\ref{thm:stable-cooperation-observing-action-profiles}, and are omitted
for brevity. Fix $\epsilon_{n}$ sufficiently small. By the same arguments
as in the proof of Theorem of \ref{thm:stable-cooperation-observing-conflicts},
the value of $q_{n}$ that balances the payoffs of $s^{1}$ and $s^{2}$
satisfy $k\cdot q_{n}<1$ for any underlying Prisoner's Dilemma. 

Recall that $p_{n}$, the average probability with which the normal
agents defect when being matched with committed agents, satisfies
$p_{n}=\alpha\cdot k\cdot q_{n}+O\left(\alpha^{2}\right)+O\left(\epsilon_{n}\right)<\alpha.$
This implies that the probability that an agent is committed, conditional
on her being the sole defector in an interaction, is higher than 50\%,
conditional on her being the sole cooperator. Next, observe that $\alpha<<1$
implies that the probability $Pr\left(\left(d,d\right)\right)=O\left(p_{n}\cdot\alpha^{2}\right)\cdot O\left(\epsilon_{n}\right)<<Pr\left(\left(c,d\right)\right)=O\left(p_{n}\cdot\epsilon_{n}\cdot\alpha\right)$,
which implies that conditional on an agent observing the signal $\left\{ \left(*,d\right),\left(\overrightarrow{\left(c,c\right)}\right)\right\} $,
it is most likely that the partner has cooperated rather than defected
in the interaction in which $\left(*,d\right)$ has been observed.
This implies that $\Pr\left(s^{\alpha}|\left\{ \left(*,d\right),\left(\overrightarrow{\left(c,c\right)}\right)\right\} \right)<\Pr\left(s^{\alpha}|\left\{ \left(d,c\right),\left(\overrightarrow{\left(c,c\right)}\right)\right\} \right)$,
and given the value of $q_{n}$ for which both actions are best replies
conditional on observing signal $\left\{ \left(d,c\right),\left(\overrightarrow{\left(c,c\right)}\right)\right\} $,
cooperation is the unique best reply when observing either $\left\{ \left(*,d\right),\left(\overrightarrow{\left(c,c\right)}\right)\right\} $
or $\left\{ \left(\overrightarrow{\left(c,c\right)}\right)\right\} $,
while defection is the unique best reply when observing at most $k-2$
mutual cooperations. This implies that $\left(\left\{ s^{1},s^{2}\right\} ,\left(q,1-q\right),\theta\equiv0\right)$
is a perfect equilibrium (where $q$ is the limit of $q_{n}$ when
$\epsilon_{n}$ converges to zero.

\subsection{Proof of Theorem \ref{thmLrepeated-game} (Repeated Game)\label{subsec:Technical-Details-in-prof-repeated}}

\textbf{Part 1}: Assume that $g>l$ (i.e., an offensive game). Assume
to the contrary that there exist a sequence of Nash equilibria of
perturbed environments that converge to a perfect equilibrium that
induces full cooperation. The fact that the perfect equilibrium induces
full cooperation implies that in any sufficiently close Nash equilibrium
(i.e., for a sufficiently large $n)$:
\begin{enumerate}
\item normal agents cooperate with high probability when observing $k$
acts of cooperation;
\item when an agent is matched with a normal partner, the agent most of
the time observes $k$ acts of cooperation;
\item when a normal agent observes $k$ acts of cooperation, the partner
is most likely normal and he is going to cooperate with a probability
close to one;
\item when an agent observes $k$ acts of defection, the partner has a positive
(and non-negligible) probability of being a committed agent and defecting
in the current match. 
\end{enumerate}
In order for these facts to be be consistent with equilibrium it must
be the case that cooperation is a best reply against a partner who
is most likely to cooperate in the current match; i.e., the direct
gain from defecting, which is very close to $g$, has to be lower
than the future indirect loss, which is independent of the partner's
action. The inequality $g>l$ then implies that cooperation is the
unique best reply against a partner who is going to cooperate with
an expected probability that is not close to 1 (because the direct
gain from defecting is a mixed average of $l$ and $g$, which is
less than $g$). This, in turn, implies that all normal agents cooperate
with a probability of one when they observe $k$ acts of defection
(because, given such a signal, the partner has a positive and non-negligible
probability of being a committed agent and defecting in the current
match). Hence, a deviator who always defects outperforms the incumbent,
since she induces normal agents to cooperate against her, and obtains
the high payoff of $1+g$ in most rounds of the repeated game.

\textbf{Part 2:} Assume that $g\leq l$, $k\geq2$, and $\delta>\frac{l}{l+1}$.
Let $\gamma=\frac{l}{\delta\cdot\left(l+1\right)}\in\left(0,1\right)$.
Let $0<\underline{\alpha}<\overline{\alpha}<1$ be two probabilities
satisfying the condition that the ratio $\nicefrac{\overline{\alpha}}{\underline{\alpha}}$
is sufficiently large (as further specified below). Consider a homogeneous
group of committed agents who follow the following strategy $s^{C}$:
\begin{enumerate}
\item defect with probability $\overline{\alpha}$ if they either (1) defected
in the last round, or (2) defected at least twice in the last $k-1$
rounds; and
\item defect with probability $\underline{\alpha}$ otherwise.
\end{enumerate}
Consider the perturbed environment $\left(\left(G,k,\delta\right),\left(\left\{ s^{C}\right\} \right),\epsilon_{n}\right)$,
for a sufficiently small $\epsilon_{\text{n}}>0$. Consider a homogeneous
population of normal agents who play according to the following strategy
$s^{N}$: 
\begin{enumerate}
\item cooperate if the agent defected in any of the last $\min\left(t,k-1\right)$
rounds;
\item otherwise (i.e., the agent cooperated in all of the last $\min\left(t,k-1\right)$
rounds):
\begin{enumerate}
\item cooperate if the partner has never defected in the last $\min\left(t,k-1\right)$
rounds;
\item defect if the partner defected at least twice in the last $\min\left(t,k-1\right)$
rounds;
\item cooperate if the partner defected only once in the last $\min\left(t,k-1\right)$
rounds and did not defect in the last round; and
\item defect with probability $q_{t}$ if the partner defected only in the
last round, where $t$ is the current round, and the sequence $\left(q_{t}\right)_{t\geq1}$
is defined recursively below.
\end{enumerate}
\end{enumerate}
Let $q_{1}=\gamma$. The value of each $q_{t}$ for $t\geq2$ is determined
such that a normal agent is indifferent between defecting and cooperating
in round $t-1$ conditional on the events that (1) the agent did not
defect in any of the previous $k-1$ rounds, and (2) the agent observes
the signal $\left(c,...,c,d\right)$ (i.e., the partner defected in
the last round and cooperated in all of the previous observed interactions).
Here we are relying on the one-deviation principle; in the next period
the agent will have a track record $\left(c,c,c,...c,d\right)$, which
means that the agent should cooperate.

The gain from defecting in round $t-1$ is equal to $l\cdot\mu_{t-1}+g\cdot\left(1-\mu_{t-1}\right)$,
where $\mu_{t-1}$ is the probability that a random partner defects
conditional on the union of the two events above. Such a defection
induces an expected loss of $\delta\cdot\left(l+1\right)\cdot q_{t}+O\left(\epsilon\right)$
for the agent in the next round (with a probability of $\left(1-\epsilon\right)$
the partner in the next round is normal, and in this case he will
defect with probability $q_{t}$ instead of cooperating, which will
induce a loss of $\delta\left(l+1\right)$ for the agent. In the round
after that the agent will have a track record $\left(c,c,c,...c,d,c\right)$
which means that the agent should cooperate again. The partner, if
normal, will cooperate for sure with the agent. Thus, an agent is
indifferent between the two actions in round $t-1$ when observing
$\left(c,...,c,d\right)$ iff
\[
l\cdot\mu_{t-1}+g\cdot\left(1-\mu_{t-1}\right)=\delta\cdot\left(l+1\right)\cdot q_{t}+O\left(\epsilon\right)\Leftrightarrow q_{t}=\frac{l\cdot\mu_{t-1}+g\cdot\left(1-\mu_{t-1}\right)}{\delta\cdot\left(l+1\right)}+O\left(\epsilon\right).
\]

Observe that the $q_{t}$'s have a uniform bound strictly below one,
i.e., $\forall t\in\mathbb{N}$, $0<q_{t}<\gamma=\frac{l}{\delta\cdot\left(l+1\right)}<1.$
Let $\left(\beta_{t}\right)_{t\in\mathbb{N}}$ be the average probability
with which normal agents defect in round $t$. Observe that $\beta_{1}=0$,
and that $\beta_{t}$ can be bounded as follows for any $t\in\mathbb{N}$:
\[
\beta_{t}\leq q_{t}\cdot\beta_{t-1}+O\left(\epsilon\right)<\gamma\cdot\beta_{t-1}+O\left(\epsilon\right).
\]
This implies that $\beta_{t}$ is bounded from above by a converging
geometric sequence, and, thus, $\beta_{t}<\frac{O\left(\epsilon\right)}{1-\gamma}$
for each $t$. This implies that the population state $\left(s_{N},1_{s_{N}}\right)$
induces full cooperation in the limit $\epsilon\rightarrow0$. 

Let $Pr\left(S^{C}|\left(c,...,c,d\right),t\right)$ be the probability
that the partner is committed conditional on the agent observing signal
$\left(c,...,c,d\right)$ in round $t$. Let $Pr\left(\left(c,...,c,d\right),t|S^{C}\right)$
($Pr\left(\left(c,...,c,d\right),t|s_{N}\right)$) be the probability
that an agent observes the signal $\left(c,...,c,d\right)$ in round
$t$ conditional on the partner being committed (normal). Observe
that $Pr\left(\left(c,...,c,d\right),t|S^{C}\right)>\underline{\alpha}^{k}$
(because a committed agent plays each pure action with a probability
of at least $\underline{\alpha}$ in each round), and that $Pr\left(\left(c,...,c,d\right),t|S^{C}\right)<sup_{t}\beta_{t}<\frac{O\left(\epsilon\right)}{1-\gamma}$
(because the average probability in which a normal agent defects is
at most $sup_{t}\beta_{t}$. By using Bayes' rule we can give a uniform
minimal bound to $Pr\left(S^{C}|\left(c,...,c,d\right),t\right)$
as follows: 
\[
\frac{\epsilon\cdot\underline{\alpha}^{k}}{\epsilon\cdot\underline{\alpha}^{k}+\frac{O\left(\epsilon\right)}{1-\gamma}}<Pr\left(S^{C}|\left(c,...,c,d\right),t\right)<1.
\]

We assume that the ratio $\nicefrac{\overline{\alpha}}{\underline{\alpha}}$
is sufficiently large such that $\overline{\alpha}\cdot Pr\left(S^{C}|\left(c,...,c,d\right),t\right)>\underline{\alpha}$
in each round $t$. Recall the definition from above and then observe
that $\mu_{t-1}=Pr\left(S^{C}|\left(c,...,c,d\right),t\right)\cdot\bar{\alpha}\in\left(\underline{\alpha},\bar{\alpha}\right)$
for each round $t$. Recall that the probabilities $\left(q_{t}\right)_{t\in\mathbb{N}}$
have been defined such that each normal agent is indifferent between
the two actions when (1) she observes the signal $\left(c,...,c,d\right)$,
and (2) she did not defect in any of the previous $k-1$ rounds. Next
we show that the normal agents have strict preferences in all other
cases. Specifically, the fact that $g\leq l$ (resp., $g<l$) implies
that each normal agent:
\begin{enumerate}
\item strictly prefers to cooperate if she defected exactly once in the
last $k-1$ rounds. This is so because if the agent defects in the
current round it induces any normal opponent in the next round to
defect for sure (instead of cooperating). This implies that defection
in the current round induces an indirect loss of at least $\delta\cdot\left(l+1\right)>l>g$,
which is larger than the agent's direct gain from defection.\footnote{Private histories in which a normal agent has defected more than once
in the last $k-1$ rounds never happen on the equilibrium path. If
one wishes to turn the above-described equilibrium into a sequential
equilibrium (where agents best reply also off the equilibrium path),
then one needs to make a stronger assumption on $\delta$, namely,
that $\delta^{k-1}>\frac{l}{l+1}$. This is so because after the off-equilibrium
history in which an agent has defected in all the last $k-1$ rounds,
an additional defection in the current round $t$ induces a future
normal partner to defect instead of cooperating only in round $t+k-1$
(because normal partners will defect in rounds $t+1$,~...~,$t+k-2,$
regardless of the agent's behavior in round $t$).} 
\item weakly (resp., strictly) prefers to cooperate if she observes the
signal $\left(c,...,c\right)$; in this case, the partner is most
likely a normal agent who is going to cooperate, and the direct gain
from defecting ($g)$ is outweighed by the larger indirect loss in
the next round ($\delta\cdot\left(l+1\right)\cdot q_{t}+O\left(\epsilon\right)>g$).
\item weakly (resp., strictly) prefers to defect if (1) the partner defected
at least twice in the last $k$ rounds, and (2) the agent did not
defect in any of the last $k-1$ rounds; in this case the partner
is most likely to be a committed agent and to defect with a high probability
of $\overline{\alpha}>\mu_{t-1}$ in each round $t$ and, thus, defection
is the agent's unique best reply.
\item weakly (resp., strictly) prefers to cooperate if the partner defected
only once in the last $k$ rounds, and this defection did not happen
in the last round; in this case the probability that the partner is
going to defect in the current match is at most $\underline{\alpha}<\mu_{t-1}$
for each round $t$ and, thus, cooperation is the agent's unique best
reply. 
\end{enumerate}
This implies that the population state $\left(\left\{ s^{N}\right\} \right)$
is indeed a Nash equilibrium of the perturbed environment for a sufficiently
small $\epsilon$. 

\section{Cheap Talk and Equilibrium Selection (Online Publication)\label{subsec:Cheap-Talk}}

Appendix \ref{sec:Evolutionary-Stability} shows that both perfect
equilibrium outcomes, namely, cooperation and defection, satisfy the
refinement of evolutionary stability. In this section we discuss how
the stability analysis changes if one introduces pre-play ``cheap-talk''
communication in our setup.

For concreteness, we focus on observation of actions. As in the standard
setup of normal-form games (without observation of past actions),
the introduction of cheap talk induces different equilibrium selection
results, depending on whether or not deviators have unused signals
to use as secret handshakes (see, e.g., \citealp{robson1990efficiency,Schlag_1993_cheapWP,kim1995evolutionary}).
If one assumes that the set of cheap-talk signals is finite, and all
signals are costless, then cheap talk has little effect on the set
of perfect equilibrium outcomes (as any perfect equilibrium of the
game without cheap talk can be implemented as an equilibrium with
cheap talk in which the incumbents send all signals with positive
probability). 

In what follows we focus on a different case, in which there are slightly
costly signals that, due to their positive cost, are not used unless
they yield a benefit. In this setup our results should be adapted
as follows.
\begin{enumerate}
\item Offensive games: No stable state exists. Both defection and cooperation
are only ``quasi-stable''; the population state occasionally changes
between theses two states, based on the occurrence of rare random
experimentations. The argument is adapted from \citet{wiseman2001cooperation}.
\item Defensive games (and $k\geq2$): The introduction of cheap talk destabilizes
all non-efficient equilibria, leaving cooperation as the unique stable
outcome. The argument is adapted from \citet{robson1990efficiency}.
\end{enumerate}
In what follows we only briefly sketch the arguments for these results,
since a formal presentation would be very lengthy, and the contribution
is somewhat limited given that similar arguments have already been
presented in the literature.

Following \citet{wiseman2001cooperation}, we modify the environment
by endowing agents with the ability to send a slightly costly signal
$\phi$ (called the \emph{secret handshake}). An agent has to pay
a small cost $c$ either to send $\phi$ to her partner or to observe
whether the partner has sent $\phi$ to her. In addition, we still
assume that each agent observes $k\geq2$ past actions of the partner.
Let $\xi$ be the initial small frequency of a group of experimenting
agents (called\emph{ mutants}) who deviate jointly. We assume that
$O\left(\epsilon\right)\cdot O\left(\xi\right)<c<O\left(\xi\right)$,
i.e., that the small cost of the secret handshake is smaller than
the initial share of mutants, but larger than the product of the two
small shares of the mutants ($O\left(\xi\right)$) and the committed
agents ($O\left(\epsilon\right)$). To simplify the analysis we also
assume that the committed agents do not use the secret handshake

Consider a population that starts at the defection equilibrium, in
which all normal agents defect regardless of the observed actions
and do not use signal $\phi$. Consider a small group of $\xi$ mutants
(``cooperative handshakers'') who send the signal $\phi$, and cooperate
iff the partner has sent $\phi$ as well. These mutants outperform
the incumbents: they achieve $\xi$ additional points by cooperating
among themselves, which outweighs the cost of $2\cdot c$ for using
the secret handshake. Thus, assuming a payoff-monotonic selection
dynamics, the mutants take over the population and destabilize the
defective equilibrium. If the underlying game is offensive, then there
is no other candidate to be a stable population state. Thus, cooperation
can be sustained only until new mutants arrive (``defective handshakers'')
who use the secret handshake and always defect. These mutants outperform
the cooperative handshakers, and would take over the population. Finally,
a third group of mutants who always defect without using the secret
handshake can take the population back to the starting point. 

If the underlying game is defensive, then there is a sequence of mutants
who can take the population into the cooperative equilibrium characterized
in the main text. Specifically, the second group of mutants (the ones
after the cooperative handshakers) include agents who send only $\phi$,
but instead of incurring the small cost $c$ of observing the partner's
secret handshake, they base their behavior on the partner's observed
actions, namely, they play some combination of the strategies $s^{1}$
and $s^{2}$. This second group of mutants would take over the population
because the cost they save by not checking the secret handshake outweighs
the small loss of $O\left(\epsilon\right)$ incurred from not defecting
against committed partners. Finally, a third group of mutants who
do not send the secret handshake, and follow strategies $s^{1}$ and
$s^{2}$, can take over the population (by saving the cost of sending
$\phi$), and induce the perfect cooperative equilibrium of the main
text. This equilibrium remains stable also with the option of using
the secret handshake because (1) mutants who defect when observing
$m=0$ are outperformed due to similar arguments to those in the main
model, and (2) mutants who send the secret handshake, and always cooperate
when observing $\phi$ (also when $m>2$), are outperformed, as the
cost of the secret handshake $c$ outweighs the gain of $O\left(\xi\right)\cdot O\left(\epsilon\right)$.

\section{{\large{}Example: Equilibrium with Partial Cooperation (Online Publication)}\label{sec:Example-non-regular-perfect-partial-cooperation}}

The following example demonstrates the existence of a \emph{non-regular}
perfect equilibrium of an offensive Prisoner's Dilemma, in which players
cooperate with positive probability.
\begin{example}[Non-regular Perfect Equilibrium with Partial Cooperation]
\label{exa:non-regular-perfect-equilbirium-partial-ccoperation}Consider
the environment $\left(G_{O},1\right)$ where $G_{O}$ is an offensive
Prisoner's Dilemma game with $g=2.3,\,l=1.7$ (see Table \ref{tab:Prisoner-Dilemma-1-1}),
and each agent observes a single action sampled from the partner's
behavior. Let $s^{*}$ be the strategy that defects with probability
10\% after observing cooperation (i.e., $m=0)$ and defects with probability
$81.7\%$ (numerical values in this example are rounded to 0.1\%)
after observing a defection (i.e., $m=1$). Let $q^{*}$ denote the
average probability of defection in a homogeneous population of agents
who follow strategy $s^{*}$. The value of $q^{*}$ is calculated
as follows:
\begin{equation}
q^{*}=\left(1-q^{*}\right)\cdot10\%+q^{*}\cdot81.7\%\,\,\Rightarrow\,\,q^{*}=35.3\%.\label{eq:q-alpha-beta}
\end{equation}
Eq. (\ref{eq:q-alpha-beta}) holds because an agent defects in either
of the following exhaustive cases: (1) she observes cooperation (which
happens with a probability of $1-q^{*}$) and then she defects with
probability 10\%, or (2) she observes defection (which happens with
a probability of $q^{*}$) and then she defects with probability 81.7\%.
This implies that the unique consistent signal $\theta^{*}$ of a
homogeneous population in which all agents follow $s^{*}$ satisfies
$\theta^{*}\left(1\right)=35.3\%$ (i.e., agents defect in $35.3\%$
of the observed interactions).\\
Next, observe that an agent who follows strategy $s^{*}$ defects
with probability 
\[
p\left(q\right)=q\cdot81.7\%+\left(1-q\right)\cdot10\%
\]
 when being matched with a partner who defects with an average probability
of $q$. This implies that the payoff of a deviator (Alice) who defects
with an average probability of $q$ is
\[
\pi_{q}\left(\left(\left\{ s^{*}\right\} ,1_{s^{*}},\theta^{*}\right)\right)=q\cdot\left(1-p\left(q\right)\right)\cdot\left(1+g\right)+\left(1-q\right)\cdot p\left(q\right)\cdot\left(-l\right)+\left(1-q\right)\cdot\left(1-p\left(q\right)\right)\cdot1.
\]
This is because with a probability of $q\cdot\left(1-p\left(q\right)\right)$
only Alice defects, with a probability of $\left(1-q\right)\cdot p\left(q\right)$
only Alice cooperates, and with a probability of $\left(1-q\right)\cdot\left(1-p\left(q\right)\right)$
both players cooperate. By calculating the FOC one can show that $q=q^{*}=35.3\%$
is the probability of defection that uniquely maximizes the payoff
of a deviator. This implies that $\left(\left\{ s^{*}\right\} ,1_{s^{*}},35.3\%\right)$
is a Nash equilibrium of the (non-regular) perturbed environments
$\left(G,k,\left\{ s^{*}\right\} ,1_{s^{*}},\epsilon\right)$ for
any $\epsilon\in\left(0,q\right)$, which implies that $\left(\left\{ s^{*}\right\} ,1_{s^{*}},\theta^{*}\right)$
is a (non-regular) perfect equilibrium. \\
The above perfect equilibrium relies on a very particular set of commitment
strategies in which all committed agents happen to play the same strategy
as the normal agents. This cannot hold in a regular set of commitment
strategies, in which different commitment strategies defect with different
average probabilities. Given this regularity, it must be the case
that the conditional probability that the partner is going to defect
is higher after he observes a defection ($m=1$) than after he observes
a cooperation ($m=0$). This implies that a deviator (Alice) who defects
with a probability of $35.3\%$ regardless of the signal will strictly
outperform the incumbents. This is because the incumbents behave the
same against Alice (as she has the same average probability of defection
as the incumbents), while Alice defects with higher probability against
partners who are more likely to cooperate (i.e., after she observes
$m=0$), which implies that due to the offensiveness of the game (i.e.,
$g>l$), Alice achieves a strictly higher payoff than the incumbents.
\end{example}

\end{document}